\documentclass[sigconf, nonacm]{acmart}

\AtBeginDocument{%
  \providecommand\BibTeX{{%
    \normalfont B\kern-0.5em{\scshape i\kern-0.25em b}\kern-0.8em\TeX}}}

\newcommand\vldbdoi{XX.XX/XXX.XX}
\newcommand\vldbpages{XXX-XXX}
\newcommand\vldbvolume{14}
\newcommand\vldbissue{1}
\newcommand\vldbyear{2023}
\newcommand\vldbauthors{\authors}
\newcommand\vldbtitle{\shorttitle}
\newcommand\vldbavailabilityurl{https://github.com/toggled/vldbsubmission}
\newcommand\vldbpagestyle{plain}

\usepackage{cleveref}
\usepackage{mathtools}
\usepackage{algorithm}
\usepackage{subfig}
\usepackage{multirow}
\usepackage{dsfont}
\usepackage{varwidth}
\usepackage[noend]{algpseudocode}
\algrenewcommand\algorithmicrequire{\textbf{Input:}}
\algrenewcommand\algorithmicensure{\textbf{Output:}}

\DeclarePairedDelimiter\abs{\lvert}{\rvert}%
\DeclareMathSymbol{\mlq}{\mathord}{operators}{``}
\DeclareMathSymbol{\mrq}{\mathord}{operators}{`'}

\usepackage{graphicx}
\usepackage{booktabs}
\usepackage{amsmath}
\usepackage{array}
\usepackage{amsfonts}
\usepackage{epsfig}
\usepackage{graphicx}
\usepackage{url}
\usepackage[bottom]{footmisc}
\usepackage{balance}
\usepackage{multirow}
\usepackage{xspace}
\usepackage{tabularx}
\usepackage{pifont}% http://ctan.org/pkg/pifont
\usepackage{url}
\newtheorem{theor}{Theorem}

\newtheorem{lem}{Lemma}

\newtheorem{defn}{Definition}

\newtheorem{exam}{Example}

\newcommand{\spara}[1]{\smallskip\noindent{\bf #1}}

\newcommand{\squishlist}{
 \begin{list}{$\bullet$}
  {  \setlength{\itemsep}{0pt}
     \setlength{\parsep}{3pt}
     \setlength{\topsep}{3pt}
     \setlength{\partopsep}{0pt}
     \setlength{\leftmargin}{2em}
     \setlength{\labelwidth}{1.5em}
     \setlength{\labelsep}{0.5em}
} }
\newcommand{\squishlisttight}{
 \begin{list}{$\bullet$}
  { \setlength{\itemsep}{0pt}
    \setlength{\parsep}{0pt}
    \setlength{\topsep}{0pt}
    \setlength{\partopsep}{0pt}
    \setlength{\leftmargin}{2em}
    \setlength{\labelwidth}{1.5em}
    \setlength{\labelsep}{0.5em}
} }

\newcommand{\squishdesc}{
 \begin{list}{}
  {  \setlength{\itemsep}{0pt}
     \setlength{\parsep}{3pt}
     \setlength{\topsep}{3pt}
     \setlength{\partopsep}{0pt}
     \setlength{\leftmargin}{1em}
     \setlength{\labelwidth}{1.5em}
     \setlength{\labelsep}{0.5em}
} }

\newcommand{\squishend}{
  \end{list}
}

\newcommand{\eat}[1]{}

\newcounter{ccc}

\DeclareMathOperator*{\argmin}{arg\,min}
\DeclareMathOperator*{\argmax}{arg\,max}
\newcommand{\bigO}{\mathcal{O}}

\newcommand{\naheed}[1]{{\leavevmode\color{black}#1}}

\algblock{ParFor}{EndParFor}
\algnewcommand\algorithmicparfor{\textbf{parallel for}}
\algnewcommand\algorithmicpardo{\textbf{do}}
\algnewcommand\algorithmicendparfor{\textbf{end\ parallel for}}
\algrenewtext{ParFor}[1]{\algorithmicparfor\ #1\ \algorithmicpardo}
\algrenewtext{EndParFor}{\algorithmicendparfor}

\begin{document}
\title{Neighborhood-based Hypergraph Core Decomposition}

\author{Naheed Anjum Arafat}
\affiliation{%
  \institution{National University of Singapore}
  \country{Singapore}
}
\email{naheed_anjum@u.nus.edu}

\author{Arijit Khan}
\affiliation{%
  \institution{Aalborg University}
  \country{Denmark}
}
\email{arijitk@cs.aau.dk}

\author{Arpit Kumar Rai}
\affiliation{%
  \institution{Indian Institute of Technology, Kanpur}
  \country{India}
}
\email{arpitkr20@iitk.ac.in}

\author{Bishwamittra Ghosh}
\affiliation{%
  \institution{National University of Singapore}
  \country{Singapore}
}
\email{bghosh@u.nus.edu}

\begin{abstract}
We propose \emph{neighborhood-based core decomposition}: a novel way of decomposing hypergraphs into hierarchical neighborhood-cohesive subhypergraphs.
Alternative approaches to decomposing hypergraphs, e.g., reduction to clique or bipartite graphs, are not meaningful in certain applications, the later also results in inefficient decomposition; while existing degree-based hypergraph decomposition does not distinguish nodes with different neighborhood sizes.
Our case studies show that the proposed decomposition is more effective than degree and clique graph-based decompositions in %diffusion problems, e.g., 
disease intervention and in extracting provably approximate and application-wise meaningful densest subhypergraphs.
%As technical contributions, 
We propose three algorithms: \textbf{Peel}, its efficient variant \textbf{E-Peel}, and a novel local algorithm:
\textbf{Local-core} with parallel implementation.
Our most efficient parallel algorithm \textbf{Local-core(P)} decomposes
hypergraph with 27M nodes and 17M hyperedges in-memory within 91 seconds
by adopting various optimizations.
Finally, we develop a new hypergraph-core model, the
\emph{(neighborhood, degree)-core} by considering
both neighborhood and degree constraints, design its decomposition algorithm \textbf{Local-core+Peel}, and demonstrate its superiority in spreading diffusion.
%and finding important hyperedges from \textcolor{red}{protein complexes}.}
%Our most efficient sequential algorithm \textbf{Local-core(OPT)} decomposes
%hypergraph with 27M nodes and 17M hyperedges in-memory within 16 minutes by adopting various optimizations
%that are unique to hypergraphs.
%Our parallel implementation \textbf{Local-core(P)}
%further expedites computation achieving up to 5x speedup over \textbf{Local-core(OPT)

\end{abstract}

\settopmatter{printfolios=true}
\maketitle

\pagestyle{\vldbpagestyle}
\begingroup\small\noindent\raggedright\textbf{PVLDB Reference Format:}\\
\vldbauthors. \vldbtitle. PVLDB, \vldbvolume(\vldbissue): \vldbpages, \vldbyear.\\
\href{https://doi.org/\vldbdoi}{doi:\vldbdoi}
\endgroup
\begingroup
\renewcommand\thefootnote{}\footnote{\noindent
This work is licensed under the Creative Commons BY-NC-ND 4.0 International License. Visit \url{https://creativecommons.org/licenses/by-nc-nd/4.0/} to view a copy of this license. For any use beyond those covered by this license, obtain permission by emailing \href{mailto:info@vldb.org}{info@vldb.org}. Copyright is held by the owner/author(s). Publication rights licensed to the VLDB Endowment. \\
\raggedright Proceedings of the VLDB Endowment, Vol. \vldbvolume, No. \vldbissue\ %
ISSN 2150-8097. \\
\href{https://doi.org/\vldbdoi}{doi:\vldbdoi} \\
}\addtocounter{footnote}{-1}\endgroup

\ifdefempty{\vldbavailabilityurl}{}{
\vspace{.3cm}
\begingroup\small\noindent\raggedright\textbf{PVLDB Artifact Availability:}\\
The source code, data, and/or other artifacts have been made available at \url{\vldbavailabilityurl}.
\endgroup
}

\section{Introduction}
\label{sec:intro}
Decomposition of a graph into hierarchically cohesive subgraphs is an important tool
for solving many graph data management problems, e.g., community detection \cite{Malvestio2020InterplayBK},
densest subgraph discovery \cite{charikar}, identifying influential nodes \cite{Malliaros16},
and network visualization \cite{Alvarez-HamelinDBV05,BatageljMZ99}.
Depending on different notions of cohesiveness,
there are several decomposition approaches: core-decomposition \cite{Batagelj11}, truss-decomposition \cite{WangC12},
nucleus-decomposition \cite{SariyuceSPC15}, etc. In this work, we are interested in decomposing hypergraphs, a generalization of graphs
where an edge may connect more than two entities.

Many real-world relations consist of polyadic entities, e.g., relations between
individuals in co-authorships \cite{HanZPJ09}, legislators in
parliamentary voting~\cite{benson2018simplicial}, items in e-shopping carts \cite{XYYWC021}, proteins in protein complexes and  metabolites in a metabolic process \cite{bty570,FSS15}. For convenience, such relations are often reduced to a clique graph or a bipartite
graph (\S \ref{sec:diff}).
However, these reductions may not be desirable due to two reasons.
First, such reductions might not be meaningful, e.g., 
a pair of proteins in a certain protein complex may not necessarily interact pairwise to create a new functional protein complex.
Second, reducing a hypergraph to a clique graph or a bipartite graph inflates the problem-size 
 \cite{huang2015scalable}:
A hypergraph in~\cite{yang2015defining} with 2M nodes and 15M hyperedges is converted to a bipartite graph
with 17M nodes and 1B edges.
A $k$-uniform hypergraph with $m$ hyperedges causes its clique graph to have $\bigO(mk^2)$ edges. The bipartite graph representation also requires distance-2 core decomposition \cite{bonchi2019distance,LiuZHX21}, which is more expensive due to inflated problem-size (\S \ref{sec:experiments}).

To this end, we propose a novel  neighborhood-cohesion based 
hypergraph core decomposition that 
decomposes a hypergraph into nested, strongly-induced maximal subhypergraphs
such that all the nodes in every subhypergraph have at least a certain number of neighbors in that subhypergraph.
Being strongly-induced means that a hyperedge is only present in a subhypergraph
if and only if all its constituent nodes are present in that subhypergraph.
\begin{figure}[tb!]
    \centering
    \fbox{
        \centering
        \includegraphics[clip, trim=0.1cm 0cm 0.05cm 0cm,width=0.9\linewidth]{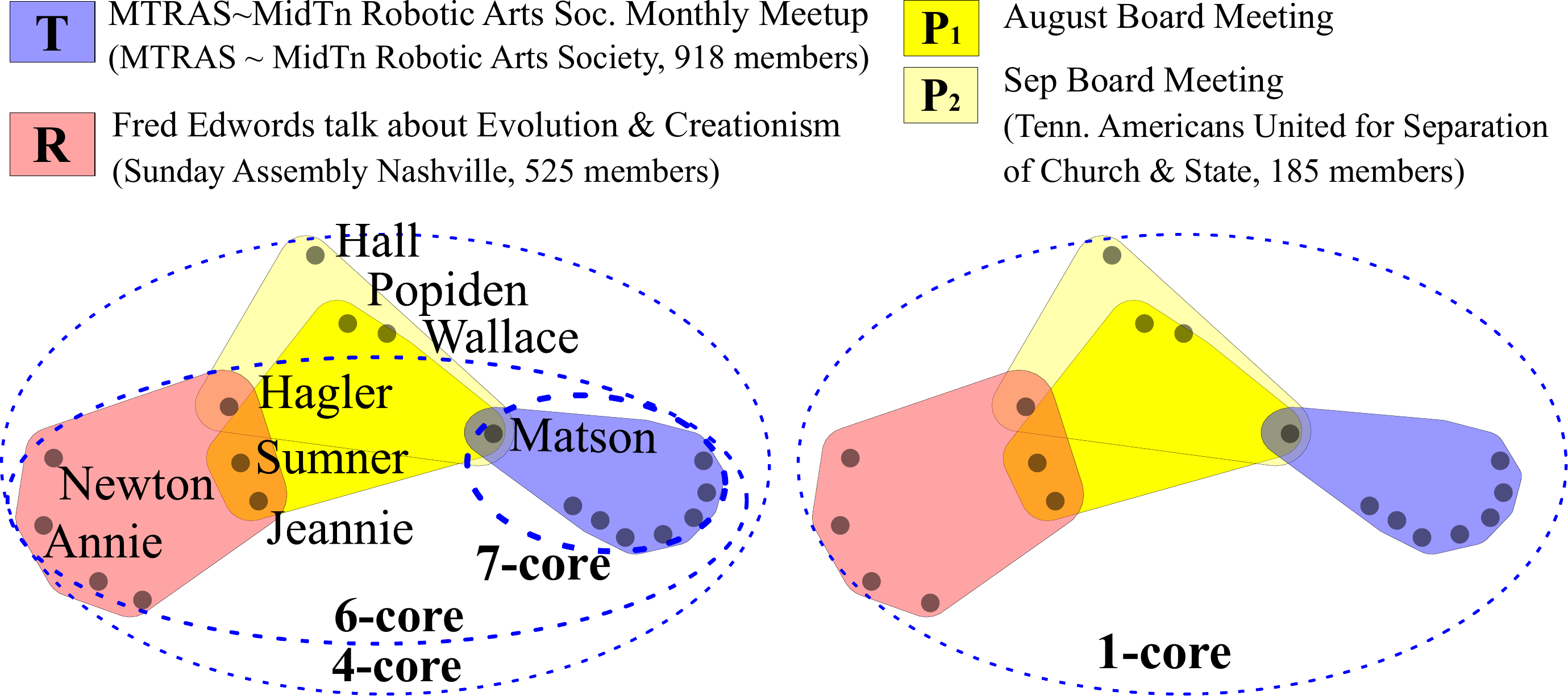}}\vspace{-4mm}
       % \label{fig:intro_nbr}
        % \caption{\footnotesize Neighborhood-based core decomposition of hypergraph $H$}
        % \subfloat[\label{fig:nbr} Neighborhood-based.]{\hspace{.5\linewidth}}
        \subfloat[]{\hspace{.5\linewidth}}
        % \subfloat[\label{fig:deg} Degree-based.]{\hspace{.5\linewidth}}
        \subfloat[]{\hspace{.5\linewidth}}
        \vspace{-1mm}
        \caption{\label{fig:intro_nbr}(a) Neighborhood-based and (b) degree-based core decomposition of hypergraph $H$}\vspace{-5mm}
\end{figure}
\vspace{-0.5mm}
\begin{exam}[Neighborhood-based core decomposition]
The hypergraph $H$ in ~\Cref{fig:intro_nbr}(a) is constructed based on \naheed{four events ($T$, $R$, $P_1$, $P_2$)} 
% a subset of events 
from the {\em Nashville Meetup Network} dataset \cite{meetup}. Each hyperedge denotes an event, nodes in a hyperedge are participants of that event. Two nodes are neighbors if
they \naheed{co-participate} in an event. \naheed{We notice that 
\textit{Hall} has $4$ neighbors: \textit{Popiden}, \textit{Hagler}, \textit{Wallace}, and \textit{Matson}. 
Similarly, \textit{Matson},  \textit{Hagler}, \textit{Jeannie}, \textit{Sumner}, \textit{Annie}, \textit{Newton}, \textit{Wallace}, and \textit{Popiden} have $13, 10, 9, 9, 6, 6, 6$, and $6$ neighbors, respectively.
As every node has $\geq 4$ neighbors, the neighborhood-based $4$-core, denoted by $H[V_4]$, is the hypergraph $H$ itself. 
The neighborhood based $6$-core is the subhypergraph $H[V_6]=\{T,R\}$ because participants of $T$ (e.g., \textit{Newton}, \textit{Annie}) and $R$ (e.g., \textit{Matson}) respectively have $6$ and $7$ neighbors in $H[V_6]$. Finally, the neighborhood based $7$-core is the subhypergraph $\{T\}$.}
\end{exam}
\vspace{-0.5mm}
\spara{Motivation.}
The only hypergraph decomposition existing in the literature is that based on degree (i.e., number of hyperedges incident on a node)
~\cite{ramadan2004hypergraph,sun2020fully}. In \emph{degree-based core decomposition},
every node in the $k$-core has degree at least $k$ in that core. 
It does not consider hyperedge sizes. As a result, nodes in the same core may have vastly 
different neighborhood sizes. \naheed{For instance, \textit{Hall} and \textit{Matson} have $4$ and $13$ neighbors, respectively, yet they belong to the same
1-core in the degree-based decomposition (Figure~\ref{fig:intro_nbr}(b)).} There are applications,
e.g., propagation of contagions in epidemiology, diffusion of information in viral-marketing, where it is desirable
to capture such differences because nodes with the same number of neighbors in a subhypergraph are known to
exhibit similar diffusion characteristics~\cite{kitsak}. Indeed from a practical viewpoint, if the infectious diseases 
control authority is looking for some key events that cause higher infection spread (e.g., meeting with 6 neighbors per participant) and hence such events need to be intervened, 
they are events \naheed{ $\{T,R\}$ (nbr-6-core) as they cause meeting with at least 6 neighbors per participant.} 
Such distinction across various events and participants is not possible according to degree-based core decomposition of $H$. \naheed{Moreover, the neighborhood based decomposition is also logical: innermost core contains $T$, which is a tech event organized by a relatively popular group with 918 members, followed by the second innermost core containing $\{R,T\}$. $R$ is an event organized by a secular organization with 525 members which often discusses religion matters. The outermost core contains two events held by a niche social activist group with only 185 members.}
\spara{Challenges.}
A hyperedge can relate to more than two nodes.
Furthermore, a pair of nodes may be related by multiple yet distinct hyperedges. Thus, a trivial
adaptation of core-decomposition algorithms for graphs to hypergraphs is difficult.
For instance, in the classic peeling algorithm for graph core decomposition, when a node is removed,
the degree of its neighbors is decreased by 1: this allows important optimizations which makes the
peeling algorithm {\em linear-time} and efficient \cite{Batagelj11,ChengKCO11,S13}. %[11, 17, 36, 44, 47, 50].
However, in a neighborhood-based hypergraph core decomposition, deleting a node may reduce
the neighborhood size of its neighboring node by more than $1$.
Hence, to recompute the number of neighbors of a deleted
node's neighbor, one must construct the residual hypergraph after deletion, which makes the
decomposition {\em polynomial-time} and thus expensive.
Furthermore, the existing lower bound that makes graph core decomposition more efficient~\cite{bonchi2019distance}
does not work for neighborhood-based hypergraph core decomposition, and we require a new lower-bound
for this purpose.
In the following, we discuss the challenges associated with adopting
the local approach~\cite{eugene15}, one of the most efficient methods for graph core decomposition to hypergraphs.

In a local approach,  a core-number estimate is updated iteratively~\cite{eugene15,K15}
or in a distributed manner~\cite{distributedcore} for every node in a graph. The initial value of a node's
core-number estimate is an upper bound of its core-number. In subsequent rounds, this estimate is iteratively decreased
based on estimates of neighboring nodes.~\cite{eugene15} uses $h$-index~\cite{hirsch2005index} for such update.
Khaouid et al. \cite{K15} also applies a similar approach to iteratively update core-number estimates.
They have shown that the following invariant must hold: \emph{every node with core-number $k$ has $h$-index at least $k$}, and \emph{the subgraph induced by nodes with $h$-index at least $k$ has at least $k$ neighbors per node in that subgraph}.
The former holds but the later may not hold in a hypergraph,
because the subhypergraph induced by nodes with $h$-index $\geq k$ may not include hyperedges that partially contain other nodes (\S \ref{sec:local}).
Due to those `missing' hyperedges, the number of neighbors of some nodes in that subhypergraph may drop below $k$,
violating the coreness condition. Thus a local approach used for computing the $k$-core~\cite{distributedcore,eugene15,K15}
or 
$(k,h)$-core~\cite{LiuZHX21}
in graphs results in {\em incorrect} neighborhood-based hypergraph cores
(\S \ref{sec:effectiveness}).

%\vspace{-2mm}
\spara{Our contributions} are summarized below.

%\vspace{-0.5mm}
{\bf Novel problem and characterization (\S\ref{sec:preliminaries}).}
We define and investigate the novel problem of neighborhood-based core decomposition in hypergraphs.
We prove that neighborhood-based $k$-cores are unique and the $k$-core contains the $(k+1)$-core.

%\vspace{-0.5mm}
{\bf Exact algorithms (\S\ref{sec:algorithm}).} We propose three exact algorithms 
with their correctness and time complexity analyses. Two of them, \textbf{Peel} and its enhancement \textbf{E-Peel}, adopt classic peeling~\cite{Batagelj11}, incurring global changes to the hypergraph. For \textbf{E-Peel}, we derive {\em novel lower-bound} on core-number that eliminates many redundant neighborhood recomputations. Our third algorithm, \textbf{Local-core} 
only makes node-level local computations. Even though the existing local method~\cite{distributedcore,eugene15} fails to correctly find neighborhood-based core-numbers in a hypergraph, our algorithm \textbf{Local-core} applies a {\em novel} \textbf{Core-correction} procedure, ensuring correct core-numbers.

%\vspace{-0.5mm}
{\bf Optimization and parallelization strategies (\S\ref{sec:optimization}).} We propose four optimization strategies to improve the efficiency of \textbf{Local-core}. Compressed representations for hypergraph (optimization-I) and the family of optimizations for efficient \textbf{Core-correction} (optimization-II) are novel to core-decomposition literature. The other optimizations, though inspired from graph literature, have not been adopted in earlier hypergraph-related works. 
We also propose parallelization of \textbf{Local-core} for the shared-memory programming paradigm.

%\vspace{-0.5mm}
{\bf (Neighborhood, degree)-core  (\S\ref{sec:kdcore}).} We define a more general hypergraph-core model,
\emph{(neighborhood, degree)-core} by considering both neighborhood and degree constraints, propose its decomposition
algorithm \textbf{Local-core+Peel}, and demonstrate its superiority in diffusion spread. 

%\vspace{-1mm}
{\bf Empirical evaluation (\S\ref{sec:experiments})}
on real and synthetic hypergraphs shows that the proposed algorithms are
effective, efficient, and practical.
Our OpenMP parallel implementation \textbf{Local-core(P)}
decomposes hypergraph with 27M nodes, 17M hyperedges in 91 seconds.

% \vspace{-2mm}
{\bf Applications (\S\ref{sec:application}).}
We show our decomposition to be more effective in disrupting diffusion than other decompositions.
Our greedy algorithm proposed for volume-densest subhypergraph recovery achieves $(d_{pair}(d_{card}-2)+2)$-approximation,
where hyperedge-cardinality (\# nodes in that hyperedge) and node-pair co-occurrence (\# hyperedges containing that pair) are at most $d_{card}$  and $d_{pair}$, respectively. 
If the hypergraph is a graph ($d_{card} = 2$), our result generalizes Charikar's $2$-approximation
guarantee for densest subgraph discovery \cite{charikar}.
Our volume-densest subhypergraphs
capture differently important meetup events
compared to degree and clique graph
decomposition-based densest subhypergraphs.

\section{Our Problem and Characterization}
\label{sec:preliminaries}
\spara{Hypergraph.}
A hypergraph $H=(V,E)$ consists of a set of nodes $V$ and a set of hyperedges $E\subseteq P(V)\setminus \phi$,
where $P(V)$ is the power set of $V$. A hyperedge is modeled as an unordered set of nodes.

\spara{Neighbors.}
Neighbors $N(v)$ of a node $v$ in a hypergraph $H=(V,E)$ is the set of nodes $u \in V$ that co-occur with $v$ in some hyperedge $e\in E$. That is,
$N(v) = \{u \in V\mid u \neq v \land \exists~{e\in E} \, \text{ s.t } \, u,v \in e \}$.

\spara{Strongly induced subhypergraph \cite{dewar2017subhypergraphs, bahmanian2015,graham_lovasz}.}
A strongly induced subhypergraph $H[S]$ of a hypergraph $H=(V, E)$, induced by a node set $S \subseteq V$,
is a hypergraph with the node set $S$ and the hyperedge set $E[S] \subseteq E$, consisting of all the hyperedges
that are subsets of $S$.
\begin{small}
\begin{align}
H[S] = (S,E[S]), \,\text{where}\, E[S]=\{e \mid e \in E \land e \subseteq S\}
\end{align}
\end{small}
In other words, every hyperedge in a strongly induced subhypergraph must exist in its parent hypergraph.
%
%\vspace{-1mm}
\subsection{Problem Formulation}
\label{sec:problem}
The {\em nbr-$k$-core} $H[V_k] = (V_k,E[V_k])$ of a hypergraph $H = (V,E)$ is the maximal (strongly) induced subhypergraph such that every node $u \in V_k$ has at least $k$ neighbours in $H[V_k]$. For simplicity, 
we denote nbr-$k$-core, $H[V_k]$ as $H_k$.
The \emph{maximum core} of $H$ is the largest $k$ for which $H_k$ is non-empty. The \emph{core-number} $c(v)$ of a node $v \in V$ is the largest $k$ such that $v \in V_k$ and $v \notin V_{k+1}$. The {\em core decomposition} of a hypergraph assigns to each node its core-number. Given a hypergraph, the problem studied in this paper is to correctly and efficiently compute its neighborhood-based core decomposition.
\subsection{Differences with Other Core Decompositions}
\label{sec:diff}
One can adapt broadly two kinds of approaches from the literature.

\spara{Approach-1.}
Transform the hypergraph into other objects (e.g., a graph),
apply existing decomposition approaches \cite{Batagelj11,bonchi2019distance} on that object, and then project the
decomposition back to the hypergraph. For instance, a hypergraph is transformed to a clique graph
by replacing the hyperedges with cliques, and classical graph core decomposition \cite{Batagelj11} is applied.
A hypergraph can also be transformed into a bipartite graph by representing the hyperedges as nodes in the
second partition and creating an edge between two cross-partition nodes if the hyperedge in the second partition
contains a node in the first partition. Finally, distance-2 core decomposition \cite{bonchi2019distance,LiuZHX21}
is applied. 
In distance-2 core-decomposition, nodes in $k$-core have at least $k$ $2$-hop neighbors in the subgraph. 
{\bf First}, the decomposition they yield may be different from that of ours: \naheed{$c(Popiden) = 5$ in both clique graph and dist-2 bipartite graph decompositions, whereas \textit{Popiden} has core-number $4$ in our decomposition. Clique graph and bipartite graph decompositions fail to identify that the relation $\langle$\textit{Popiden}-\textit{Hagler}$\rangle$ should not exist without the existence of event $P_2$, since $P_2$ is the only event where they co-participate. {\bf Second}, the resulting decomposition may be unreasonable: low-importance events $P_1$ and $P_2$ by the same interest group are placed in different cores causing difficulty in separating less-important events from more-important ones.}{\bf Third}, such transformations inflate the problem size (\S \ref{sec:intro}). Bipartite graph representation results in inefficient decomposition (\S \ref{sec:exp_efficiency}).

%\vspace{-1mm}
\spara{Approach-2.}
The \emph{degree} $d(v)$ of a node $v$ in hypergraph $H$ is the number of
hyperedges incident on $v$~\cite{berge}, i.e.,
$ d(v) = \abs{\{e \in E\mid v \in e\}}$.
Sun et al.~\cite{sun2020fully} define
the {\em deg-$k$-core} $H_k^{deg}$ of a hypergraph $H$ as the maximal (strongly) induced subhypergraph of $H$ such
that every node $u$ in $H_k^{deg}$ has degree at least $k$ in $H_k^{deg}$.
This approach does not consider hyperedge sizes (\S \ref{sec:intro}).
\naheed{Therefore, it does not necessarily yield the same decomposition as our approach (Figure~\ref{fig:intro_nbr})}.
\subsection{Nbr-$k$-Core: Properties}
\label{sec:properties}
\begin{theor}
The nbr-$k$-core $H_k$ is unique for any $k>0$.
\label{th:unique}
\end{theor}
%
%\vspace{-2mm}
\begin{proof}
Let, if possible, there be two distinct nbr-$k$-cores: $H_{k_1} = (V_{k_1},E[{k_1}])$ and $H_{k_2} = (V_{k_2},E[{k_2}])$ of a hypergraph $H = (V,E)$.
By definition, both $H_{k_1}$ and $H_{k_2}$ are maximal strongly induced subhypergraphs of $H$. 
Construct the union hypergraph $H_k = (V_{k_1} \cup V_{k_2}, E[V_{k_1} \cup V_{k_2}])$. For any $u \in V_{k_1} \cup V_{k_2}$, $u$ must be in either $V_{k_1}$ or $V_{k_2}$.
Thus, $u$ must have at least $k$ neighbours in 
$H_{k_1}$ or $H_{k_2}$.
Since $E[V_{k_1}] \cup E[V_{k_2}] \subseteq E[V_{k_1} \cup V_{k_2}]$,
$u$ must also have at least $k$ neighbours in $H_k$.
Since $H_k$ is a supergraph of both  $H_{k_1}$ and $H_{k_2}$, 
$H_{k_1}$ and $H_{k_2}$ are not maximal,
leading to a contradiction. 
\end{proof}
\begin{theor}
The $(k+1)$-core is contained in the $k$-core, $\forall \, k>0$.
\label{th:containment}
\end{theor}
%
%\vspace{-4mm}
\begin{proof}
Let, if possible, for some node $u \in V_{k+1}$,  $u \notin V_{k}$.
Construct $S=V_{k}\cup V_{k+1}$. Since $u~\notin V_k$, but $u \in V_{k+1} \subset S$, 
we get $\abs{S} \geq \abs{V_k} + 1$.
It is easy to verify that every node $v \in S$ has at least $k$ neighbours in $H[S]$ and $\abs{S} > \abs{V_k}$.
Then, $V_k$ is not maximal and thus not nbr-$k$-core, which is a contradiction. The theorem follows.
\end{proof}
\begin{figure}[tb!]
    \centering
        \fbox{
        \centering
        \includegraphics[clip, trim=0.05cm 0cm 0.05cm 0cm,width=0.7\linewidth]{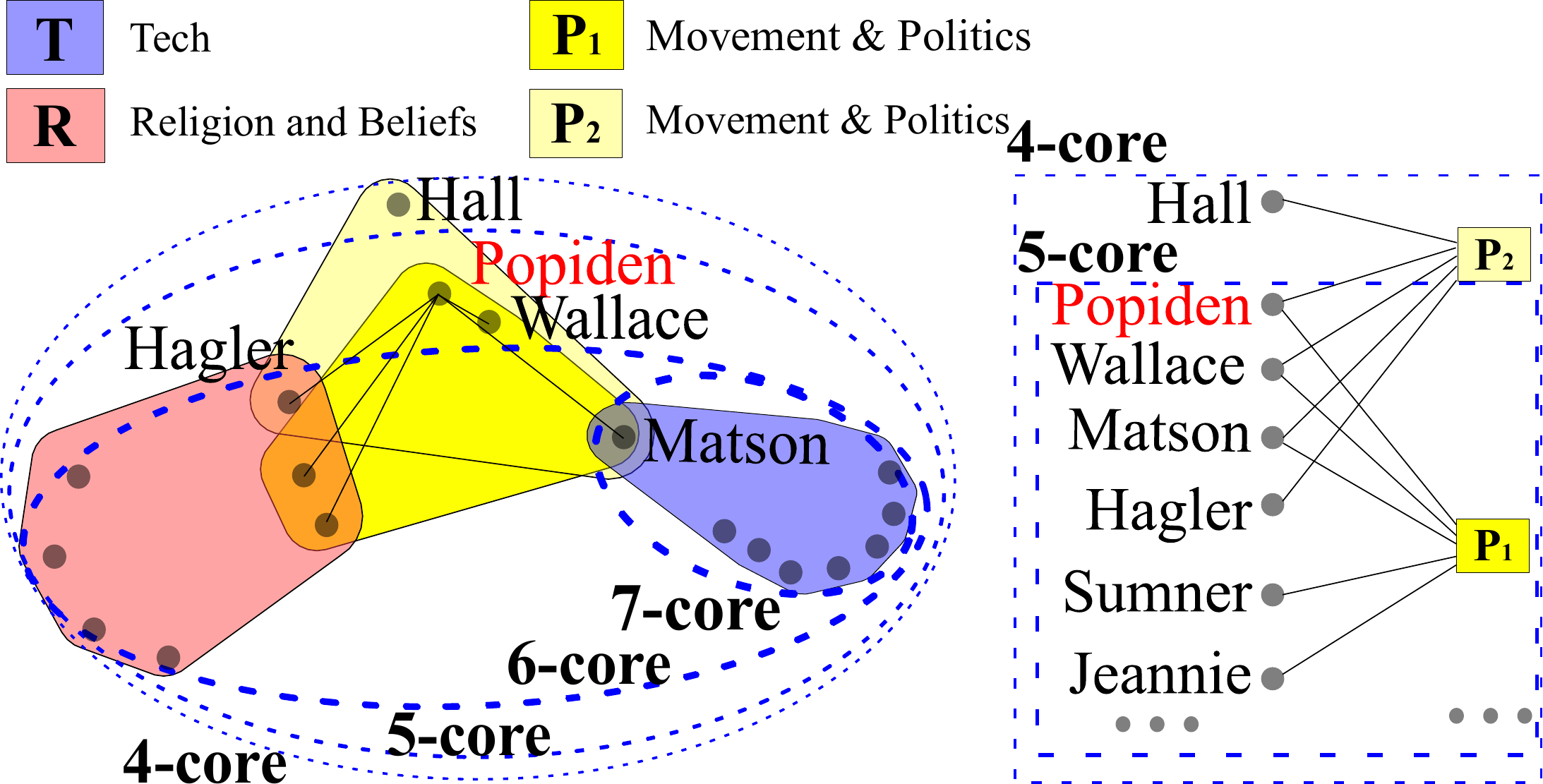}}\vspace{-4mm}
        \subfloat[]
        {\hspace{0.5\linewidth}}
        \subfloat[]
        {\hspace{0.2\linewidth}}\vspace{-3mm}
        \caption{Alternative decompositions are sometimes undesirable: (a) Core decomposition of clique graph of H and Dist-2 core decomposition of the bipartite graph of H produce the same decomposition. Similar events ($P_1$ and $P_2$) are in the different cores. (b) Bipartite graph representation of H.\label{fig:cl_bp}}\vspace{-3mm}
\end{figure}

%\vspace{-2mm}
\section{Algorithms}
\label{sec:algorithm}
We propose three algorithms: 
Our algorithms \textbf{Peel} and its efficient variant \textbf{E-Peel} are inspired by peeling methods similar to graph core computation~\cite{bonchi2019distance,Batagelj11}.
The algorithm \textbf{Local-core} is inspired by local approaches to graph core computation~\cite{eugene15,distributedcore}.
In the classic peeling algorithm
for graph core decomposition, when a node is removed, the degree
of its neighbors reduces by 1: this permits efficient optimizations \cite{Batagelj11,ChengKCO11,S13}, e.g., \cite{Batagelj11}
proposes bin-sort to order nodes
and develops a {\em linear-time} peeling algorithm.
However, in a neighborhood-based hypergraph core decomposition, deleting a node may reduce the neighborhood size of
its neighboring node by more than 1. Hence, to recompute the number of neighbors of a deleted node’s neighbor, one must construct
the residual hypergraph, which makes the decomposition {\em polynomial-time} and costly.
The algorithms $\textbf{E-Peel}$ and $\textbf{Local-core}$, despite being inspired by the existing family of graph algorithms, are by no means trivial adaptations. For $\textbf{E-Peel}$, we devise a new local lower-bound for core-numbers as
the lower-bound for graph core~\cite{bonchi2019distance} is insufficient for our purpose. For \textbf{Local-core}, we show how a direct adaptation of local algorithm \cite{distributedcore,eugene15,K15} leads to incorrect core computations (\S \ref{sec:local} and \S \ref{sec:effectiveness}). Hence, we devise \emph{hypergraph $h$-index} and \emph{local coreness constraint} and employ them to compute hypergraph cores correctly.
\subsection{Peeling Algorithm}
\label{sec:peeling}
Following Theorem~\ref{th:containment}, the $(k +1)$-core can be computed from the $k$-core
by ``peeling'' all nodes whose neighborhood sizes are less than $k+1$.
~\Cref{alg:naivecore} describes our peeling algorithm: \textbf{Peel}, which
processes the nodes in increasing order of their neighborhood sizes (Lines 4-10).
$B$ is a vector of lists: Each cell $B[i]$ is a list storing all
nodes whose neighborhood sizes are $i$ (Line 3).
When a node $v$ is processed at iteration $k$,
its core-number is assigned to $c(v)=k$ (Line 7),
it is deleted from the set of ``remaining'' nodes $V$ (Line 10).
The neighborhood sizes of the nodes in $v$'s neighborhood
are recomputed (each neighborhood size can decrease by more than $1$,
since when $v$ is deleted, all hyperedges involving $v$ are also deleted),
and these nodes are moved to the appropriate cells in $B$ (Lines 8-9).
The algorithm completes when all nodes in the hypergraph are
processed and have their respective core-numbers computed.

\spara{Proof of correctness.}
Initially $B[i]$ contains all nodes whose neighborhood sizes are $i$.
When we delete some neighbor of a node $u$, the neighborhood size of $u$ is recomputed,
and $u$ is reassigned to a new cell corresponding to its reduced neighborhood size
until we find that the removal of a neighbor $v$ of $u$ reduces $u$'s neighborhood size
even below the current iteration number $k$ (Line 9). When this happens, we
correctly assign $u$'s core-number $c(u)=k$.
{\bf (1)} Consider the remaining subhypergraph formed by the remaining nodes and
hyperedges at the end of the $(k-1)^{{\text{th}}}$ iteration. Clearly, $u$ is in the $k$-core since $u$
has at least $k$ neighbors in this remaining subhypergraph, where all nodes in the remaining subhypergraph
also have neighborhood sizes $\geq k$. {\bf (2)} The removal of $v$
decreases $u$'s neighborhood size smaller than the current iteration number $k$, thus when
the current iteration number increases to $k+1$, $u$ will not have enough neighbors
to remain in the $(k+1)$-core.

\spara{Time complexity.}
Each node $v$ is processed exactly once from $B$ in ~\Cref{alg:naivecore};
when it is processed and thereby deleted from $V$, neighborhood sizes of the nodes in $v$'s neighborhood are
recomputed. Assume that the maximum number of neighbors and hyperedges of a node
be $d_{nbr}$ and $d_{hpe}$, respectively.
Thus, ~\Cref{alg:naivecore} has time complexity
$\bigO\big(|V|\cdot d_{nbr}\cdot(d_{nbr}+d_{hpe})\big)$.
\begin{algorithm}[!tb]
\begin{algorithmic}[1]
\caption{\label{alg:naivecore}\small Peeling algorithm: \textbf{Peel}}
\scriptsize
\Require Hypergraph $H = (V,E)$
\Ensure Core-number $c(u)$ for each node $u \in V$
\ForAll {$u \in V$}
    \State Compute $N_V(u)$ \Comment{set of neighbors of $u$ in $H=(V,E)$}
    \State $B[\abs{N_V(u)}] \gets B[\abs{N_V(u)}] \cup \{u\}$
\EndFor
\ForAll {$k = 1,2,\ldots, \abs{V}$}
    \While{$B[k] \neq \phi$}
        \State Remove a node $v$ from $B[k]$
        \State $c(v) \gets k$
        \ForAll{$u \in N_V(v)$}
            \State Move $u$ to $B[\max{\left(\abs{N_{V \setminus \{v\}}(u)},k \right)}]$
        \EndFor
        \State $V \gets V \setminus \{v\}$
    \EndWhile
\EndFor
\State Return $c$
\end{algorithmic}
\end{algorithm}
\subsection{Efficient Peeling with Bounding}
\label{subsec:improved_peel}
An inefficiency in~\Cref{alg:naivecore} is that
it updates the cell index of every node $u$ that is a neighbor of a deleted node $v$.
To do so, it has to compute the number of neighbors of $u$ in the newly constructed
subhypergraph. To delay this recomputation, 
we derive a local lower-bound for $c(u)$ via~\Cref{lem:localLB} and use it to eliminate
many redundant neighborhood recomputations and cell updates (\Cref{alg:Improvednaivecore}).
The intuition is that a node $u$ will not be deleted at some iteration $k<$ the lower-bound on $c(u)$,
thus we do not require computing $u$'s neighborhood size until the value of $k$ reaches the lower-bound on
$c(u)$. Our lower-bound is local since it is specific to each node. 
\begin{lem}[Local lower-bound]
\label{lem:localLB}
Let $e_m(v)= \argmax\{\abs{e}: e \in E \land v \in e \}$ be the highest-cardinality hyperedge
incident on $v \in V$. For all $v \in V$,
\begin{small}
\begin{align}
c(v) \geq \max\left( \abs{e_m(v)} - 1, \min_{u \in V}\,\abs{N(u)} \right) = \mathbf{LB}(v)
\end{align}
\end{small}
\end{lem}
\begin{algorithm}[!tb]
\begin{algorithmic}[1]
\caption{\label{alg:Improvednaivecore}\small Efficient peeling algorithm with bounding: \textbf{E-Peel}}
\scriptsize
\Require Hypergraph $H = (V,E)$
\Ensure Core-number $c(u)$ for each node $u \in V$
\ForAll {$u \in V$}
    \State Compute $\mathbf{LB}(u)$
    \State $B[\mathbf{LB}(u)] \gets B[\mathbf{LB}(u)] \cup \{u\}$
    \State $setLB(u) \gets True$
\EndFor
\ForAll {$k = 1,2,\ldots, \abs{V}$}
    \While{$B[k] \neq \phi$}
        \State Remove a node $v$ from $B[k]$
        \If{$setLB(v)$}
            \State $B[\abs{N_V(v)}] \gets B\left[\max\left(\abs{N_V(v)},k\right)\right] \cup \{v\}$ %B[\abs{N_V(v)}] \cup \{v\}$
            \State $setLB(v) \gets False$
        \Else
            \State $c(v) \gets k$
            \ForAll{$u \in N_V(v)$}
                \If{ $\neg setLB(u)$}
                    \State Move $u$ to $B\left[\max{\left(\abs{N_{V\setminus \{v\}}(u)},k \right)}\right]$
                \EndIf
            \EndFor
            \State $V \gets V \setminus \{v\}$
        \EndIf
    \EndWhile
\EndFor
\State Return $c$
\end{algorithmic}
\end{algorithm}
%
%\vspace{-2mm}
\begin{proof}
Notice that $c(v) \geq  \min_{u \in V}\abs{N(u)}$, since all nodes in the input hypergraph must be
in the $(\min_{u \in V}\abs{N(u)})$-core. Next, we show that $c(v) \geq \abs{e_m(v)} - 1$, by contradiction.
Let, if possible, $ \abs{e_m} - 1 > c(v)$. This implies that
$v$ is not in the $(\abs{e_m} - 1)$-core, denoted by $H[V_{\abs{e_m} - 1}]$.
Consider $V' = V_{\abs{e_m}-1} \cup \{u:u \in e_m\}$.
Clearly, $\abs{V'} \geq \abs{V_{\abs{e_m} - 1}} + 1$, since $v \notin V_{\abs{e_m} -1}$,
but $v \in e_m$, so $v\in V'$. We next show that $H[V_{\abs{e_m} - 1}]$ is
not the maximal subhypergraph where every node has at least $\abs{e_m} - 1$ neighbors,
which is a contradiction.

To prove non-maximality of $H[V_{\abs{e_m} - 1}]$,
it suffices to show that for any $u \in V'$, $N_{V'}(u) \geq \abs{e_m} - 1$.
If $u \in V_{\abs{e_m} - 1} \subset V'$,
$\abs{N_{V'}(u)} \geq \abs{N_{V_{\abs{e_m} - 1}}(u)} \geq \abs{e_m} - 1$.
If $u \in e_m$, $N_{V'}(u) \geq N_{e_m}(u) = \abs{e_m} - 1$.

Since our premise $\abs{e_m} - 1 > c(v)$ contradicts the fact
that $H[V_{\abs{e_m} - 1}]$ is the ($\abs{e_m} -1$)-core, $\abs{e_m} - 1 \leq c(v)$.
\end{proof}

\spara{Algorithm.} Our efficient peeling approach is given in~\Cref{alg:Improvednaivecore}: \textbf{E-Peel}.
In Line~14, we do not recompute neighborhoods and update cells for those neighboring nodes $u$
for which $setLB$ is True, thereby improving the efficiency.
{\em $setLB$ is True for nodes for which $\mathbf{LB}()$ is known, but $N_V()$ at the current iteration is unknown}.
\begin{exam}
\Cref{fig:nbr1}(a) illustrates the improvements made by~\Cref{alg:Improvednaivecore} in
terms of neighborhood recomputations and cell updates.
Since $\mathbf{LB}(x) = 1$ and every neighbor $u \in \{a,b,c,d\}$ has $\mathbf{LB}(u) = 5$,
~\Cref{alg:Improvednaivecore} computes $c(x)$ before $c(u)$.
Due to the local lower-bound-based initialization in Lines~1-4 and
ascending iteration order of $k$, $x$ is popped before $u$.
The first time $x$ is popped from $B$,
$x$ goes to $B[4]$ due to Line~9
and $setLB(x)$ is set to False in Line~10.
The next time $x$ is popped (also at iteration $k=4$),
the algorithm computes $c(x)$ in Line~12. $setLB(u)$
is still True for $u$ (Lines~13-15), as the default initialization
of $setLB(u)$ has been True (Line~4). Hence,
none of the computations in Line~15 is executed for $u$. Intuitively, since the $5$-core does not contain $x$, deletion of $x$ and the yellow hyperedges should be inconsequential to computing $c(u)$ correctly.
$c(u)$ is computed in the next iteration ($k=5$) after
it is popped and is reassigned to $B[5]$, and $setLB(u)$ becomes False.
\end{exam}
\spara{Proof of correctness.} The proof of correctness follows that of \Cref{alg:naivecore}.
When a node $v$ is extracted from $B[k]$ at iteration $k$, we check $setLB(v)$.
{\bf (1)} Lemma~\ref{lem:localLB} ensures that, if we extract a node $v$ from $B[k]$ and $setLB(v)$ is True,
then $c(v)\geq k$. In that case, we compute the current value of $N_V()$, where $V$ denotes the set of remaining nodes,
and insert $v$ into the cell:  $B\left[\max\left(\abs{N_V(v)},k\right)\right]$.
We also set $setLB(v)$ = False, implying that $N_V(v)$ at the current iteration is known.
{\bf (2)} In contrast, if we extract a node $v$ from $B[k]$ and $setLB(v)$ is False, this indicates that
$c(v)=k$, following the same arguments as in \Cref{alg:naivecore}.
In this case, we correctly assign $v$'s core-number to $k$,
and $v$ is removed from $V$. Moreover, for those neighbors $u$ of $v$ for which $setLB(u)$ is True, implying
that $c(u)\geq k$, we appropriately delay recomputing their neighborhood sizes.
\begin{figure}
    %\vspace{-4mm}
    \centering
    \subfloat{\includegraphics[width = 0.2\textwidth]{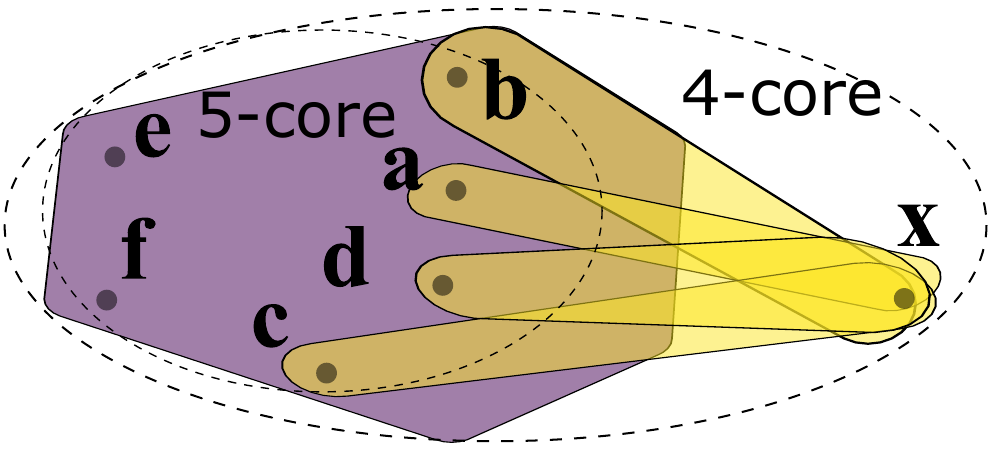}}
    \hspace{2mm}
    \subfloat{\includegraphics[width = 0.2\textwidth,trim = {1.7cm 1.7cm 0.3cm 1.5cm},clip]{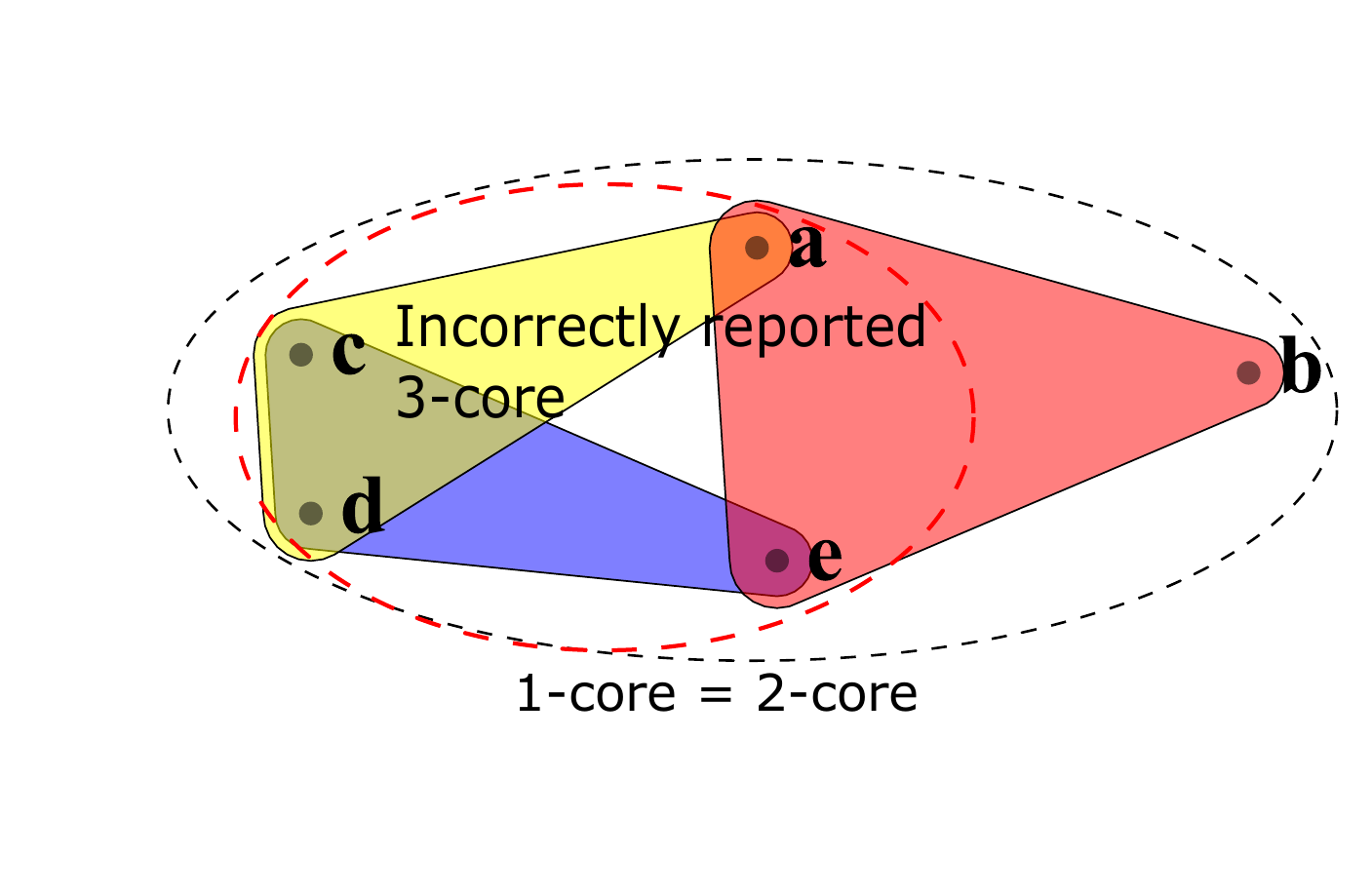}}
    \vspace{-2mm}
    \caption{\footnotesize (a) During $x$'s core-number computation,~\Cref{alg:Improvednaivecore} does not perform neighborhood recomputations and cell updates for $x$'s neighbors $\{a,b,c,d\}$; thus saving four redundant neighborhood recomputations and cell updates. (b) For any $n>1$, the $h$-index (\Cref{def:wrongHind}) of node $a$ never reduces from $h_{a}^{(1)} = \mathcal{H}(2,3,3,4) = 3$ to its core-number $2$: $\lim_{n \to \infty} h_{a}^{(n)} = 3 \neq$ core-number of $a$. Because $a$ will always have at least 3 neighbors ($c$, $d$, and $e$) whose $h$-indices are at least $3$. As a result, the na\"ive approach reports an incorrect $3$-core.}
    \label{fig:nbr1}
    \vspace{-4mm}
\end{figure}

\spara{Time complexity.}
Following similar analysis as in ~\Cref{alg:naivecore}, \textbf{E-Peel} has time complexity
$\bigO\big(\alpha\cdot|V|\cdot d_{nbr}\cdot(d_{nbr}+d_{hpe})\big)$, where $\alpha\leq 1$ is the
ratio of the number of neighborhood recomputations in ~\Cref{alg:Improvednaivecore} over
that in ~\Cref{alg:naivecore}. Based on our experimental results in \S\ref{sec:experiments},
\textbf{E-Peel} can be up to 17x faster than \textbf{Peel}.
\subsection{Local Algorithm}
\label{sec:local}
Although \textbf{Peel} and its efficient variant \textbf{E-Peel} correctly computes core-numbers, they must modify the remaining hypergraph at every iteration by peeling nodes and hyperedges. Peeling impacts the hypergraph data structure globally and must be performed in sequence. Thus, there is little scope for making \textbf{Peel} and \textbf{E-Peel} more efficient via parallelization. Furthermore, they are not suitable in a time-constrained setting where a high-quality partial solution is sufficient. We propose a novel local algorithm that is able to provide partial solutions, amenable to optimizations, and parallelizable.

\noindent\textbf{Na\"ive adoption of local algorithm in hypergraphs: a negative result.}
Eugene et al.~\cite{eugene15} adopt
Hirsch's index~\cite{hirsch2005index}, popularly known as the \emph{$h$-index}~(\Cref{defn:H}),
to propose a local algorithm for core computation in graphs. This algorithm relies on a
recurrence relation
that defines higher-order $h$-index (\Cref{def:wrongHind}).
The local algorithm for graph core computation starts by computing $h$-indices of order $0$ for every node in the graph.
At each iteration $n>0$, it computes order-$n$ $h$-indices using order-$(n-1)$ $h$-indices computed in the previous iteration.
It is well-known that higher-order $h$-indices monotonically converge to the core-number of that node in the graph.
\begin{defn}[$\mathcal{H}$-operator~\cite{eugene15,hirsch2005index}]
\label{defn:H}
Given a finite set of positive integers $\{x_1,x_2,\ldots,x_t\}$,
$\mathcal{H}\left(x_1,x_2,\ldots,x_t\right) = y>0$,
where $y$ is the maximum integer such that there exist at least $y$ elements in $\{x_1,x_2,\ldots,x_t\}$, each of which is at least $y$.
\end{defn}
\begin{exam}[$H$-operator]
$\mathcal{H}(1,1,1,1) = 1$,
$\mathcal{H}(1,1,1,2) = 1$,
$\mathcal{H}(1,1,2,2) = 2$,
$\mathcal{H}(1,2,2,2) = 2$,
$\mathcal{H}(1,2,3,3) = 2$,
$\mathcal{H}(1,3,3,3) = 3$
\end{exam}
\begin{defn}[$h$-index of order $n$~\cite{eugene15}]
\label{def:wrongHind}
Let $\{u_1,u_2,\ldots,u_t\}$ be the set of neighbors of node $v \in V$ in graph $G = (V,E)$.
The $n$-order $h$-index of node $v \in V$, denoted as $h_v^{(n)}$, is defined for any $n \in \mathbb{N}$ by the recurrence relation
\begin{small}
\begin{equation}
\label{eqn:graphhn1}
    h_v^{(n)}=
    \begin{cases}
      \abs{N(v)} & n=0\\
      \mathcal{H}\left(h^{(n-1)}_{u_1},h^{(n-1)}_{u_2},\ldots,h^{(n-1)}_{u_t}\right) & n \in \mathbb{N}\setminus \{0\}
  \end{cases}
\end{equation}
\end{small}
\end{defn}
For neighborhood-based hypergraph core decomposition via local algorithm,
we define $h_v^{(0)}$ as the number of neighbors of node $v$ in hypergraph $H = (V,E)$
(instead of graph $G$). The definition of $h_v^{(n)}$ for $n>0$ remains the same.
However, this {\em direct adoption of local algorithm to compute hypergraph cores
does not work}. Although one can prove that the sequence $(h_v^{(n)})$ adopted for hypergraph
has a limit, that value in-the-limit is not necessarily the core-number $c(v)$ for every
$v \in V$. For some node, the value in-the-limit of its $h$-indices is strictly greater than the core-number of that node.
The reason is as follows. $\mathcal{H}$-operator acts as both necessary and sufficient condition
for computing graph cores. It has been shown that the subgraph induced by $G[S]$, where $S$ contains all neighbors $u$ of a node $v$ such that $h_u^{(\infty)} \geq c(v)$, satisfies $h_v^{(\infty)} = c(v)$~\cite[p.5 Theorem 1]{eugene15}. However, ~\Cref{def:wrongHind} is not sufficient to show $h_v^{(\infty)} = c(v)$ for a hypergraph. Because it is not guaranteed that $v$ will have at least $h_v^{\infty}$ neighbors in the subhypergraph $H[\{u: h_u^{\infty} \geq h_v^{\infty}\}]$ that is reported as the $c(v)$-core. So, the reported $c(v)$-core can be incorrect.
\begin{exam}
\label{ex:wronglocal}
For the hypergraph in~\Cref{fig:nbr1}(b), the values in-the-limit of $h$-indices are $h_{a}^{(\infty)} = h_{c}^{(\infty)} = h_{d}^{(\infty)}= h_{e}^{(\infty)} = 3$ and $h_{b}^{(\infty)} = 2$. No matter how large $n$ is chosen,~\Cref{eqn:graphhn1} does not help $h_{a}^{(n)}$ to reach the correct core-number (=2) for $a$.
Three neighbors of node $a$, namely $c$, $d$, and $e$ have their $h^{\infty}$-values at least $3 = h_v^{\infty}$ in ~\Cref{fig:nbr1}(b). But, $a$ does not have at least $3$ neighbors in the subhypergraph
$H[\{a, c, d , e \}] = H[\{u: h_u^{\infty} \geq h_v^{\infty}\}]$. Thus, the $3$-core $H[\{a, c, d , e \}]$ reported by the na\"ive $h$-index based local approach is incorrect. The reason is that $a$ and $e$ are no longer neighbors to each other in $H[\{a, c, d , e \}]$ due to the absence of $b$.
\end{exam}
\begin{algorithm}[!tb]
\begin{algorithmic}[1]
\caption{\label{alg:localcore}\small Local algorithm with local coreness constraint: \textbf{Local-core}}
\scriptsize
\Require Hypergraph $H = (V,E)$
\Ensure Core-number $c(v)$ for each node $v \in V$
\ForAll {$v \in V$}
    \State $\hat{h}_{v}^{(0)} = h_{v}^{(0)} \gets \abs{N(v)}$.
\EndFor

\ForAll {$n = 1,2,\ldots, \infty$}

    \ForAll {$v \in V$}
         \State $h_v^{(n)} \gets \min\left(\mathcal{H}(\{\hat{h}_u^{(n-1)} : u \in N(v)\}), \hat{h}_v^{(n-1)} \right)$
    \EndFor
    \ForAll {$v \in V$}
        \State $c(v) \gets \hat{h}_v^{(n)} \gets $ \textbf{Core-correction ($v$, $h_v^{(n)}$, $H$)}
    \EndFor
    \If {$\forall v, \hat{h}_v^{(n)} == h_v^{(n)}$}
       \State \textbf{Terminate Loop}
    \EndIf

\EndFor
\State Return $c$
\end{algorithmic}
\end{algorithm}
\noindent\textbf{Local algorithm with local coreness constraint.}
Motivated by the observation mentioned above,
we define a constraint as a sufficient condition,
upon satisfying which we can guarantee that for every node
$v$,  1) the sequence of its $h$-indices converges and 2) the value in-the-limit $h_v^{(\infty)}$ is such that
$v$ has at least $h_v^{(\infty)}$ neighbors in the subhypergraph induced by $H[u: h_u^{(\infty)} \geq h_v^{(\infty)}]$. The first condition is critical for algorithm termination. The second condition is critical for correct computation of core-numbers as discussed in~\Cref{ex:wronglocal}.
\begin{defn}[Local coreness constraint (LCC)]
\label{defn:lccsat}
Given a positive integer $k$, for any node $v \in V$, let $H^+(v) = (N^+(v), E^+(v))$ be the subhypergraph of $H$ such that for any $n>0$
\begin{small}
\begin{align}
& E^+(v) = \{e \in Incident(v): h_u^{(n)} \geq k, \forall u \in e \} \nonumber \\
& N^+(v) = \{u: u \in e, \forall e\in E^+(v)\} \setminus \{ v\}
\end{align}
\end{small}
Local coreness constraint (for node $v$) is satisfied at $k$, denoted as $LCCSAT(k)$, iff $\exists H^+(v)$ with at least $k$ nodes, i.e., $\abs{N^+(v)} \geq k$.
Here, $Incident(v)$ is the set of hyperedges incident on $v$.
\end{defn}
We define \emph{Hypergraph $h$-index} based on the notion of $LCCSAT$ and a re-defined recurrence relation for $h_v^{(n)}$.
\begin{defn}[Hypergraph $h$-index of order $n$]
\label{defn:cor_core}
The Hypergraph $h$-index of order $n$ for node $v$, denoted as $\hat{h}_v^{(n)}$,
is defined for any natural number $n \in \mathbb{N}$ by the following recurrence relation:
\begin{small}
\begin{equation}
\label{eqn:hhatn}
    \hat{h}_v^{(n)}=
    \begin{cases}
      \abs{N(v)} & n=0\\
      h_v^{(n)}& n>0 \land LCCSAT(h_v^{(n)})\\
      \max\{ k \mid k<h_v^{(n)} \land LCCSAT(k) \} & n>0 \land \neg LCCSAT(h_v^{(n)})
  \end{cases}
\end{equation}
\end{small}
\end{defn}
$h_v^{(n)}$ is a newly defined recurrence relation
on hypergraphs:
\begin{small}
\begin{equation}
\label{eqn:hn}
    h_v^{(n)}=
    \begin{cases}
      \abs{N(v)} & n=0\\
      \min\left(\mathcal{H}\left(\hat{h}^{(n-1)}_{u_1},\hat{h}^{(n-1)}_{u_2},\ldots,\hat{h}^{(n-1)}_{u_t}\right), \hat{h}_v^{(n-1)} \right) & n \in \mathbb{N}\setminus \{0\}
  \end{cases}
\end{equation}
\end{small}

The recurrence relations in Equations~(\ref{eqn:hhatn}) and~(\ref{eqn:hn}) are coupled:
$\hat{h}^{(n)}_v$ depends on the evaluation of $h^{(n)}_v$, which in turn depends on the evaluation of $\hat{h}^{(n-1)}_v$.
Such inter-dependency causes both sequences to converge, as proven in our correctness analysis.

\noindent\textbf{Local-core} (\Cref{alg:localcore}) initializes $h_v^{(0)}$ and $\hat{h}_v^{(0)}$ to $\abs{N(v)}$ for every node $v \in V$ (Lines~1-2) following~\Cref{eqn:hn} and~\Cref{eqn:hhatn}, respectively. At every iteration $n>0$, ~\Cref{alg:localcore} first computes $h_v^{(n)}$ for every node $v \in V$ (Lines~4-5) following~\Cref{eqn:hn}. In order to decide whether the algorithm should terminate at iteration $n$ (Lines~8-9), the algorithm computes $\hat{h}_v^{(n)}$ using ~\Cref{alg:core_correct}. ~\Cref{alg:core_correct} checks for every node $v \in V$, whether $LCCSAT(h_v^{(n)})$ is True or False. Following~\Cref{eqn:hhatn}, if $LCCSAT(h_v^{(n)})$ is True it returns $h_v^{(n)}$; if $LCCSAT(h_v^{(n)})$ is False, a suitable value lower than $h_v^{(n)}$ is returned. The returned value $\hat{h}_v^{(n)}$ is considered as the estimate of core-number $c(v)$ at that iteration %, hence the assignment $c[v] \gets \hat{h}_v^{(n)}$ in
(Line 7).
\begin{algorithm}[!tb]
\begin{algorithmic}[1]
\caption{\label{alg:core_correct}\small \textbf{Core-correction} procedure }
\scriptsize
\Require Node $v$, $v$'s hypergraph $h$ index $h_v^{(n)}$, hypergraph $H$
\While{$h_v^{(n)} > 0$}
\State Compute $E^+(v) \gets \{e \in Incident(v): h_u^{(n)} \geq h_v^{(n)}, \forall u \in e\}$
\State Compute $N^+(v) = \{u: u \in e, \forall e\in E^+(v)\} \setminus \{ v\}$
\If {$\abs{N^+(v)} \geq h_v^{(n)}$}
    \State return $h_v^{(n)}$
\Else
    \State $h_v^{(n)} \gets h_v^{(n)} - 1$
\EndIf
\EndWhile
\end{algorithmic}
\end{algorithm}
To compute $LCCSAT(h_v^{(n)})$,~\Cref{alg:core_correct} checks in Line 4 if the subhypergraph $H^+(v) = (N^+(v),E^+(v))$ constructed in Lines~2-3 contains at least $h_v^{(n)}$ neighbors of $v$. If $v$ has at least $h_v^{(n)}$ neighbors in the subhypergraph $H^+(v)$, due to~\Cref{eqn:hhatn} no correction to $h_v^{(n)}$ is required. In this case, ~\Cref{alg:core_correct} returns $h_v^{(n)}$ in Line 5.
If $v$ does not have at least $h_v^{(n)}$ neighbors in the subhypergraph $H^+(v)$ (Line 4),
$LCCSAT(h_v^{(n)})$ is False by~\Cref{defn:lccsat}. Following~\Cref{eqn:hhatn}, a correction to $h_v^{(n)}$ is required. In search for a suitable corrected value lower than $h_v^{(n)}$ and a suitable subhypergraph $H^+(v)$, Line 7 keeps reducing $h_v^{(n)}$ by 1. Reduction to $h_v^{(n)}$ causes $\abs{N^+(v)}$ to increase, 
while $h_v^{(n)}$ 
decreases, until the condition in Line 4 is satisfied. At some point a suitable subhypergraph must be found.

~\Cref{thm:coreconv} proves that the numbers returned by~\Cref{alg:localcore} at that point indeed coincide with the true core-numbers. The termination condition $\hat{h}_v^{(n)} = h_v^{(n)}$ must be satisfied at some point because~\Cref{thm:interleave} proves that $\lim_{n \to \infty} \hat{h}_v^{(n)} = \lim_{n \to \infty} h_v^{(n)}$ %for every node
$\forall v \in V$.
\begin{exam}
Consider iteration $n=1$ of ~\Cref{alg:localcore}, when the input to the algorithm is the hypergraph in~\Cref{fig:nbr1}(b). The algorithm corrects the core-estimate $h_a^{(1)} = 3$ to $\hat{h}_a^{(1)} = 2$ in Line 7. Because in Line 4 of \textbf{Core-correction}, the algorithm finds that for $h_a^{(1)} = 3$, $a$ only has $\abs{N^+(v)} = 2$ neighbors in $H^+(v) = H[\{a, c, d , e \}]$ thus violating the condition that $\abs{N^+(v)}> h_a^{(1)}$. Hence $h_a^{(1)}$ needs to be corrected to satisfy $LCCSAT(h_a^{(1)})$. In Line 7 of \textbf{Core-correction}, it reduces $h_a^{(1)}$ by one. Subsequently for  $h_a^{(1)}=2$ the subhypergraph $H^+(a) = H$ indeed satisfies $LCCSAT(h_a^{(1)})$. This is how
~\Cref{alg:localcore} corrects the case of incorrect core-numbers discussed in~\Cref{ex:wronglocal}. 
\end{exam}
\subsection{Theoretical Analysis of Local-core}
\label{subsec:proofs_localcore}
\spara{Proof of Correctness.}
~\Cref{alg:localcore} terminates after a finite number of iterations because by~\Cref{thm:interleave}, for
any $v \in V$, sequences $(h_v^{(n)})$ and $(\hat{h}_v^{(n)})$ are finite and have the same limit.
At the limit, $\forall v \in V$, $\lim_{n \to \infty} h_v^{(n)} = \lim_{n \to \infty} \hat{h}_v^{(n)}$ holds and it
follows from~\Cref{thm:coreconv} that $\forall v \in V,~\lim_{n \to \infty} h_v^{(n)} = \lim_{n \to \infty} \hat{h}_v^{(n)} = c(v)$.
Thus, when the algorithm terminates, it returns the correct core-number for every node $v \in V$.

%Due to limitation of space, %we only provide proof sketches, while the formal 
%proofs are given in our extended version \cite{full}.
%
\begin{theor}
\label{thm:interleave}
For any node $v \in V$ of a hypergraph $H = (V,E)$, the two sequences
$(h_v^{(n)})$ defined by~\Cref{eqn:hn} and $(\hat{h}_v^{(n)})$ defined by~\Cref{eqn:hhatn} have the same limit:
$\lim_{n \to \infty} h_v^{(n)} = \lim_{n \to \infty} \hat{h}_v^{(n)}$.
\end{theor}
\eat{
\spara{Proof sketch.} Construct a new sequence $(h\hat{h}_v)$ by interleaving components from $(h_v^{(n)})$ and $(\hat{h}_v^{(n)})$ as the following:
\begin{small}
\begin{equation}
\label{eq:interleave}
(h\hat{h}_v) = (h_v^{(0)}, \hat{h}_v^{(0)},h_v^{(1)}, \hat{h}_v^{(1)},h_v^{(2)}, \hat{h}_v^{(2)}\ldots )
\end{equation}
\end{small}
It can be verified that this \emph{interleave sequence}~\cite[Defn~680]{thierry2017handbook}, denoted as $(h\hat{h}_v)$,
is monontonically non-increasing and is lower-bounded by $0$.
By Monotone convergence theorem~\cite{bartle}, $(h\hat{h}_v)$ has a limit. An interleave sequence has a limit if and only if its constituent sequence pairs are convergent
and have the same limit~\cite{thierry2017handbook}. Since $(h\hat{h}_v)$ has a limit,
the sequences $(h_v^{(n)})$ and $(\hat{h}_v^{(n)})$ converges to the same limit:
$\lim_{n \to \infty} h_v^{(n)} = \lim_{n \to \infty} \hat{h}_v^{(n)}$.
}
\begin{proof}
Construct a new sequence $(h\hat{h}_v)$ by interleaving components from $(h_v^{(n)})$ and $(\hat{h}_v^{(n)})$ as the following:
\begin{small}
\[ (h\hat{h}_v) = (h_v^{(0)}, \hat{h}_v^{(0)},h_v^{(1)}, \hat{h}_v^{(1)},h_v^{(2)}, \hat{h}_v^{(2)}\ldots )
\]
\end{small}
We first show that this \emph{interleave sequence}~\cite[Defn~680]{thierry2017handbook}, denoted as $(h\hat{h}_v)$,
is monontonically non-increasing and is lower-bounded by $0$.
Arbitrarily select any pair of successive components from $(h\hat{h}_v)$.

\noindent\textbf{Case 1 $(\hat{h}_v^{(n-1)}, h_v^{(n)})$:} For $n \geq 1$, if $\mathcal{H}\left(\hat{h}_{u_1}^{(n-1)}, \hat{h}_{u_2}^{(n-1)},\ldots,\hat{h}_{u_t}^{(n-1)} \right) > \hat{h}_v^{(n-1)}$, by the recurrence relation (\Cref{eqn:hn}) $h_v^{(n)} = \hat{h}_v^{(n-1)}$. Otherwise, $\hat{h}_v^{(n-1)} \geq \mathcal{H}\left(\hat{h}_{u_1}^{(n-1)}, \hat{h}_{u_2}^{(n-1)},\ldots,\hat{h}_{u_t}^{(n-1)} \right) = h_v^{(n)}$.
Here, the equality in the final expression is again due to the recurrence relation in \Cref{eqn:hn}. Thus, for any $n \geq 1,~\hat{h}_v^{(n-1)} \geq h_v^{(n)}$.

\noindent\textbf{Case 2 $(h_v^{(n)}, \hat{h}_v^{(n)})$: } For $n = 0$, $h_v^{(0)} = \hat{h}_v^{(0)}$ by ~\Cref{eqn:hhatn}.
For $n>0 \, \land \, LCCSAT(h_v^{(n)})$, again by~\Cref{eqn:hhatn}, $h_v^{(n)} =  \hat{h}_v^{(n)}$.
For $n>0 \, \land \, \neg LCCSAT(h_v^{(n)})$, $\hat{h}_v^{(n)}= \max \{k : k < h_v^{(n)} \land LCCSAT(k)\} < h_v^{(n)}$.
So, for any $n \geq 0$, irrespective of $LCCSAT(h_v^{(n)})$, $h_v^{(n)} \geq \hat{h}_v^{(n)}$.

Since for any arbitrarily chosen pair of successive components, the left component is not smaller than the right component,
the sequence $(h\hat{h})$ is monotonically non-increasing: $h_v^{(0)} \geq \hat{h}_v^{(0)} \geq h_v^{(1)} \geq \hat{h}_v^{(1)} \geq h_v^{(2)} \geq \hat{h}_v^{(2)} \geq \ldots$

Next, we prove that $h_v^{(n)} \geq 0$ and $\hat{h}_v^{(n)} \geq 0$ for all $n \geq 0$ by induction. For $n=0$ due to~\Cref{eqn:hn} and~\Cref{eqn:hhatn}, $h_v^{(0)} = \hat{h}_v^{(0)} = N(v) \geq 0$. Assume for $n=m$, $h_v^{(m)} \geq 0$ and $\hat{h}_v^{(m)} \geq 0$ (induction hypothesis). For $n=m+1$, the recurrence relation for $h_v^{(m+1)}$ is either $\abs{N(v)} \geq 0$, or output of an $\mathcal{H}$-operator (strictly greater than $0$ by~\Cref{defn:H}), or $\hat{h}_v^{(m)} \geq 0$ (by induction hypothesis). Considering all three cases: $h_v^{(m+1)} \geq 0$. For any $n=m+1$ as defined in~\Cref{eqn:hhatn}, $\hat{h}_v^{(m+1)}$ is either $\abs{N(v)} \geq 0$,
or $h_v^{(m+1)} \geq 0$ (just proven), or a value $k^* = \max \{k : k < h_v^{(m+1)} \land LCCSAT(k) \}$. The set $\{ k : k < h_v^{(m+1)} \land LCCSAT(k) \}$ must contain $0$. Because for any $v$, $h_v^{(m+1)}\geq 0$ (proven already) and $LCCSAT(0)$ is trivially True.
Since the set $\{ k : k < h_v^{(m+1)} \land LCCSAT(k) \}$ contains $0$, the maximum of this set $k^*$ must be at least $0$.
Considering all three cases: $\hat{h}_v^{(m+1)} \geq 0$ which completes our proof that the interleave sequence is lower-bounded by $0$.

It follows that $(h\hat{h}_v)$ is monotonically non-increasing and lower-bounded by $0$. By Monotone convergence theorem~\cite{bartle}, $(h\hat{h}_v)$ has a limit. An interleave sequence has a limit if and only if its constituent sequence pairs are convergent
and have the same limit~\cite{thierry2017handbook}. Since $(h\hat{h}_v)$ has a limit,
the sequences $(h_v^{(n)})$ and $(\hat{h}_v^{(n)})$ converges to the same limit:
$\lim_{n \to \infty} h_v^{(n)} = \lim_{n \to \infty} \hat{h}_v^{(n)}$.
\end{proof}
\begin{theor}
\label{thm:coreconv}
 If the local coreness-constraint is satisfied for all nodes $v \in V$ at the terminal iteration, the corrected $h$-index at the terminal iteration $\hat{h}_v^{(\infty)}$ satisfies: $\hat{h}_v^{(\infty)} = c(v)$.
\end{theor}
\eat{
\spara{Proof sketch.} Given any $v \in V$, we show that $c(v) \geq \hat{h}_v^{\infty}$, followed by showing $\hat{h}_v^{\infty} \geq c(v)$.
The former holds since it can be proved that $v \in \hat{h}_v^{\infty}$-core. However, by definition of core-number, $c(v)$ is the largest integer
for which $v \in c(v)$-core. Thus, $c(v) \geq \hat{h}_v^{\infty}$. Next, $\hat{h}_v^{\infty} \geq c(v)$ can be proved by induction on the number
of iteration.
}
\begin{proof}
Given any $v \in V$, we first show that $c(v) \geq \hat{h}_v^{\infty}$, followed by showing $\hat{h}_v^{\infty} \geq c(v)$.

Let us construct $N' \subseteq V$ where every node $w \in N'$ has their converged value $\hat{h}_{w}^{\infty} \geq \hat{h}_v^{\infty}$. Since at the terminal iteration, all such nodes $w$ satisfies LCC at their respective $\hat{h}_w^{\infty}$, clearly $w$ also satisfies LCC at any integer $k < \hat{h}_{w}^{\infty}$, including $k = \hat{h}_v^{\infty}$. By definition of $LCCSAT(\hat{h}_v^{\infty})$ (for $w$), there exists a corresponding subhypergraph $H^+(w) = (N^+(w),E^+(w))$ such that any $u \in N^+(w)$ satisfies $\hat{h}_u^{\infty} \geq \hat{h}_v^{\infty}$ and $\abs{N^+(w)} \geq \hat{h}_v^{\infty}$. Construct a new hypergraph $H^+[N'] = (\cup_{w} N^+(w), \cup_{w} E^+(w))$. Any node $u \in H^+[N']$ satisfies $\hat{h}_u^{\infty} \geq \hat{h}_v^{\infty}$ by construction of $H^+[N']$. As a result, for such $u$, the following holds:
$u \in N'$ (by construction of $N'$), $u$'s corresponding $N^+(u) \subseteq \cup_{w} N^+(w)$ (by construction of $H^+[N']$), and $\abs{N^+(u)} \geq \hat{h}_v^{\infty}$ (by definition of LCC). Hence every node $u$ in $H^+[N']$ has at least $\hat{h}_v^{\infty}$ neighbors in $H^+[N']$.

Since $\hat{h}_v^{\infty}$-core is the maximal subhypergraph where every node has at least $\hat{h}_v^{\infty}$ neighbors in that subhypergraph, $H^+[N']$ is a subhypergraph of $\hat{h}_v^{\infty}$-core: $H^+[N'] \subseteq \hat{h}_v^{\infty}$-core. Notice that $v \in H^+[N']$. By definition of core-number, $c(v)$ is the largest integer for which $v \in c(v)$-core. Thus, $c(v) \geq \hat{h}_v^{\infty}$.

Next, we prove that $\hat{h}_v^{\infty} \geq c(v)$. Let us denote the $c(v)$-core to be $H'$. Let us define $\mathbf{LB_{H'}} = \min_{u\in H'} \abs{N_{H'}(u)}$. Since every node $u$ in $c(v)$-core has at least $c(v)$ neighbors in that core: $\abs{N_{H'}(u)} \geq c(v)$ for any $u\in H'$. Therefore, $\mathbf{LB_{H'}} \geq c(v)$. 
Since $v \in H'$ 
and LCCSAT($\hat{h}_v^{\infty}$) is true, applying~\Cref{lem:LBH}: $\hat{h}_v^{\infty} \geq \mathbf{LB_{H'}} \geq c(v)$.
\end{proof}
\begin{lem}
\label{lem:LBH}
Let $\mathbf{LB_{H'}}$ be the minimum number of neighbors of any node in a subhypergraph $H'=(V',E')$ of $H=(V,E)$. For any $v \in V'$ and integer $n\geq 0$,
$LCCSAT(\hat{h}_{v}^{(n)}) \text{ is True} \implies  \hat{h}_{v}^{(n)} \geq \mathbf{LB_{H'}}$.
\end{lem}
\begin{proof}
We prove by induction. For $n=0$ and $v\in V'$, $\hat{h}_v^{(0)} = h_v^{(0)} = \abs{N(v)} \geq \min_{u\in V'}\,\abs{N_{H'}(u)} = \mathbf{LB_{H'}}$.
Assume the statement to be true for $n=m$: $LCCSAT(\hat{h}_{v}^{(m)}) \text{ is True} \implies \hat{h}_{v}^{(m)} \geq \mathbf{LB_{H'}}$ for all $v \in V'$. We show that the statement is also true for $n=m+1$.

By definition of $\hat{h}_v^{(n)}$, since $LCCSAT(\hat{h}_v^{(m+1)})$ is true and $m+1 >0$, it follows that
$\hat{h}_v^{(m+1)} = h_v^{(m+1)}$.
But $ h_v^{(m+1)}$ can only assume two values for $m+1>0$: either $h_v^{(m+1)} = \hat{h}_v^{(m)}$ or $h_v^{(m+1)} =  \mathcal{H}\left(\hat{h}^{(m)}_{u_1},\hat{h}^{(m)}_{u_2},\ldots,\hat{h}^{(m)}_{u_t}\right)$. Here $u_1,\ldots,u_t$ are the neighbors of $v$ in $H$. In the former case, $h_v^{(m+1)} = \hat{h}_v^{(m)} \geq \mathbf{LB_{H'}}$ (by induction hypothesis). In later case,
$\mathcal{H}\left(\hat{h}^{(m)}_{u_1},\hat{h}^{(m)}_{u_2},\ldots,\hat{h}^{(m)}_{u_t}\right) \geq \mathbf{LB_{H'}}$ since there exist
at least $\mathbf{LB_{H'}}$ operands where each such operand have value at least $\mathbf{LB_{H'}}$. This is because
(1) $v \in V' \implies \mathbf{LB_{H'}} \leq \abs{N_{V'}(v)}$ (by definition of $\mathbf{LB_{H'}}$). 
Thus, the number of $u_i$'s such that $u_i \in V'$ is at least $\mathbf{LB_{H'}}$.
(2) Applying inductive hypothesis on operands $\{u_i: u_i \in V'\}$, $u_i$ satisfies LCC at  $\hat{h}^{(m)}_{u_i}$;
therefore $\hat{h}^{(m)}_{u_i} \geq \mathbf{LB_{H'}}$. Considering all cases, $\hat{h}_v^{(m+1)} \geq \mathbf{LB_{H'}}$.
\end{proof}

\spara{Time complexity.}
Assume that~\Cref{alg:localcore} terminates at iteration $\tau$ of the for-loop at Line~3. Each iteration has time-complexity $\bigO(\sum_v \abs{N(v)}(d(v)+\abs{N(v)}) + \sum_v \abs{N(v)})$, the first term is due to Lines~6-7 and the second term is due to Lines~4-5. Computing $\mathcal{H}$-operator requires hypergraph $h$-indices of $v$'s neighbors and costs linear-time: $\bigO(\abs{N(v)})$. Core-correcting $v$ requires at most $\abs{N(v)}$ iterations of while-loop (Line~1,~\Cref{alg:core_correct}), at each iteration constructing each of $E^+(v)$ and $N^+(v)$ costs $\bigO(d(v) + \abs{N(v)})$. Thus, \textbf{Local-core} has time complexity $\bigO\left(\tau *\sum_v\right.$  $\left.(d(v)\abs{N(v)} + \abs{N(v)}^2)\right)$.

\spara{Upper-bound on the number of iterations.}
To derive an upper-bound on the number of iterations $\tau$ required by \textbf{Local-core}, we define the notion of \emph{neighborhood hierarchy} for nodes in a hypergraph and prove that every node converges by the time the for-loop iterator in Line~3 (\Cref{alg:localcore}) reaches its neighborhood hierarchy.

\begin{defn}[Neighborhood hierarchy] Given a hypergraph $H$, the $i$-th neighborhood hierarchy, denoted by $N_i$, is the set of nodes that have the minimum number of neighbors in $H[V']$, where $V' = V \setminus \cup_{0\leq j<i}N_j$. Formally,
\begin{small}
\begin{equation}
    N_i = \{v: \argmin_{v \in V'} \abs{N_{H[V']} (v)} \}
\end{equation}
\end{small}
\end{defn}
\naheed{Consider the hypergraph $H$ in~\Cref{fig:intro_nbr}(a). Since \textit{Hall} has 4 neighbors, $N_0 = \{Hall\}$. After deleting \textit{Hall}, nodes \textit{Wallace} and \textit{Popiden} have the lowest number of neighbors $(4)$ in the remaining sub-hypergraph. Hence, $N_1 = \{Wallace, Popiden\}$. After deleting them, nodes in $R$ has the lowest number of neighbors $(6)$ in the remaining sub-hypergraph. Hence, $N_2 =R$. Finally, $N_3 = T$.}

\begin{lem}
\label{lem:nbrhier}
Given two vertices $v_i \in N_i$, $v_j \in N_j$ such that $i \leq j$, it holds that $c(v_i) \leq c(v_j)$.
\end{lem}
Lemma \ref{lem:nbrhier} follows from the correctness of the peeling algorithm.

\begin{theor}[individual node convergence]
\label{thm:convhier}
Given a node $v \in N_i$ in a hypergraph $H$, it holds that $\forall n \geq i,~\hat{h}^{(n)}_v = c(v)$ .
\end{theor}

\begin{proof}
We prove by induction on $i$.
For $i = 0$, $N_0$ is the set of nodes with the minimum number of neighbors in the hypergraph. Thus, $\forall v \in N_0, c(v) = \abs{N(v)}$ and by~\Cref{eqn:hhatn} (for $n=0$), we have: $h_v^{(0)} = \abs{N(v)} = c(v)$. Suppose the theorem holds for up to $i=m$, i.e., $\forall u \in N_j \mid j\leq m$ it holds that $\forall n\geq j, \hat{h}^{(n)}_u = c(u)$. Let $v \in N_{m+1}$. We need to show that $\hat{h}_v^{(m+1)} = c(v)$.

By~\Cref{thm:coreconv} and~\Cref{thm:interleave}, $\hat{h}_v^{(m+1)} \geq c(v)$ because $\hat{h}$'s are monotonically non-increasing with $c(v)$ as the limit.

Now we show that $\hat{h}_v^{(m+1)} \leq c(v)$.
By~\Cref{eqn:hhatn}, we have two cases. Case 1: ($\hat{h}_v^{(m+1)} = h_v^{(m+1)}$): In this case, $\hat{h}_v^{(m+1)} = h_v^{(m+1)} \leq \mathcal{H}(\hat{h}_{u_1}^{m},\hat{h}_{u_2}^{m},\ldots,\hat{h}_{u_t}^{m})$ by~\Cref{eqn:hn} as the minimum of two integers is not larger than either of them.
Case 2: ($\hat{h}_v^{(m+1)} < h_v^{(m+1)}$): In this case, $\hat{h}_v^{(m+1)} < h_v^{(m+1)} \leq \mathcal{H}(\hat{h}_{u_1}^{m},\hat{h}_{u_2}^{m},\ldots,\hat{h}_{u_t}^{m})$ for the same reason as in case 1. So it suffices to show that
$\mathcal{H}(\hat{h}_{u_1}^{m},\hat{h}_{u_2}^{m},\ldots,\hat{h}_{u_t}^{m}) \leq c(v)$.

To that end, let us partition the $t$ neighbors of $v$ into two sets $S'$ and $S''$ such that $S' = \{u': u' \in N(v) \land  u' \in N_j \text{ for some } j<m+1\}$ and $S'' = \{u'': u'' \in N(v) \land  u'' \in N_j \text{ for some } j\geq m+1\}$. Let us order the terms in $\mathcal{H}$: $\mathcal{H}(\hat{h}^{m}_{u'_1},\ldots,\hat{h}^{m}_{u'_x},\hat{h}^{m}_{u''_1},\ldots,\hat{h}^{m}_{u''_y})$ where $u'_x \in S'$ and $u''_y \in S''$. (1): $\forall u'_x \in S'$, by induction hypothesis and~\Cref{lem:nbrhier}, $h_{u'_x}^{m} = c(u'_x) \leq c(v)$ because $v \in N_{m+1}$ and $u'_x \in N_j$ such that $j<m+1$ by definition of $S'$. (2): $y = \abs{S''}$, and $S''$ contains all neighbors of v in $H' = H[\cup_{j \geq m+1} N_j]$.
$H'$ is a subhypergraph where every node, including $v \in N_{m+1}$, has number of neighbors at least $\abs{S''}$. Hence, every node in $H'$ has their core-number at least $\abs{S''}$.
Consequently, $c(v) \geq \abs{S''} =  y$. Combining the two results (1) and (2), the first $x$ terms in $\mathcal{H}()$ are all $<c(v)$, and even if the next $y$ terms could be $\geq c(v)$, there are less than $c(v)$ of them. Hence, $\mathcal{H}(\hat{h}_{u_1}^{m},\hat{h}_{u_2}^{m},\ldots,\hat{h}_{u_t}^{m}) \leq c(v)$ which completes the proof.
\end{proof}
 
%The proof by induction on $i$ is given in our extended version \cite{full}.
Clearly, 
$\tau$ is at most the largest neighborhood hierarchy of the nodes. 
\section{Optimization and Parallelization of the local-core Algorithm}
\label{sec:optimization}
We propose four optimizations to improve the efficiency of \textbf{Local-core} (\S \ref{sec:local}).
~\Cref{alg:localcoreopt} presents the pseudocode for the optimized algorithm, \textbf{Local-core(OPT)},
where all four optimizations are indicated. Optimization-I adopts sparse representations to efficiently
evaluate neighborhood queries in hypergraphs, while Optimization-II consists of three implementation-specific
improvements to efficiently perform \textbf{Core-correction}.
{\em Optimizations-I and II have not been
used in earlier core-decomposition works for both graphs and hypergraphs}.
Our Optimizations-III and IV are motivated by~\cite{LiuZHX21},
where similar optimizations are proposed for graph $(k,h)$-core decomposition to improve convergence
and eliminate redundant computations, respectively. {\em However such optimizations have not been adopted in earlier hypergraph-related works}.

\spara{Optimization-I (Compressed hypergraph representation).}
\textbf{Local-core} makes two primitive neighborhood queries on hypergraph structures:
neighbors enumeration (for $h$-index computation, Lines 2 and 5, Algorithm~\ref{alg:localcore}) and incident-hyperedges enumeration
(during \textbf{Core-correction}, Line 2, Algorithm \ref{alg:core_correct}). A na\"ive implementation keeps a $\abs{V} \times \abs{V}$
matrix for neighbors counting queries and a $\abs{V} \times \abs{E}$ matrix for incident-hyperedge queries.
However, storing such matrices in memory is expensive and unnecessary for large hypergraphs since these matrices are sparse in practice. Hence, it is imperative to adopt a compressed sparse representation,
e.g., {\em Compressed sparse row} (\textsf{CSR}) for these matrices.
For example, in \textsf{CSR} representation we use two arrays $\mathcal{F}$ and $\mathcal{N}$ for storing neighbors.
For all nodes $v \in V$, $\mathcal{N}[v]$ stores the starting index in $\mathcal{F}$
containing neighbors of node $v$. To facilitate neighborhood enumeration query, neighbors of node $v$ are stored at contiguous indices $\mathcal{N}[v]$, $\mathcal{N}[v]+1$, $\ldots$, $\mathcal{N}[v+1] - 1$ in $\mathcal{F}$. Any neighbor $u$ connected to $v$ by a hyperedge has their own set of neighbors $N(u)$ stored at contiguous indices $\mathcal{N}[u]$, $\mathcal{N}[u]+1$,$\ldots$, $\mathcal{N}[u+1]-1$ in $\mathcal{F}$. Hence for any node $v \in V$,
$\mathcal{N}[v+1]-\mathcal{N}[v]$ gives the number of neighbors $\abs{N(v)}$ in $\bigO(1)$ time and neighborhood enumeration takes $\bigO(\abs{N(v)})$ time; in contrast to the matrix representations which take $\bigO(\abs{V})$ time for both neighbors-enumeration and incident-hyperedges enumeration queries.
\textsf{CSR} representations, due to their contiguous memory locations,
can also exploit spatial memory access patterns in the \textbf{Local-core} algorithm.
\begin{algorithm}[!tb]
\begin{algorithmic}[1]
\caption{\label{alg:localcoreopt}\small Optimized local algorithm: \textbf{Local-core(OPT)}}
\scriptsize
\Require Hypergraph $H = (V,E)$
\Ensure Core-number $c(v)$ for each node $v \in V$
\State Construct CSR representations /* Opt-I */
\ForAll {$v \in V$}
    \State Compute $\mathbf{LB}(v)$ /* local lower-bounds for Opt-IV */
    \State $c(v) \gets \hat{h}_{v}^{(0)} \gets h_{v}^{(0)} \gets \abs{N(v)}$ /* core-estimate c(v) initialized for Opt-III*/
\EndFor

\ForAll {$n = 1,2,\ldots, \infty$}
    \ForAll {$v \in V$}
        \If{$h_v^{(n)} > LB(v)$} /* Opt-IV*/
            \State $c(v) \gets h_v^{(n)} \gets \min\left(\mathcal{H}(\{c(u) : u \in N(v)\}), c(v) \right)$ /* Opt-III*/
        \EndIf
    \EndFor
    \ForAll {$v \in V$}
        \If{$h_v^{(n)} > LB(v)$ } /* Opt-IV  */
            \State $c(v) \gets \hat{h}_v^{(n)} \gets $ \textbf{Core-correction ($v$, $h_v^{(n)}$, $\mathcal{H}$)} /* Opt-II \& Opt-III*/
        \EndIf
    \EndFor
    \If {$\forall v, \hat{h}_v^{(n)} == h_v^{(n)}$}
       \State \textbf{Terminate Loop}
    \EndIf
\EndFor
\State Return $c$
\end{algorithmic}
\end{algorithm}

\spara{Optimization-II (Efficient \textbf{Core-correction} and LCCSAT).}
We design three optimization methods for more efficient \textbf{Core-correction}.

\emph{Hyperedge-index for efficient $E^+$ computation:} In Line~2 of \textbf{Core-correction} (\Cref{alg:core_correct}),
we check if $h_u^{(n)} \geq h_v^{(n)}$ for every node $u \in e$ such that $e \in Incident(v)$.
This computation incurs $\bigO(d(v)+ \abs{N(v)})$ at every while-loop (Line~1) of \textbf{Core-correction}. Since the number of while-loop iterations is $\abs{N(v)}$ in the worst-case, an inefficient $E^+$ computation contributes $\bigO(\abs{N(v)}(d(v)+ \abs{N(v)}))$ to the total cost of \textbf{Core-correction}.
We reduce this cost to $\bigO(\abs{N(v)}d(v))$ by maintaining an
index $E_e$ for hyperedges. $E_e^{(n)}$ records for every hyperedge
$e \in E$ the minimum of $h$-indices of its constituent nodes ($\min_{u \in e} h_u^{(n)}$).
We compute $E^+(v)$ by traversing only the incident hyperedges whose $E_e^{(n)} \geq h_v^{(n)}$.
Storing $E_e$ for all hyperedges costs $\bigO(\abs{E})$ space and constructing the indices costs $\bigO(\sum_{e\in E} \abs{e})$ time.
However, hyperedge-indices are constructed only once before every iteration in \textbf{Local-core} (Line 3, \Cref{alg:localcore}).
Moreover, hyperedge-index helps efficiently compute $N^+$ as follows.

\emph{Dynamic programming for efficient $N^+$ computation:}
The while loop in the \textbf{Core-correction} procedure computes $k = h_v^{(n)},h_v^{(n)}-1,\ldots,\hat{h}_v^{(n)}$; and for
every $\hat{h}_v^{(n)} \leq k \leq h_v^{(n)}$, it recomputes $E^+(v)$ and $N^+(v)$ until returning $\hat{h}_v^{(n)}$ as output. Without any optimization, the cost of computing $N^+(v)$ at every while-loop iteration $k$ (Line~1) is $\bigO(d(v) + \abs{N(v)})$. Since the number of while-loop iterations is $(h_v^{(n)}-\hat{h}_v^{(n)}) \leq \abs{N(v)}$, an inefficient $N^+$ computation contributes $\bigO(\abs{N(v)}(d(v)+ \abs{N(v)}))$ to the total cost of \textbf{Core-correction}.

We reduce this cost to $\bigO(d(v)+ \abs{N(v)})$ by
constructing an
index $B$ such that $B[h_v^{(n)}]$ records the set of incident hyperedges whose hyperedge-index $E_{e}^{(n)} \geq h_v^{(n)}$,
and for every $k < h_v^{(n)}$, $B[k]$ records the incident hyperedges whose hyperedge-index $E_{e}^{(n)} = k$.

Let us denote $E^+(v)$ and $N^+(v)$ at $k$ as $E^+(v,k)$ and $N^+(v,k)$, respectively.
Exploiting the index structure $B$, we have the following dynamic programming paradigm for efficiently computing $N^+(v,k)$.
\begin{small}
\begin{equation}
\label{eqn:graphhn}
        N^+(v,k) =
        \begin{cases}
          \cup_e B[k] & k = h_v^{(n)}\\
          N^+(v,k+1) \cup (\cup_e B[k]) & \hat{h}_v^{(n)} \leq k < h_v^{(n)}
      \end{cases}
\end{equation}
\end{small}

Instead of traversing all incident hyperedges at every $\hat{h}_v^{(n)} \leq k \leq h_v^{(n)}$, we only traverse hyperedges at $B[k]$ to compute $N^+(v,k)$.
Since the indices $B[k]$ are mutually exclusive, each incident hyperedge is traversed at most once during the entire \textbf{Core-correction} procedure.
The storage cost of $B$ is $\bigO(h_v^{(n)} + d(v))$,
as there are at most $h_v^{(n)}$ keys in $B$ and exactly $d(v)$ hyperedges are stored in $B$.
Due to efficient $E^+$ and $N^+$ computation, core-correction costs $\bigO(\abs{N(v)}d(v) + d(v)+ \abs{N(v)})$ instead of $\bigO(\abs{N(v)}d(v) + \abs{N(v)}^2)$.

\emph{Efficient LCCSAT computation:} 
We return \emph{True} upon finding the first hyperedge $e \in Incident(v)$ adding which to $E^+(v)$ causes $\abs{N^+(v)} \geq h_v^{(n)}$. 
Adding subsequent incident hyperedges to $E^+(v)$ increases $\abs{N^+(v)}$ more without affecting $LCCSAT(k)$=\emph{True}.

\spara{Optimization III (Faster convergence).}
In Line~8 (\Cref{alg:localcoreopt}), instead of computing $h_v^{(n)}$ as per~\Cref{eqn:hn}, we compute a smaller value $g_v^{(n)}$ by replacing some operands $\hat{h}_{u_i}^{(n-1)}$ in the $\mathcal{H}$-operator with $h_{u_i}^{(n)} \leq \hat{h}_{u_i}^{(n-1)}$(due to \Cref{thm:interleave}). This leads to faster convergence since $g_v^{(n)} \leq h_v^{(n)}$, because lowering a few arguments may cause the output
of $\mathcal{H}()$ to decrease further, e.g., $\mathcal{H}(1,1,2,2) = 2$, whereas $\mathcal{H}(1,1,1,2) = 1$. Computing $g_v$ does not affect the correctness because the new interleave sequence $(h_v, g_v,\hat{g}_v)$ is also monotonically non-increasing and lower-bounded by 0. This is because $h_v^{(n)} \geq g_v^{(n)}$ and \textbf{Core-correction} never increases $g_v^{(n)}$, so $g_v^{(n)} \geq \hat{g}_v^{(n)}$. Thus, $(h_v^{(n)})$, $(g_v^{(n)})$, and $(\hat{g}_v^{(n)})$ have the same limit as that in~\Cref{thm:interleave}: $\hat{h}^{(\infty)}$. Finally, from~\Cref{thm:coreconv}, it follows that $\hat{h}^{(\infty)}=c(v)$ holds despite this optimization.

\spara{Optimization-IV (Reducing redundant $\mathcal{H}$ computations and LCCSAT checks).}
We use local lower-bound $\mathbf{LB}(v)$ on core-numbers (\Cref{lem:localLB}) to reduce the number of
$h$-index computations and core corrections. At some iteration $n$, if $\hat{h}_v^{(n)}$=$\mathbf{LB}(v)$,
it ensures that $c(v)$=$\mathbf{LB}(v)$.
Thus, $h$-index for $v$ would no longer reduce in future iterations;
otherwise, $c(v) = \hat{h}_v^{(\infty)} < \hat{h}_v^{(n)} = \textbf{LB}(v)$,
which is a contradiction. Hence, it must be that  $c(v) = \hat{h}^{(\infty)} = \hat{h}^{(n)}$.
We need not compute the $h$-index and core corrections for $v$ at future iterations.

Note that Liu et al.~\cite{LiuZHX21} determine the redundancy of $\mathcal{H}$ computation for node $v$ based on the convergence of both $h_v^{(n)}$ and $h_u^{(n)}$ of neighbors $u \in N(v)$. This does not work for our problem, because 
even if a node $v$ and its neighbors' $h$-indices have converged, node $v$ may still need core correction (e.g., Example~\ref{ex:wronglocal}).
\begin{algorithm}[!tb]
\begin{algorithmic}[1]
\caption{\label{alg:kdcore}\small \textbf{Local-core+Peel for (k,d)-core decomposition}}
\scriptsize
\Require Hypergraph $H = (V,E)$
\Ensure All non-empty $(k,d)$-cores of $H$; $1\leq k \leq k_{max}$, $1\leq d \leq d_{max}$
\State Initialize $\mathbb{C} = \{ C(k,d) \gets \phi\quad\forall k_{min} \leq k \leq k_{max}, \, d_{min} \leq d \leq d_{max}\}$
\State Compute neighborhood based core-number $c(v)$ $\forall v \in V$ $\qquad$ /* \textbf{Local-core(OPT)}*/
\ForAll {nbr-based core-number $k = 1,2,\ldots,k_{max}=\max_{v\in V}c(v)$}
    \State Construct nbr $k$-core $H[V_k]: V_k = \{v: c(v) \geq k\}$
    
    \ForAll{$u \in V_k$}
        \State Initialize $B[deg_{V_k}(u)] \gets B[deg_{V_k}(u)] \cup \{u\}$
    \EndFor
    \ForAll{$d = 1,2,\ldots,\max_{u\in V_k}(deg_{V_k}(u))$}
        \While{$B[d] \neq \phi$}
            \State Remove a node $v$ from $B[d]$.
            \State $C(k,d) = C(k,d) \cup \{v\}$
            \ForAll{$u \in N_{V_k}(v)$}
                \If{$\abs{N_{V_k\setminus \{v\}}(u)} \geq k$}
                \State Move $u$ to $B[\max\left(deg_{V_k\setminus \{v\}}(u), d\right)]$
                \Else
                \State  Move $u$ to $B[d]$
                \EndIf
            \EndFor
            \State  $V_k \gets V_k \setminus \{v\}$
        \EndWhile
    \EndFor    
\EndFor
\State Return $\mathbb{C}$
\end{algorithmic}
\end{algorithm}

\spara{Efficiency comparison of Local-core(Opt) and Peel.}
The terms $d_{nbr}$ and $d_{hpe}$ in the time complexity of \textbf{Peel} (\S \ref{sec:peeling})
are upper bounds of $|N(v)|$ and $d(v)$, respectively, $\forall v \in V$.
Hence, the complexity of \textbf{Peel} is comparable to $\bigO(\sum_{v \in V} |N(v)|d(v) + N(v)^2)$,
which is same as the complexity of one round of \textbf{Local-core} (\S \ref{sec:local}).
Thus, \textbf{Local-core} is slower than \textbf{Peel}.
However, our optimizations reduce its complexity significantly.
For instance, by efficiently computing $E^+$ and $N^+$, optimization-II reduces the complexity to $\bigO(\tau \sum_{v} (d(v)\abs{N(v)}+d(v)+\abs{N(v)}))
= \bigO(\tau \sum_{v} d(v)\abs{N(v)}) = \bigO(\tau \abs{V} d_{hpe} d_{nbr})$.
Hence, \textbf{Local-core(Opt)} is at least $\frac{d_{nbr}+d_{hpe}}{\tau d_{hpe}}=(\frac{d_{nbr}}{\tau d_{hpe}}+\frac{1}{\tau})$-times faster than
\textbf{Peel} due to Optimization-II.
A node usually has more neighbors than its degree in large-scale, real-world hypergraphs (e.g., {\em aminer}, {\em dblp} in~\Cref{table:dataset}),
resulting in $d_{nbr}>d_{hpe}$. Optimizations III and IV further reduce $\tau$, making the ratio $>1$, and thus \textbf{Local-core(Opt)} is
much faster than \textbf{Peel} in our experiments.

\spara{Parallelization of \textbf{Local-core}.}
%\begin{algorithm}[!tb]
%\begin{algorithmic}[1]
%\caption{\label{alg:localcoreP}\small Parallel Local algorithm: \textbf{Local-core(P)}}
%\scriptsize
%\Require Hypergraph $H = (V,E)$
%\Ensure Core-number $c[v]$ for each node $v \in V$
%\State Construct CSR representations
%\ParFor {$v \in V$}
%    \State $c[v] \gets \hat{h}_{v}^{(0)} \gets h_{v}^{(0)} \gets \abs{N(v)}$
    % \State Initialize $hcount_v \gets (0,\ldots,0)_{\abs{N(v)}}$
%\EndParFor
%
%\ForAll {$n = 1,2,\ldots, \infty$}
%    \ParFor {$v \in V$}
%         \State $c[v] \gets h_v^{(n)} \gets \min\left(\mathcal{H}(\{c[u] : u \in N(v)\}), c[v] \right)$
%        \State $c[v] \gets \hat{h}_v^{(n)} \gets $ \textbf{Core-correction ($v$, $h_v^{(n)}$, $H$)}
%    \EndParFor
%    \If {$\forall v, \hat{h}_v^{(n)} == h_v^{(n)}$}
%        \State \textbf{Terminate Loop}
%    \EndIf
%\EndFor
%\State return $c$
%\end{algorithmic}
%\end{algorithm}
We propose parallel 
\textbf{Local-core(P)} algorithm
following the shared-memory, data parallel programming paradigm, which further improves efficiency.
The algorithm partitions nodes into $T$ partitions, where $T$ is the number of threads.
Each thread is responsible for computing core-numbers of nodes in its own partition.
To improve load-balancing, we adopt the longest-processing-time-first scheduling approach~\cite{graham1969bounds}
such that the aggregated number of neighbors of nodes in different threads are roughly the same.
\textbf{Local-core(P)} has three primary differences compared to sequential  \textbf{Local-core(OPT)} (Algorithm~\ref{alg:localcoreopt}).

{\bf First}, at iteration $0$ (Line 3), every thread initializes hypergraph $h$-indices for its allocated nodes 
in parallel. Concurrent computation of $\abs{N(v)}$ and $\abs{N(u)}$ for nodes allocated to different
threads requires concurrent reads to the CSR representation. Since the CSR representation does not change across queries,
both queries will produce the same result as their sequential counterparts.
{\bf Second}, at subsequent iterations ($n>0$), every thread computes $h_v^{(n)}$
and the corrected value $\hat{h}_v^{(n)}$ for its allocated nodes 
in parallel (Lines 6-11).
Interestingly, our algorithm still terminates with the correct output because
none of our previous theoretical results rely on any particular computation order of nodes in the same iteration. The computation order only affects the number
of iterations required for convergence. 
{\bf Finally}, all threads keep their respective core-correction counter, which counts the number of allocated nodes that have been core-corrected in a given round. Once all the counters report $0$, \textbf{Local-core(P)} terminates.
\begin{table}
\scriptsize
\vspace{-3mm}
\caption{\label{table:dataset} \footnotesize Datasets: $|V|$ \#nodes, $|E|$ \#hyperedges, $d(v)$ (mean) degree of a node,
$|e|$ (mean) cardinality of a hyperedges, $|N(v)|$ (mean) \#neighbors per node}
\vspace{-3mm}
\begin{tabular}{c||c|c|c|c|c|c}
%\hline
 & hypergraph & $\abs{V}$ & $\abs{E}$ & {$d(v)$} & {$\abs{e}$} & {$\abs{N(v)}$} \\ \hline \hline
\multirow{3}{*}{Syn.} & {\em bin4U} & 500 & 12424 & 99.4$\pm$8.5  & 4$\pm$0 & 225.3$\pm$15.5  \\ \cline{2-7}
 & {\em bin3U} & 500 & 16590 & 99.5$\pm$8 & 3$\pm$0  & 164.1$\pm$11.6  \\  \cline{2-7}
  & {\em pref3U} & 125329 & 250000 & 5.9$\pm$915.9  & 3$\pm$0 &  4.5$\pm$412.4  \\ \hline
\multirow{5}{*}{Real} & {\em enron} & 4423 & 5734 & 6.8$\pm$32 & 5.2$\pm$5 & 25.3$\pm$44 \\ \cline{2-7}
 & {\em contact} & 242 & 12704 & 127$\pm$55.2  & 2.4$\pm$0.5  & 68.7$\pm$26.6  \\ \cline{2-7}
 & {\em congress} & 1718 & 83105 & 426.2$\pm$475.8 & 8.8$\pm$6.8 & 494.7$\pm$248.6\\ \cline{2-7}
 & {\em dblp} & 1836596 & 2170260 & 4$\pm$11.6 & 3.4$\pm$1.8 & 9$\pm$21.4  \\ \cline{2-7}
 & {\em aminer} & 27850748 & 17120546 & 2.3$\pm$5 & 3.7$\pm$2.6 & 8.4$\pm$24.1\\ %\hline
\end{tabular}
\vspace{-6mm}
\end{table}
% 
%
%\vspace{-2mm}
\section{Extension to 
(neighborhood, degree)-core decomposition}
\label{sec:kdcore}

We propose a new hypergraph-core model, \textbf{(neighborhood, degree)-core}, or $(k,d)$-core in short, by considering both neighborhood and degree constraints. 
$(k,d)$-core eliminates problems such as the nbr-$k$-core of a node can be very large if it is involved in one large-size hyperedge. Notice that the $d$-value of the $(k,d)$-core for such a node would be small, differentiating it from other nodes having both higher degrees and a higher number of neighbors. We demonstrate the superiority of $(k,d)$-core in diffusion spread 
(\S \ref{sec:inf_applications}).

\spara{$(k,d)$-core.} Given integers $k, d>0$, the $(k,d)$-core, denoted as $H[V_{(k,d)}] = (V_{(k,d)}, E[V_{(k,d)}])$ is the maximal subhypergraph s.t. every node $u \in V_{(k,d)}$ has at least $k$-neighbors and degree $\geq d$ in the strongly induced subhypergraph $H[V_{(k,d)}]$.

$(k,d)$-core decomposition is difficult: {\bf (1)} Unlike neighborhood-based core decomposition that defines a total hierarchical order across different cores, $(k,d)$-core decomposition forms a core lattice defining only partial containments: 
$(k+i,d+j)$-core is 
contained in $(k,d)$-core $\forall k,d > 0; i,j \geq 0$. {\bf (2)} Let $k_{max}$ and $d_{max}$ denote the maximum core-numbers via
neighborhood- and degree-based core decompositions, respectively. We can have up to $\bigO(k_{max} \cdot d_{max})$ different $(k,d)$-cores. 
Due to such challenges, we keep the problem of designing a holistic local algorithm for $(k,d)$-core decomposition open. 
Following a similar notion from multi-layer core decomposition \cite{GalimbertiBGL20}, we propose a hybrid local and peeling approach. 

\spara{Local-core+Peel algorithm (Algorithm~\ref{alg:kdcore}) for $(k,d)$-core decomposition.}
First, we initialize the $(k,d)$-cores as empty (Line~1) and compute neighborhood-based core decomposition of all nodes using \textbf{Local-core(OPT)} (Line~2).
At each outer for-loop (Line~3) iteration $k$, we construct the nbr $k$-core ($H[V_k]$) (Line~4) and peel the nodes in $V_k$ to find various $(k,d)$-cores with the same $k$ but different $d$. 
To do so, we use a vector of lists $B$ to define the initial peeling order of $V_k$ based on their degrees in cell $H[V_k]$ (Lines~5-6) similar to algorithm \textbf{Peel}. 
We traverse $B$ (Line~7) in ascending order of values $d$ and peel nodes from the current cell $B[d]$ until $B[d]$ is empty (Line~8). Whenever a node is removed from $B$ (Line~9), it is assigned to $(k,d)$-core $C(k,d)$ (Line~10), and its neighbors in $H[V_k]$ are moved to appropriate cells in $B$ (Lines 11-15). 
After we assign all the nodes in $V_k$ to various $(k,d)$-cores for varying $d$, we proceed to do the same for $V_{k+1}$ at the next iteration of the outer for-loop (Line~2). 

\eat{
\spara{Local-core+Peel algorithm for $(k,d)$-core decomposition.}
We initialize all $(k,d)$-cores as empty and compute neighborhood-based core decomposition of all nodes using \textbf{Local-core(OPT)}. Our algorithm (given in \cite{full}) has two for-loops:
At each outer for-loop iteration $k$, we construct the nbr-$k$-core ($H[V_k]$) and peel the nodes in $V_k$ to find various $(k,d)$-cores with the same $k$ but different $d$. 
To do so, we use a vector of lists $B$ to define the initial peeling order of $V_k$ based on their degrees in cell $H[V_k]$ similar to algorithm \textbf{Peel}. 
We traverse $B$ in ascending order of values $d$ and peel nodes from the current cell $B[d]$ until $B[d]$ is empty. Whenever a node is removed from $B$, it is assigned to proper $(k,d)$-core, and its neighbors in $H[V_k]$ are moved to appropriate cells in $B$. 
After we assign all the nodes in $V_k$ to various $(k,d)$-cores for varying $d$, we proceed to do the same for $V_{k+1}$ at the next iteration of the outer for-loop.  
}
% \input{4-applications}
%
%\vspace{-5mm}
\section{Empirical Evaluation}
\label{sec:experiments}
We empirically evaluate the performance of our algorithms on four synthetic and five real-world datasets
(Table~\ref{table:dataset}). We implement our algorithms
in GNU C++11 and OpenMP API version 3.1. All experiments are conducted on a server with 128 AMD EPYC 32-core processors and 256GB RAM.
\textbf{Our code and datasets are at \cite{code}}. 

%\vspace{-1mm}
\spara{Datasets.}
Among synthetic hypergraphs, {\em bin4U} and {\em bin3U} are $4$-uniform and $3$-uniform hypergraphs, respectively,
generated using state-of-the-art hypergraph configuration model~\cite{arafat2020construction}.
%The node degrees in both hypergraphs are such that the number of nodes with degree $i$ is $\binom{500}{i} (0.5)^i (0.5)^{500-i}$ and the average degree is $100$ approximately.
%The node degrees in both hypergraphs follow a binomial distribution with parameters $n = 500$ and $p = 0.2$.
{\em pref3U} is a $3$-uniform hypergraph generated using the hypergraph preferential-attachment model~\cite{avin2019random}
with parameter $p = 0.5$, %and $k = 3$,
where $p$ is the probability of a new node being preferentially attached to existing nodes in the
hypergraph. %and $k$ is the uniform cardinality.
The node degrees in \emph{pref3U} approximately follows a power-law distribution with exponent $ = 2.2$.
%We consider four real-world datasets: \emph{enron}, \emph{contact}, \emph{congress}, and \emph{dblp}.
Among real-world hypergraphs, {\em enron} is a hypergraph of emails, where each email correspondence
is a hyperedge and users are nodes \cite{benson2018simplicial}. We derive the {\em contact}
(in a school) dataset from a graph where each maximal clique is viewed as a hyperedge and
individuals are nodes~\cite{benson2018simplicial}. In the {\em congress} dataset,
nodes are congress-persons and each hyperedge comprises of sponsors and co-sponsors (supporters)
of a legislative bill %put forth in both the House of representatives and the senate
\cite{benson2018simplicial}.
In {\em dblp}, nodes are authors and a hyperedge consists of authors in a publication recorded on DBLP~\cite{benson2018simplicial}.
Similarly, {\em aminer} consists of authors and publications recorded on Aminer \cite{aminer_data}.
%({\footnotesize\url{https://www.aminer.org/open-academic-graph}}).
%
%
%
\begin{figure*}[tb!]
    \centering
     \subfloat[\footnotesize Ineffectiveness of graph $h$-index]{\includegraphics[scale=0.2]{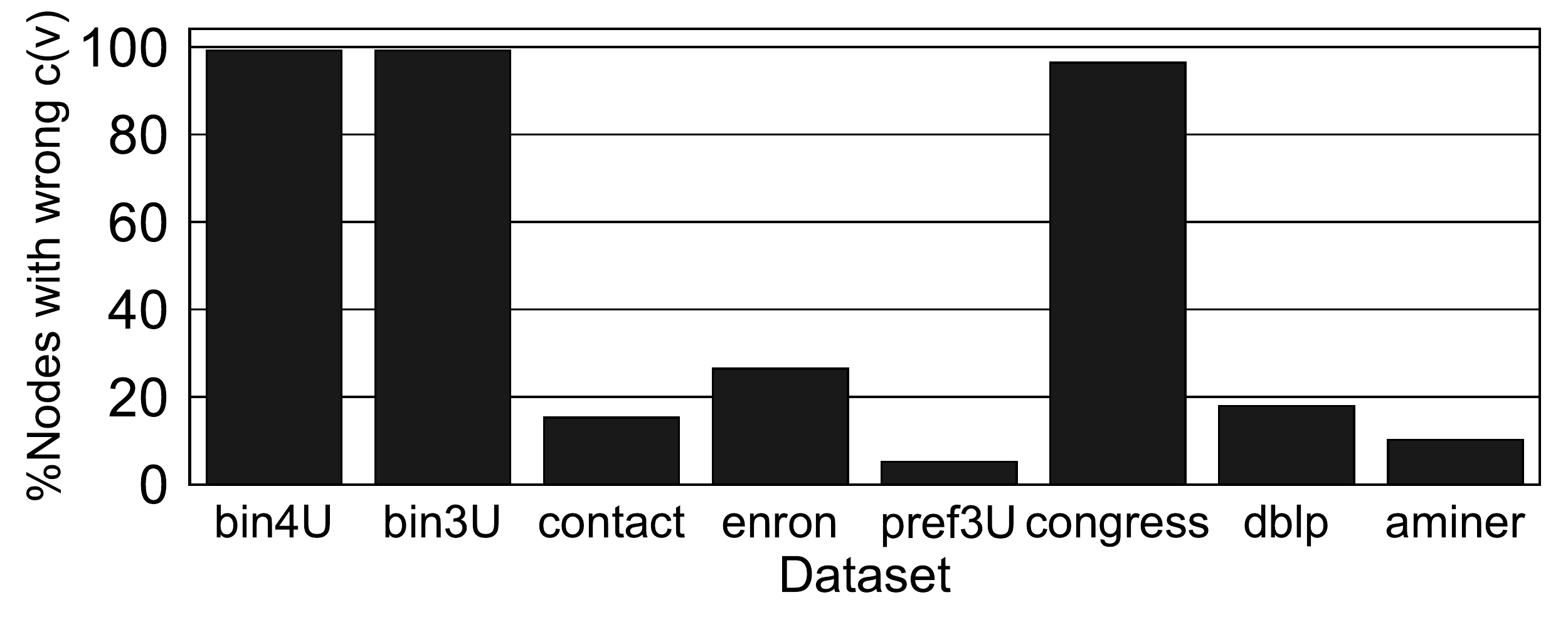}}
    % \subfloat[Convergence of \textbf{Local-core}]{\includegraphics[scale=0.2]{expt/Convergence.pdf}}
    \subfloat[\footnotesize Average error of hyp. and graph $h$-index on {\em Enron}]{\includegraphics[scale=0.2]{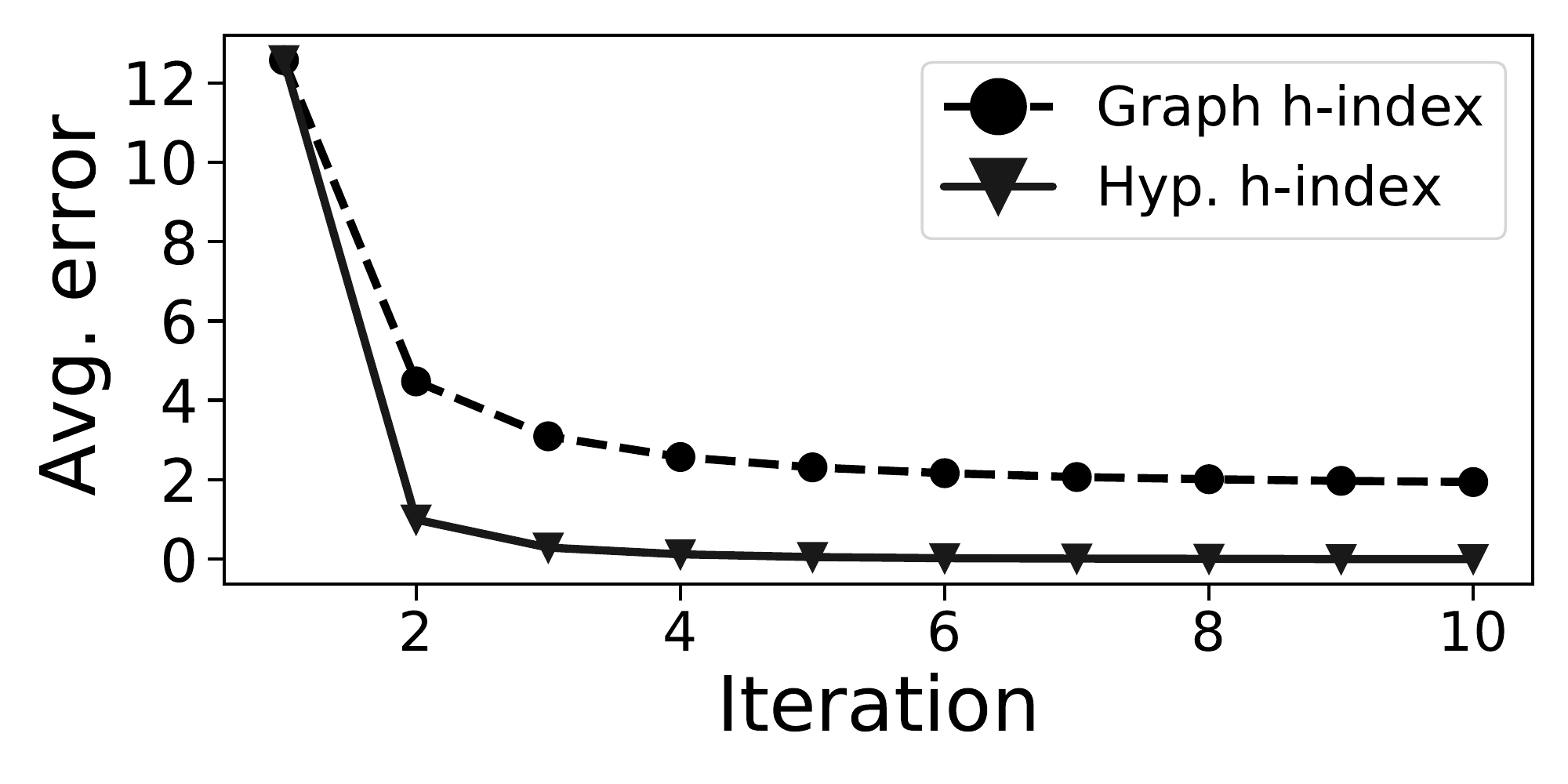}}
    \vspace{1mm}
    \subfloat[\footnotesize Average error of hyp. and graph $h$-index on {\em bin4U}]{\includegraphics[scale=0.2]{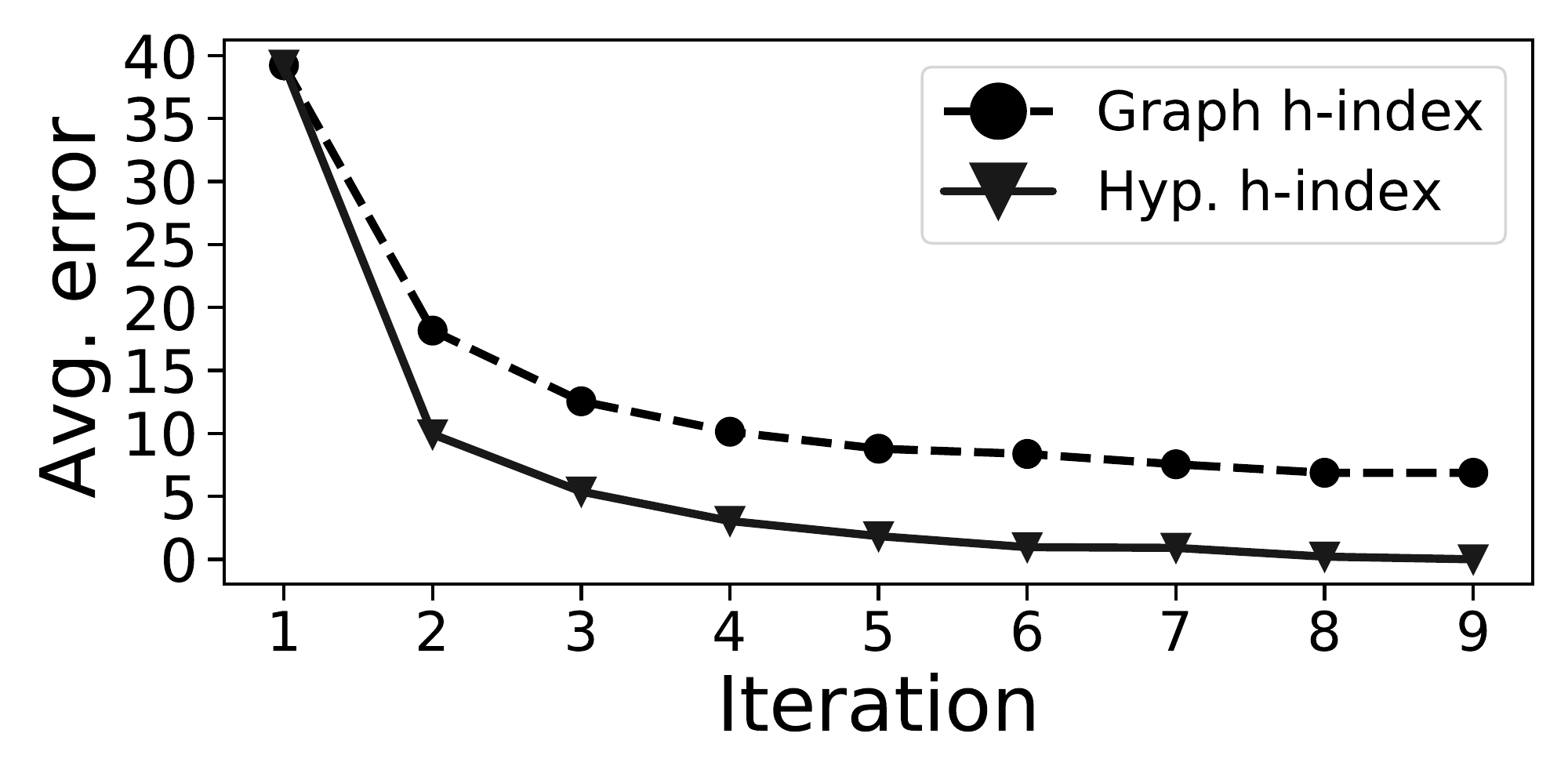}}
    \subfloat[\footnotesize Convergence of \textbf{Local-core}]{\includegraphics[scale=0.16]{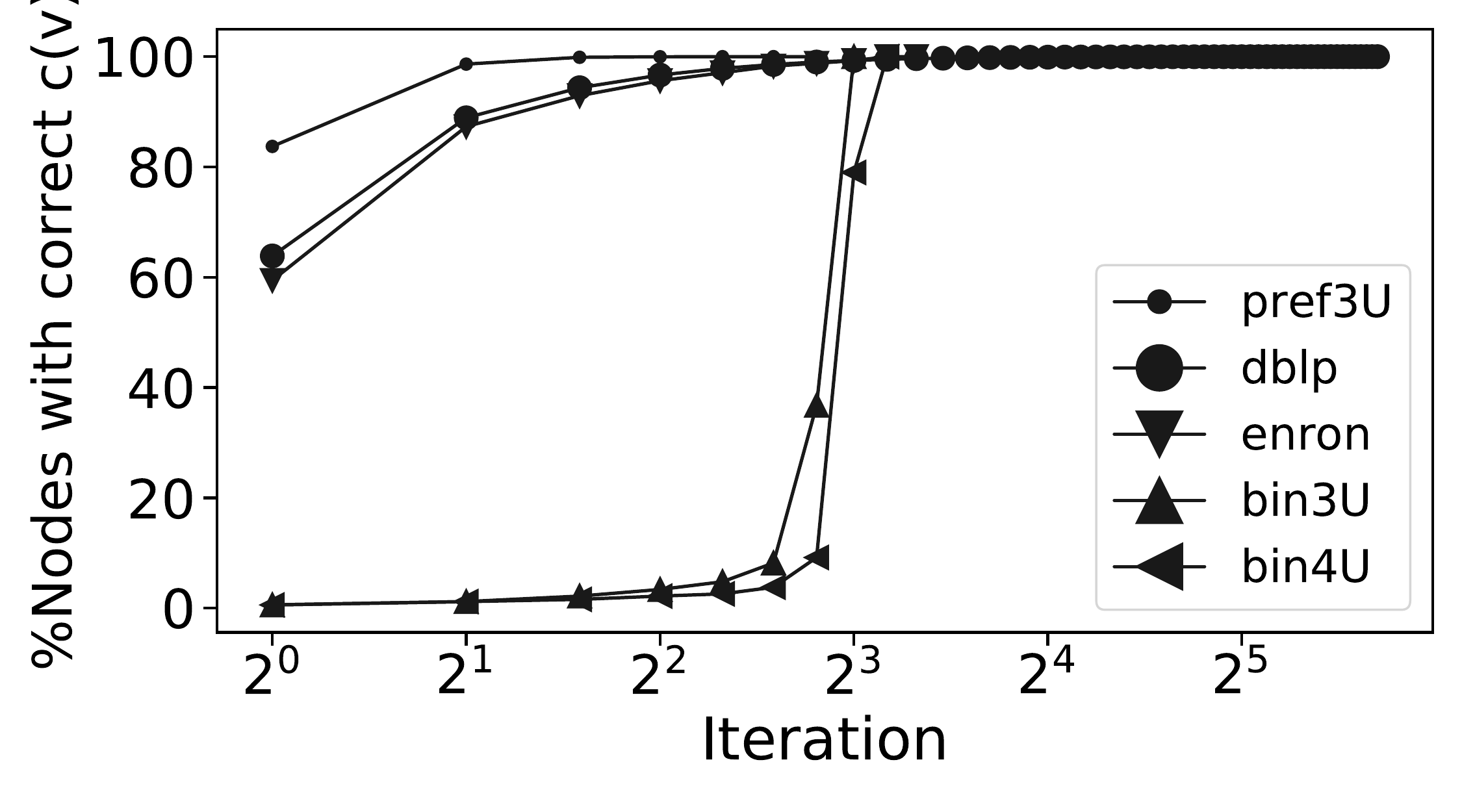}}
    \vspace{-4mm}
    \caption{\footnotesize Effectiveness evaluation of hypergraph $h$-index and \textbf{Local-core}}
    % \caption{\small Impact of Optimizations on \textbf{\% wrong core-numbers of Local-core without core-correction.}}
    \label{fig:wr_localcore}
    \vspace{-5mm}
\end{figure*}
\begin{figure*}[tb!]
\vspace{1.5mm}
    \centering
    \fbox{
    \includegraphics[scale=0.52,trim={2.4cm 2cm 1.5cm 2cm},clip]{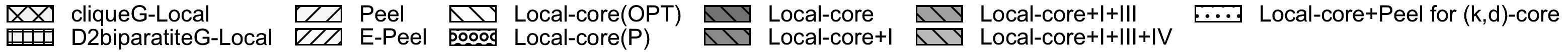}
    }
    \vspace{-4mm}
    \caption{\footnotesize Legends for Figures \ref{fig:exectm} and \ref{fig:fig:par_local_tm}}
    \vspace{-6mm}
\end{figure*}
\begin{figure*}[tb!]
\vspace{-1mm}
    \centering
    \subfloat[\footnotesize %\naheed{End-to-end (E2E) running time for competing algorithms on smaller datasets}
    ]{
    % \includegraphics[scale=0.22,trim={0.4cm 0.5cm 0.4cm 4cm},clip ]
    % \fbox{
    \includegraphics[scale=0.22,trim={0.72cm 0.66cm 0.55cm 0.55cm}]{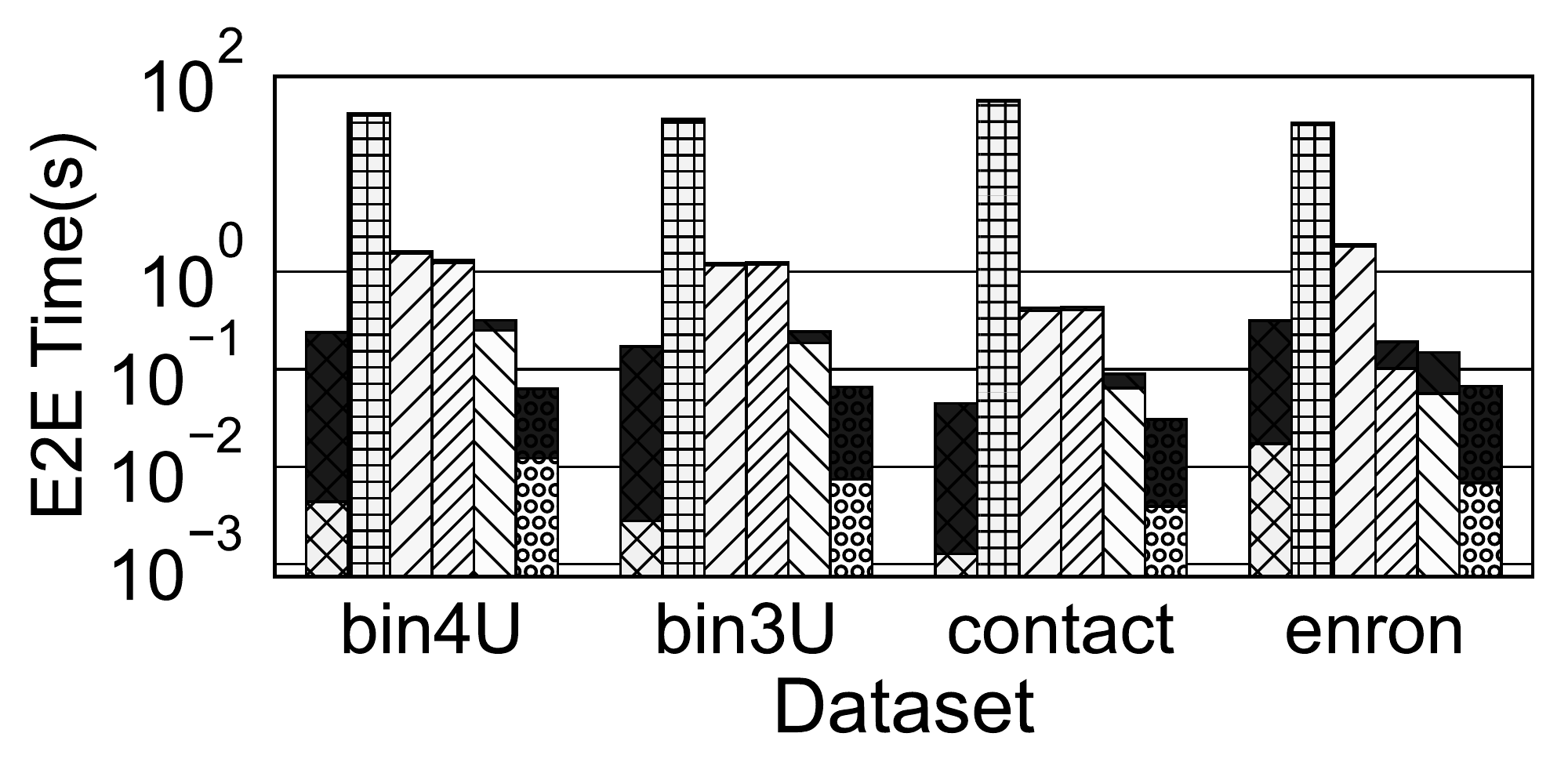}
    % }
    }
    \subfloat[\footnotesize %\naheed{End-to-end (E2E) running time for competing algorithms on bigger datasets}
    ]{
    % \fbox{
    \includegraphics[scale=0.22,trim={0.35cm 0.66cm 0.3cm 0.4cm},clip]
    {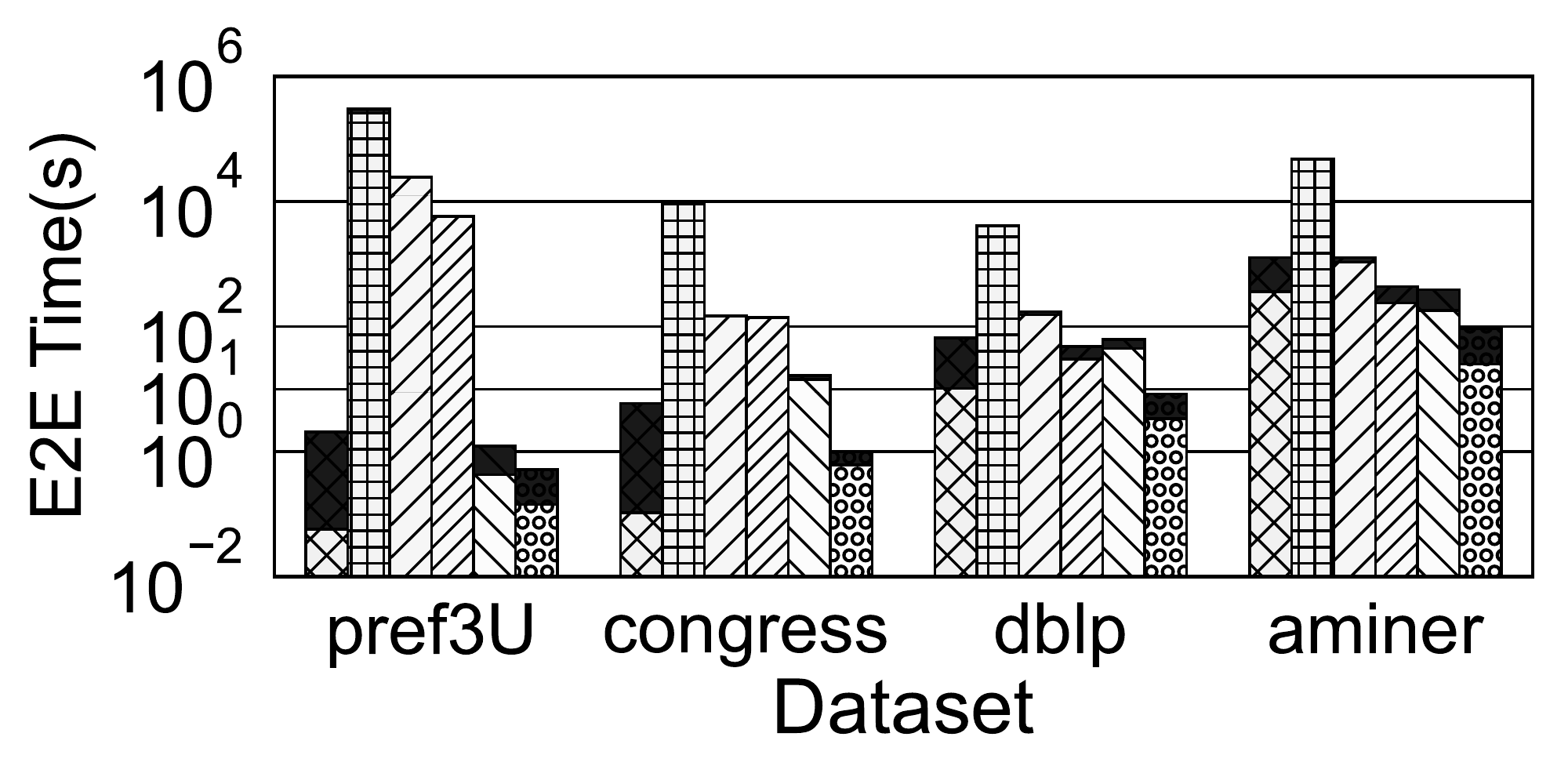}
    }
    % }
    % \subfloat[\footnotesize Smaller datasets]{\includegraphics[scale=0.2]{expt/small_execution_time_opt.pdf}
    \subfloat[\footnotesize %\naheed{Running time of algorithms}
    ]
    % {\includegraphics[scale=0.22,trim={0.4cm 0.5cm 0.4cm 4cm},clip]
    {
    % \fbox{
    \includegraphics[scale=0.22,trim={0.8cm 0.7cm 0.42cm 0.4cm}]{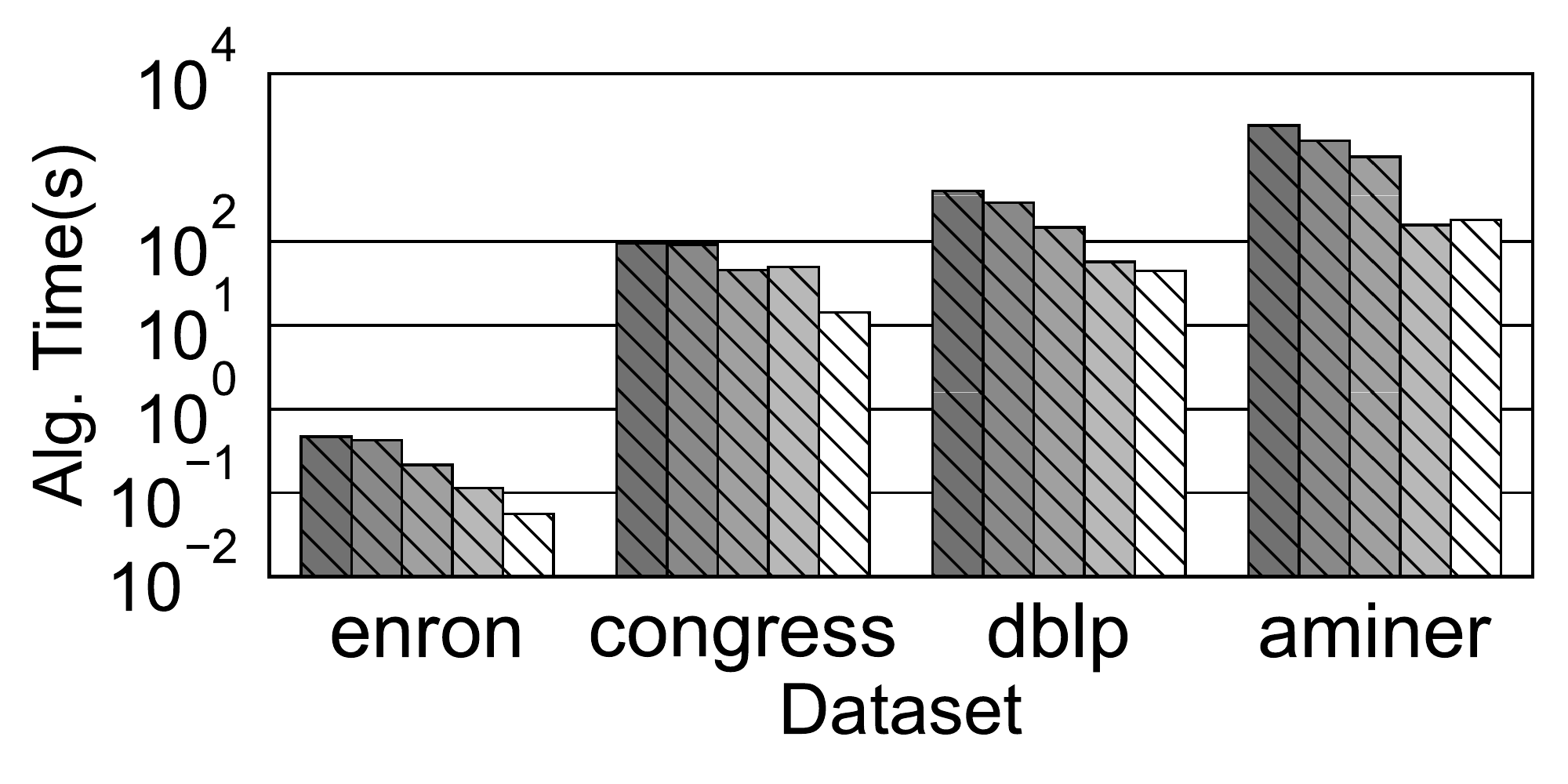}
    % }
    }
    \subfloat[\footnotesize %\naheed{Running time of algorithms}
    ]
    {
    % \fbox{
    \includegraphics[scale=0.21,trim={0.74cm 0.66cm 0.48cm 0.48cm}]{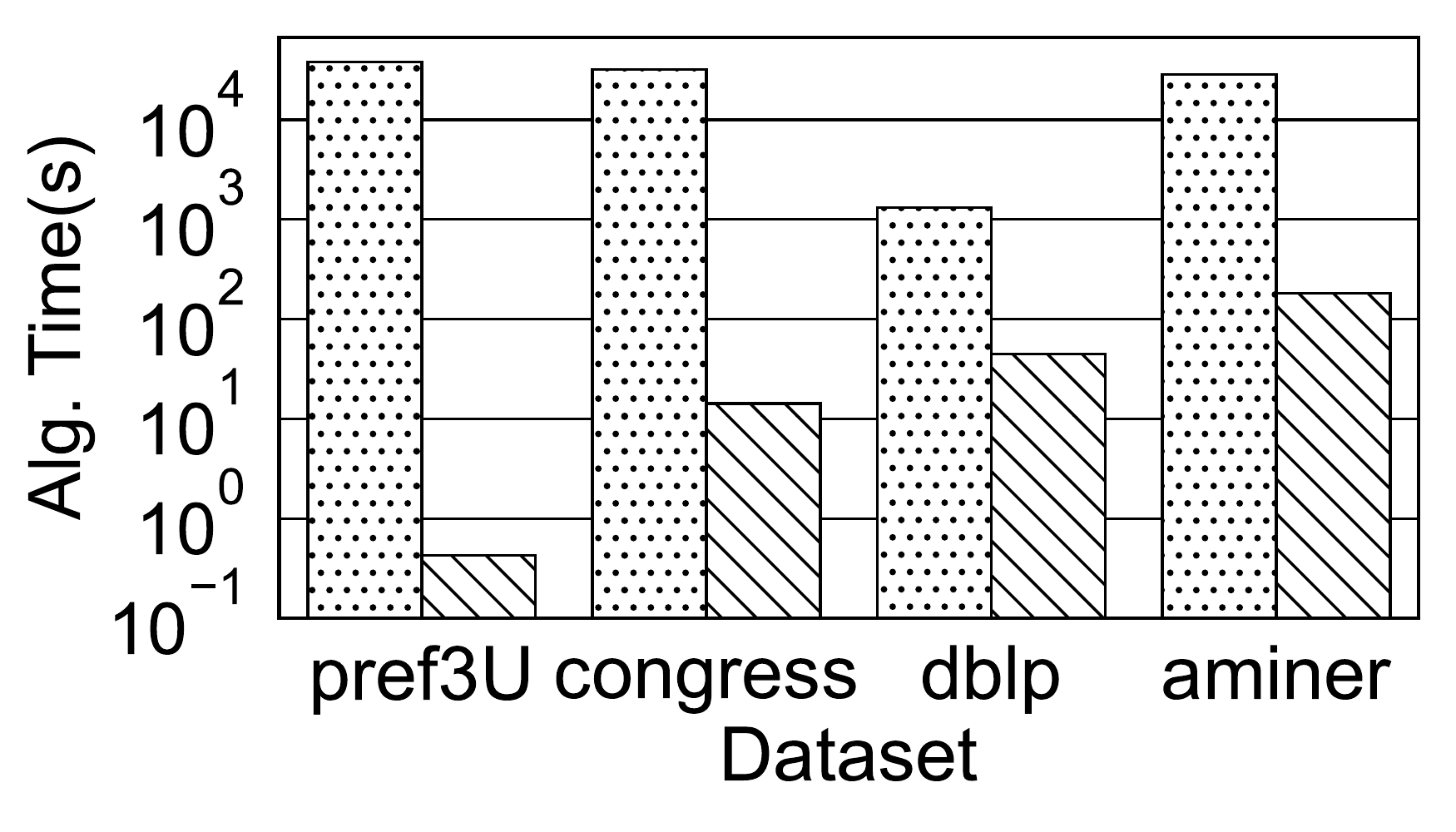}}
    % }
    \vspace{-4mm}
    \caption{\footnotesize (a)-(b) End-to-end (E2E) running time of our algorithms: \textbf{Peel}, \textbf{E-Peel}, \textbf{Local-core(OPT)}, \textbf{Local-core(P)} with 64 Threads
    vs. those of baselines: \textbf{Clique-Graph-Local} and \textbf{Distance-2 Bipartite-Graph-Local}. End-to-end (E2E) running time = data structure initialization time (shaded with dark-black on top of each bar) + algorithm's execution time. (c) Impact of the four optimizations to \textbf{Local-core} algorithm's execution time. Here, Local-core(OPT) = Local-core + Optimizations-I + II +III + IV. (d) \textbf{Local-core+Peel} algorithm's execution time for (neighborhood, degree)-core decomposition vs. \textbf{Local-core(OPT)} for neighborhood-core decomposition.}
    \label{fig:exectm}
    \vspace{-5mm}
\end{figure*}
\begin{figure}[tb!]
\vspace{-2mm}
    %\centering
    % \subfloat[\footnotesize{\em pref3U}]{\includegraphics[scale=0.22]{expt/revparallel_time_pref3U.pdf}}
    %\subfloat[\footnotesize {\em congress}]{\includegraphics[scale=0.22]{expt/revparallel_time_congress.pdf}}
    % \vspace{0.05pt}
    \subfloat[\footnotesize{\em dblp}]{\includegraphics[scale=0.2]{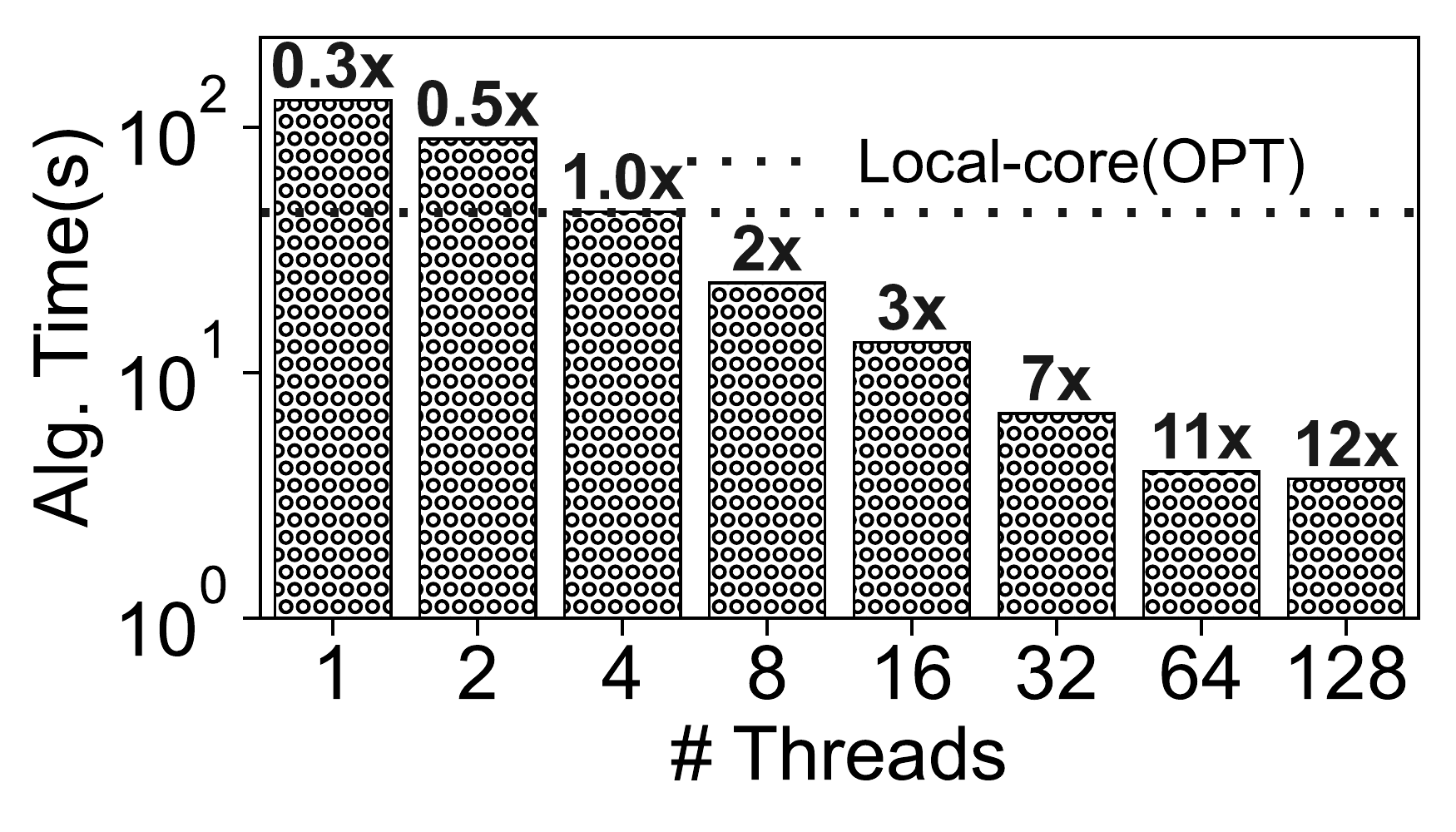}}
    \subfloat[\footnotesize{\em aminer}]{\includegraphics[scale=0.2]{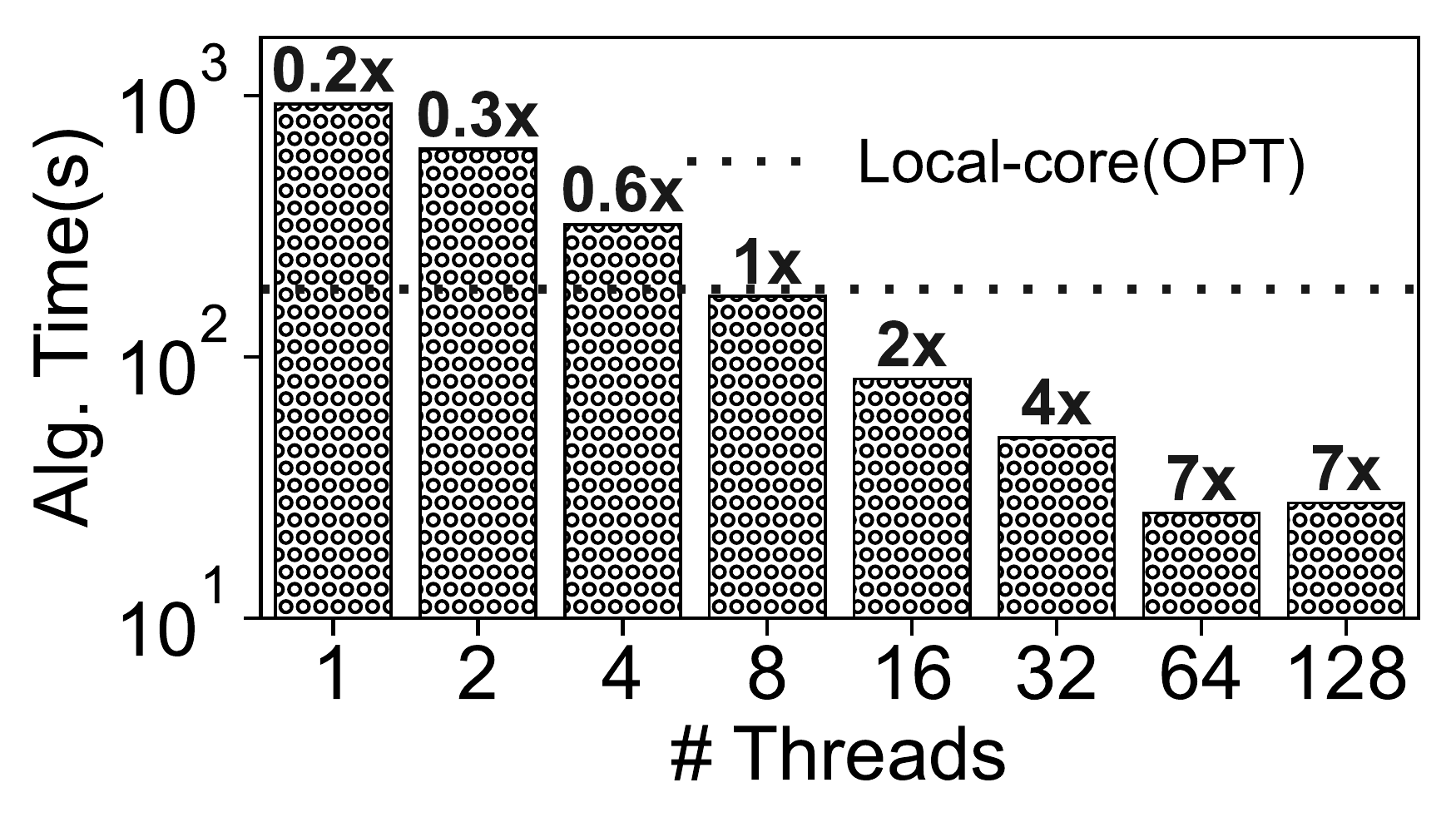}}
    \vspace{-5mm}
    \caption{\footnotesize \textbf{Local-core(P)} algorithm's execution time on larger datasets. \textbf{Local-core(P)} achieves up to 12x %with or without dynamic scheduling and up to 5-25x speedup with dynamic scheduling 
    speedup compared to sequential \textbf{Local-core(OPT)}.}
    \label{fig:fig:par_local_tm}
    \vspace{-5mm}
    \end{figure}
\begin{figure}[tb!]
%\vspace{-3mm}
    \centering
    %\subfloat[\footnotesize{\em pref5U}]
    {\includegraphics[scale = 0.2,trim={0.25cm 0.45cm 0.55cm 0.3cm},clip]{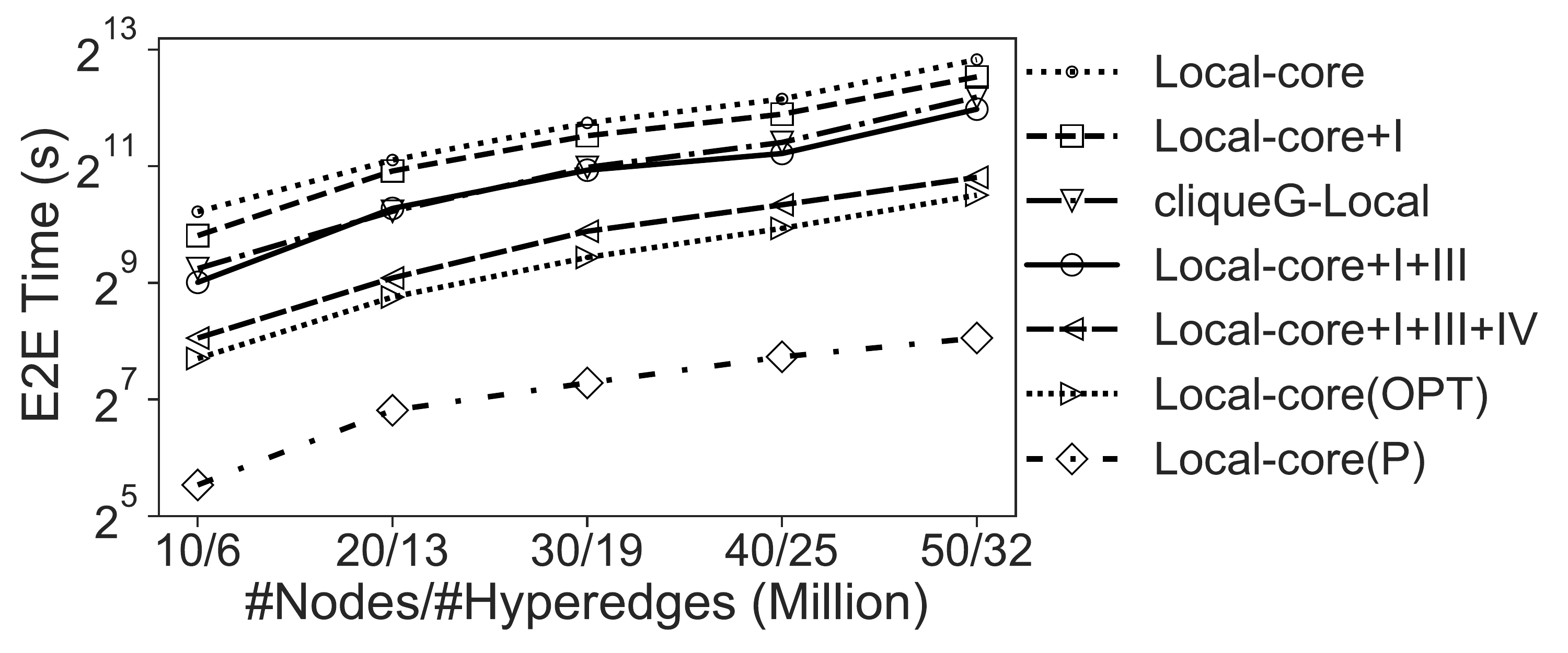}}
    \vspace{-3mm}
    \caption{\footnotesize {\naheed{Scalability comparison of the proposed algorithms with baseline.}} }
    \label{fig:scal}
    \vspace{-6mm}
\end{figure}
%\vspace{-2mm}
\subsection{Effectiveness of Local-core Algorithm}
\label{sec:effectiveness}
%In this experiment we show the novelty of the proposed hypergraph $h$-index by showing that a direct adaptation of graph $h$-index without any core-correction produces incorrect core-numbers. We also empirically establish the correctness of \textbf{Local-core} by showing that the percentage of the correctly converged nodes up-to each iteration increases from iteration-to-iteration eventually reaching
%100\% when the algorithm terminates. We use~\cref{def:wrongHind} to compute graph $h$-index and~\cref{defn:cor_core} to compute hypergraph $h$-index.
%
% \noindent\textbf{Accuracy vs. execution time: Local algorithm.}
\vspace{-1mm}
\spara{Exp-1: Novelty \& importance of hypergraph $h$-index.}
We demonstrate the novelty of the proposed hypergraph $h$-index (\Cref{defn:cor_core}) by showing that
a direct adaptation of graph $h$-index (\Cref{def:wrongHind}) without any core correction, that is,
{\em running the local algorithm from \cite{distributedcore,eugene15} may produce incorrect hypergraph core-numbers}.
 ~\Cref{fig:wr_localcore}(a) depicts that a local algorithm that only considers graph $h$-index without adopting our
novel \textbf{Core-correction} (\S~\ref{sec:local}) generates incorrect core-numbers for at least $90\%$ nodes on
\textit{bin4U}, \textit{bin3U}, and \textit{congress}.
On \textit{contact}, \textit{enron}, \textit{pref3U}, \textit{dblp}, and \textit{aminer},
core-numbers for at least 15\%, 26\%, 5\%, 17\%, and 10\% nodes are incorrect, respectively.
As nodes in \textit{bin4U}, \textit{bin3U}, and \textit{congress} have relatively higher $\text{mean}(\abs{N(v)})$ (Table~\ref{table:dataset}),
%The reason why graph $h$-index without any core correction produces incorrect core-numbers
%for more nodes on \textit{bin4U}, \textit{bin3U}, and \textit{congress} is that
there are more correlated neighbors in these datasets. Incorrect $h$-indices of correlated neighbors
have a domino-effect: A few nodes with wrong $h$-values, unless corrected, may cause all their neighbors to
have wrong $h$-values, which may in turn cause the neighbors' neighbors to have wrong $h$-values, and so on.
%For instance, wrong $h$-value in nodes $a,e$ in~\cref{fig:nbr1}(b) caused nodes $c,d$'s $h$-values to be wrong.
%As nodes in \textit{bin4U}, \textit{bin3U}, and \textit{congress} have relatively high $\text{mean}(\abs{N(v)})$,
%there are more correlated neighbors causing this domino-effect:
Unless all such correlated nodes are corrected,
almost all nodes eventually end up with wrong core-numbers.
% \textcolor{red}{state reason}.
%Here $\% \text{Wrong core} = \text{num. of nodes with incorrrect core numbers}\times 100/{\abs{V}}$.
%

We also compare the average error in the core-number estimates at each
iteration by graph $h$-index and our hypergraph $h$-index.
Here, $\text{avg. error at iteration n} = \sum_{u \in V} (h^{(n)}(u) - core(u))/\abs{V}$.
Figures \ref{fig:wr_localcore}(b)-(c) show avg. errors incurred at the end of each iteration on \textit{enron} and
\textit{bin4U}. At initialization, both indices have the same error on a specific dataset.
For both indices, avg. error at a given iteration is less than or equal to that in the previous iteration.
However at higher iterations, hypergraph $h$-index incurs less avg. error compared to graph $h$-index.
At termination, although hypergraph $h$-index produces correct core-numbers,
graph $h$-index has non-zero avg. error. These results suggest that hypergraph $h$-index estimates core-numbers more accurately
than graph $h$-index at intermediate iterations. We notice similar trends in other datasets, however for the same
iteration number, graph $h$-index produces higher avg. error on \textit{bin4U} than that on \textit{enron}. This is due to
more number of correlated neighbors in \textit{bin4U} than that in \textit{enron} as stated earlier.

\spara{Exp-2: Convergence of Local-core.}~\Cref{fig:wr_localcore}(d) shows that as  \# iterations increases in our \textbf{Local-core} algorithm (\S~\ref{sec:local}), the percentage of nodes with correctly converged core-numbers increases. The number of iterations for convergence depend on
hypergraph structure and computation ordering of nodes.
%For instance, \textit{dblp} takes $52$ iterations to converge, whereas \textit{bin3U} and \textit{bin4U} each takes 8 iterations, and \textit{pref} takes 7 iterations to converge. We find that \textcolor{blue}{after iteration 1, \textit{pref3U}, \textit{dblp} and \textit{enron} already have 83\%, 63\% and 59\% nodes with correct core-numbers, respectively. Further investigation reveals that 1) the statistical $\text{mode}(\abs{N(v_c)})$ of the converged nodes ($v_c$) is smaller ($2$, $2$, $1$) in \textit{pref3U}, \textit{dblp} and \textit{enron} compared to that ($21$, $135$, $186$) in \textit{contact}, \textit{bin3U} and \textit{bin4U}; and 2) a greater \%nodes (75\%, 29\%, 18\%) in \textit{pref3U}, \textit{dblp} and \textit{enron} have correlated neighbors ($\abs{N(v)}$) less than the threshold $\text{mode}(v_c)$ at iteration $1$, as opposed to that (1\%, 0.8\%, 0.6\%) in \textit{contact}, \textit{bin3U} and \textit{bin4U}. As nodes with fewer correlated neighbors converge earlier than those with more correlated neighbors, it is not surprising that applying hypergraph $h$-index only once is sufficient for nodes with $1$ or $2$ correlated neighbors to converge.
As \textit{pref3U}, \textit{dblp}, and \textit{enron} have more nodes with fewer correlated neighbors, more nodes achieve correct core-numbers after the first iteration.
On \textit{pref3U} 98\% nodes already converge by iteration 2. %It took 5 more iterations for the remaining 2\% nodes to converge.
This observation also suggests that one can terminate \textbf{Local-core} algorithm early %(e.g., at iteration 2)
at the expense of a fraction of incorrect results %(e.g., up to 2\% nodes)
on \textit{pref3U}.
%Notice that even though \textit{contact} has only 3\% nodes with correct core-numbers at the first iteration,
%its convergence rate is steeper than that of \textit{bin3U} and \textit{bin4U}, since
%\textit{contact} has a lower $\text{mean}(\abs{N(v)})$, and thus lesser correlated neighbors
%compared to that for \textit{bin3U} and \textit{bin4U}.

%
% \textcolor{red}{any early-termination we can do because of this, accuracy-efficiency trade-off plot?}
%
%It should be noted that the incorrect algorithm in~\cref{fig:wr_localcore}(a) makes all the nodes converge as well but often not to the right core-numbers, whereas \textbf{Local-core} makes all nodes converge to the right core-numbers. This subtle point was discussed earlier in~\cref{ex:whywronglocal}.
%
\vspace{-2mm}
\subsection{Efficiency Evaluation}
\label{sec:exp_efficiency}
We evaluate the efficiency of our proposed algorithms: \textbf{Peel}, \textbf{E-Peel}, \textbf{Local-core(OPT)}, and \textbf{local-core(P)}
vs. baselines: \textbf{Clique-Graph-Local} and \textbf{Bipartite-Graph-Local}. The baselines \textbf{Clique-Graph-Local} and \textbf{Bipartite-Graph-Local}
consider clique graph and bipartite graph representations (\S \ref{sec:diff}), respectively, of the hypergraph and then apply local algorithms \cite{eugene15,LiuZHX21} 
for graph core decomposition. 
From Figures \ref{fig:exectm}(a)-(b), which report end-to-end running times, %show the execution times of \textbf{Peel}, \textbf{E-Peel}, optimized sequential \textbf{Local-core} algorithm \textbf{Local-core(OPT)}, and its parallel variant \textbf{Local-core(P)}. 
we find that %that will be analysed in subsequent experiments.
%\begin{enumerate}
%    \item
\textbf{(1) E-Peel} is more efficient than \textbf{Peel} on all datasets. 
%    \item
\textbf{(2) Local-core(OPT)} is more efficient than \textbf{Peel} on all datasets.
%    \item
\textbf{(3) Local-core(P) is the most efficient algorithm among all the datasets.} \textbf{(4)} The baseline \textbf{Bipartite-Graph-Local} is the least efficient among all algorithms. This is because reducing the hypergraph to bipartite graph inflates the number of graph nodes and edges (\S \ref{sec:intro}), and distance-2 core decomposition \cite{LiuZHX21} on this inflated graph significantly increases the running time (\S \ref{sec:diff}). On the other hand, \textbf{Clique-Graph-Local}'s end-to-end running time is comparable to \textbf{Local-core(OPT)} (e.g., slightly more efficient on {\em contact} and {\em congress}, while less efficient on {\em enron} and {\em aminer}). This is because reduction to clique graph also inflates the problem size (\S \ref{sec:intro}), but graph-based $h$-index computation \cite{eugene15} on the clique graph does not require any core-correction. However, clique graph decomposition gives a different decomposition than our decomposition and in certain applications (as shown in \S \ref{sec:diff}, \S \ref{sec:application}) our decomposition is more desirable than clique graph decomposition.
% \item
% \naheed{\textbf{(4)} \textcolor{red}{.. explain results about \textbf{Clique-Graph-Local} and \textbf{Bipartite-Graph-Local}.}} 
%\end{enumerate}
%
%
\begin{table}[!tb]
%\vspace{-2mm}
\caption{\label{tab:init} \footnotesize {Data structure initialization times of all algorithms on {\em aminer}}}
\vspace{-4mm}
\scriptsize{
\resizebox{8.5cm}{!}{
\begin{tabular}{c|c|c|c|c|c}
%\hline
\textbf{Local-core(P)} & \textbf{Local-core (OPT)} & \textbf{E-Peel} & \textbf{Peel}  & \textbf{Clique-G-Local} & \textbf{Bipartite-G-Local} \\ 
% & \textbf{(OPT)}      &                 &                &  \textbf{-Local}      &   \textbf{-Graph-Local}        \\ 
\hline
66 s & 510 s                     & 195 s        & 198 s         & 806  s                  & 793 s           
\end{tabular}}
}
\vspace{-6mm}
\end{table}

Various algorithms require different data structures,
which include alternative representations of a hypergraph, e.g., finding neighbors and incidence hyperedges 
for all nodes (all algorithms), creating bipartite graph (\textbf{Bipartite-Graph-Local}), clique graph (\textbf{Clique-Graph-Local}), and CSR (\textbf{Bipartite-Graph-Local}, \textbf{Clique-Graph-Local}, \textbf{Local-core(OPT)}, \textbf{Local-core(P)}).
In end-to-end running times (Figures \ref{fig:exectm}(a)-(b)),
we included initialization times of data structures (shaded with dark-black on top of each bar). We
also report initialization times on our largest {\em aminer} dataset separately in Table~\ref{tab:init}.  
We parallelize CSR construction, finding neighbors and incidence hyperedges for nodes using OpenMP with dynamic scheduling \cite{openMP}. Parallel initialization time of data structures and execution time of \textbf{Local-core(P)}
are 66 sec and 25 sec, respectively, on {\em aminer}, thus end-to-end time being 91 sec.
% We find that the parallel data structure initialization time
% of \textbf{Local-core(P)} with 32 threads  on aminer is the least (Table~\ref{tab:init}). Increasing more number of threads (e.g. 64) incurs significantly more overhead (due to dynamically scheduling workload among more threads during runtime) compared to the time spent on data structure construction. Finally, considering the best running times for initialization (66 sec with 32 threads)
% and algorithm execution (25 sec for \textbf{Local-core(P)}, 64 threads), the end-to-end decomposition
% time for {\em aminer} with \textbf{Local-core(P)} is 66+25 = 91 seconds.  
%with 32 threads is less than that with 64 threads (Table~\ref{tab:init}),
%since \textcolor{red}{...explain}. Finally, considering the best running times for , initialization (66 sec with 32 threads)
%and algorithm execution (21 sec for \textbf{Local-core(P)}, 64 threads), the end-to-end decomposition
%time for {\em aminer} with \textbf{Local-core(P)} is 66+21 = 87 seconds.}

\spara{Exp-3: Efficiency of \textbf{E-Peel}.}
As stated in \S\ref{subsec:improved_peel}, the speedup of \textbf{E-Peel} over \textbf{Peel}
is related to $\alpha$, which is the ratio of the $\#\abs{N(u)}$ queries made by
\textbf{E-Peel} to that of \textbf{Peel}. %~\Cref{fig:epeel_speed}
%shows that datasets where $\alpha$ is smaller (larger), \textbf{E-Peel} achieves greater (lesser) speedup. For instance,
In Figures \ref{fig:exectm}(a)-(b), \textbf{E-Peel} achieves the highest speedup (17x compared to \textbf{Peel}) on \textit{enron}  because $\alpha = 0.35$ is the smallest
on this dataset.
% We also find that $\alpha \cong 1$ on \textit{congress}, \textit{bin3U}, and \textit{contact}, thus \textbf{E-Peel} gains almost no speedup on these datasets.
%

%We are unable explain why on \textit{enron} $\alpha=0.38$ whereas on \textit{congress} $\alpha \cong 1$, because it is impossible to derive apriori the exact $\#\abs{N(u)}$-queries that will be made by \textbf{Peel} or \textbf{E-Peel} analytically.
%\textcolor{red}{discuss why the lower-bound is not effective on these three datasets}
%
%\begin{figure}[!htb]
%\vspace{-4mm}
%    \centering
%     \subfloat[Small datasets]{\includegraphics[scale=0.2]{expt/small_execution_time_opt.pdf}}
%    \subfloat[Large datasets]{\includegraphics[scale=0.2]{expt/large_execution_time_opt.pdf}}
%    \vspace{-3mm}
%    \caption{\small Impact of the four optimizations to \textbf{Local-core}.}
%    \label{fig:opt_local_tm}
%    \vspace{-4mm}
%\end{figure}
%
%\begin{figure}[!htb]
%    \includegraphics[scale=0.25]{expt/Epeel_exp.pdf}
%    \vspace{-3mm}
%    \caption{\small The speedup of \textbf{E-Peel} (= ratio of the average execution time of \textbf{Peel} to that of \textbf{E-Peel}) is inversely related to the ratio of \# neighborhood computations $(\alpha)$}
%    \label{fig:epeel_speed}
%    \vspace{-4mm}
%\end{figure}
% \noindent\textbf{Exp-2: Efficiency of Local-algorithms}

In remaining experiments (Figures \ref{fig:exectm}(c)-(d), 
 \ref{fig:fig:par_local_tm}), we compare similar algorithms that have same initialization time on a dataset, thus we compare these algorithms' execution times (`Alg. Time' on $y$-axis).

\spara{Exp-4: Efficiency and impact of optimizations to \textbf{Local-core}.}
We next analyze the efficiency of the proposed optimizations (\S\ref{sec:optimization})
with respect to \textbf{Local-core} without any optimization.
\textbf{Local-core+I} incorporates optimization-I to \textbf{Local-core},
\textbf{Local-core+I+III} incorporates optimization-III on top of \textbf{Local-core+I},
\textbf{Local-core+I+III+IV} incorporates optimization-III on top of \textbf{Local-core+I+III}.
Finally, \textbf{Local-core(OPT)} incorporates all four optimizations.
Figures \ref{fig:exectm}(c) shows the execution times of \textbf{Local-core}, \textbf{Local-core+I}, \textbf{Local-core+I+III}, \textbf{Local-core+I+III+IV},
and \textbf{Local-core(OPT)} on four representative datasets. On all hypergraphs,
we observe that adding each optimization generally reduces the execution time.
% \naheed{Among them, we find Optimizations-I, II and III to be the most effective ones reducing the algorithm runtime by at least 1.4x-2x on \textit{dblp} and \textit{aminer} datasets.} 

% The impact of adding the same optimization on different datasets is different in general.
% \textcolor{red}{For instance, on {\em enron} the speedup of \textbf{Local-core+I+III+IV} w.r.t \textbf{Local-core+I+III} is 1.5x,whereas on {\em congress} the speedup is 1x (no speedup). This is because on {\em enron}, a large number of nodes ($90\%$) have their core-numbers equal to their respective local lower-bounds (\Cref{lem:localLB}).
% However, on congress, only $0.05\%$ nodes have core-numbers equal to their local lower-bounds. Hence, more redundant
% $\mathcal{H}$ computations are saved on {\em enron} compared to that in {\em congress}.}
% \textcolor{red}{need to investigate more in-depth, such as, which optimizations are more effective and why so, why some optimization does not work on some datasets, etc.}

\spara{Exp-5: Impact of parallelization.} %to Local-core.}
% \begin{figure}[!htb]
%     \centering
%      \subfloat[aminer]{\includegraphics[scale=0.2]{expt/parallel_time_aminer.pdf}}
%     \subfloat[dblp]{\includegraphics[scale=0.2]{expt/parallel_time_dblp.pdf}}
%     \vspace{-3mm}
%     \caption{\small Speedup of Local-core(P) compared to sequential algorithms.}
%     \label{fig:fig:par_local_tm}
%     \vspace{-4mm}
% \end{figure}
We test the parallelization performance of \textbf{Local-core(P)} by varying the number of threads (\Cref{fig:fig:par_local_tm}).
Adding 64-128 threads reduces %the per-thread workload in a single iteration. Thus, individual iterations take less time culminating in the reduction
%of
the overall execution time up to 7-12x on larger {\em dblp} and {\em aminer} datasets. %For example, in~\Cref{fig:fig:par_local_tm}(b), \textbf{Local-core(P)} is \textbf{5.5x} faster on \textit{congress}
%when the number of threads increases from 1 to 64. However in Figure~\ref{fig:fig:par_local_tm}(a), \textbf{Local-core(P)} achieves only up to \textbf{1.9x} speedup. The limited
%speedup is due to two reasons: 1) As shown in Exp-2, 98\% nodes converge by 2 iterations in \textit{pref3U}.
%In the remaining iterations, only $2\%$ nodes are processed in parallel. This limits the parallelism as many threads remain idle.
%2) \textit{pref3U} has a skewed neighborhood-size distribution -- the standard deviation (412.4 in Table~\ref{table:dataset})
%is the largest among all datasets. Thus, there are a few non-converged nodes with significantly large numbers of neighbors.
%The thread handling such a node takes longer than the rest, %of the threads,
%causing %the overall
%higher execution time. %to increase.
%\naheed{We find that \textcolor{red}{... mention that dynamic load balancing has little impact than static load balancing} (\Cref{fig:fig:par_local_tm})}. 
We measure execution times of \textbf{Local-core(P)} without load-balancing (\S\ref{sec:optimization}),
and find that load-balancing speeds up \textbf{Local-core(P)} up to 1.8x on \textit{aminer} and \textit{dblp}.%, 1.1x times on \textit{congress},
%but insignificantly on \textit{pref3U}.

\spara{\naheed{Expt-6: Scalability.}} \naheed{We test scalability of the proposed algorithms and the baseline (\textbf{Clique-Graph-Local}) on  5-uniform hypergraphs with power law degree distribution (exponent 2.1) ~\cite{avin2019random}, consisting from 10 million up to 50 million nodes. Figure~\ref{fig:scal} shows that our most-efficient algorithm \textbf{Local-core(P)}, as well as the sequential algorithms \textbf{Local-core-I+III+IV} and \textbf{Local-core(OPT)}, outperform the baseline \textbf{Clique-Graph-Local}.  Our \textbf{Local-core(P)} is also scalable, requiring only 248 sec end-to-end running time on 5-uniform hypergraph with 50 million nodes and 32 million hyperedges.}

\spara{Expt-7: Efficiency of \textbf{Local-core+Peel} for (neighborhood, degree)-core decomposition.}
Figure \ref{fig:exectm}(d) shows that the \textbf{Local-core+Peel} algorithm for $(k,d)$-core decomposition takes more time than \textbf{Local-core(OPT)} on larger datasets. This is expected since the number of possible $(k,d)$-cores is $\bigO(k_{max}\cdot d_{max})$, whereas the number of neighborhood-based cores is only $k_{max}$ (\S \ref{sec:kdcore}). As the output size increases, the algorithm \textbf{Local-core+Peel}, on top of running \textbf{Local-core(OPT)}, explicitly constructs subhypergraphs $H[V_k]$ and peels $H[V_k]$ based on degree-dimension, for all nbr core-numbers $k$. %\vspace{3mm}
\section{Applications and Case Studies}
\label{sec:application}
We study the superiority of important nodes from neighborhood-based decomposition over those from other decompositions.
\begin{figure}[tb!]
    \vspace{-4mm}
    \centering
    \subfloat[{\em \scriptsize dblp}]{\includegraphics[scale=0.22]{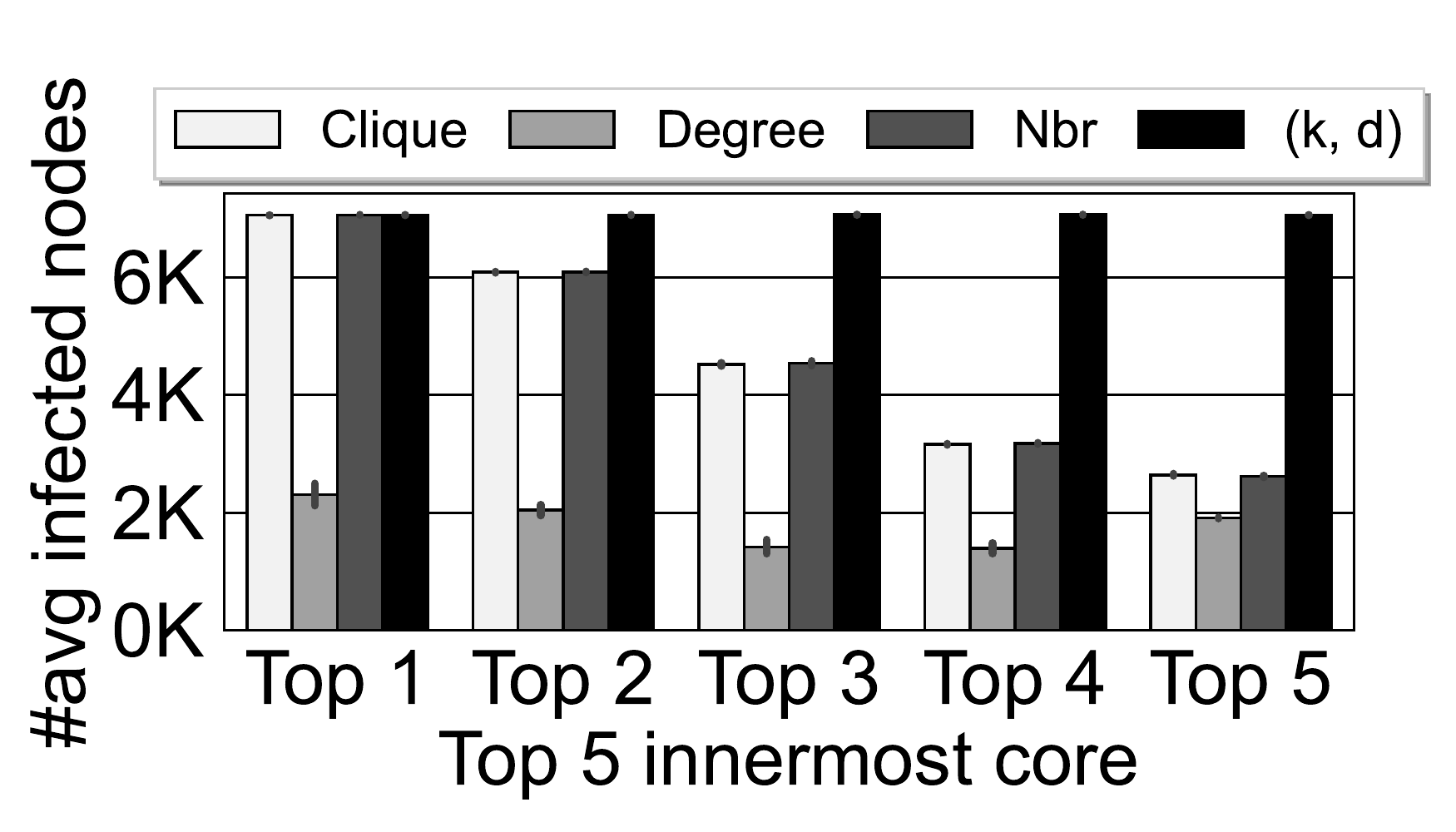}}\vspace{-2.2mm}
    % \hspace{0.1mm}
    \subfloat[{\em \scriptsize dblp}]{\includegraphics[scale=0.18,trim={0.4cm 0.6cm 0.5cm 0.3cm},clip]{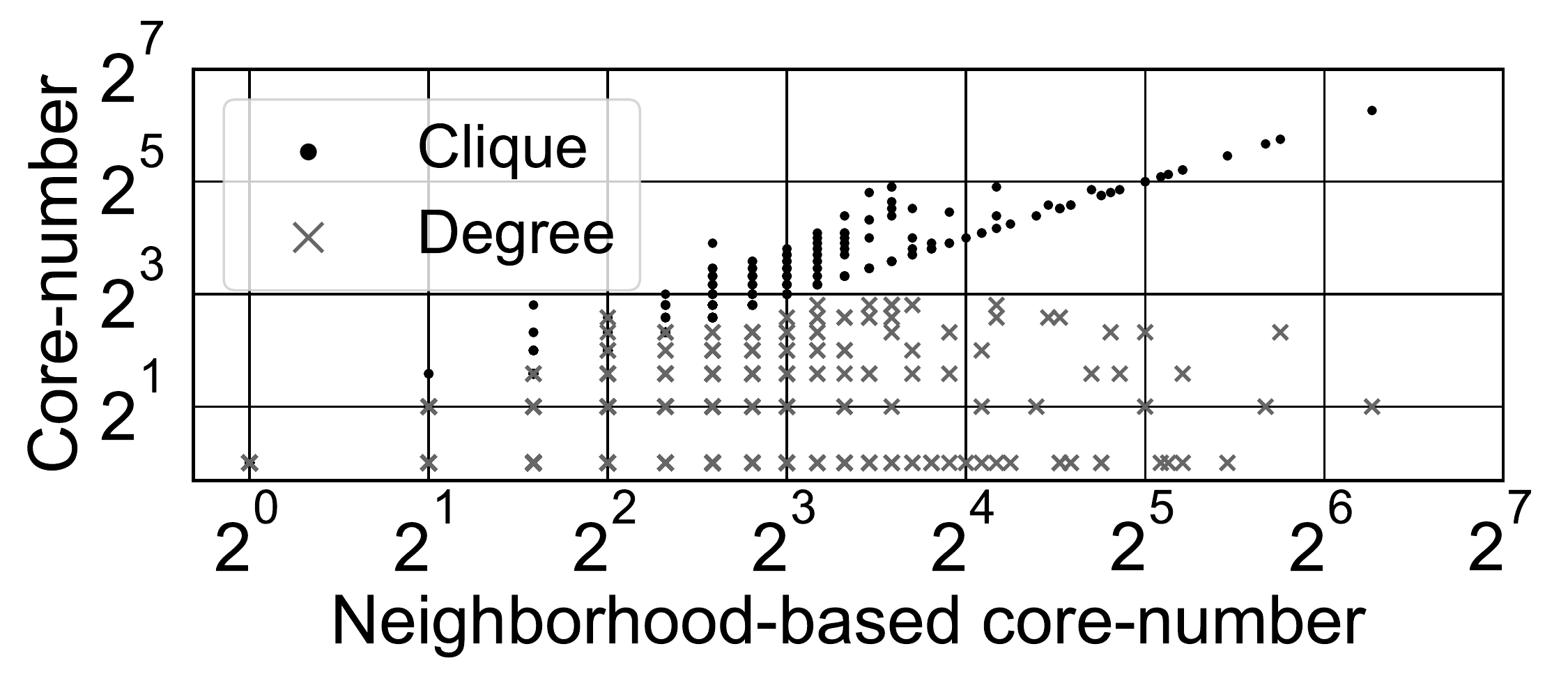}}\vspace{-2.2mm}
    \subfloat[{\em \scriptsize dblp}]{\includegraphics[scale=0.22]{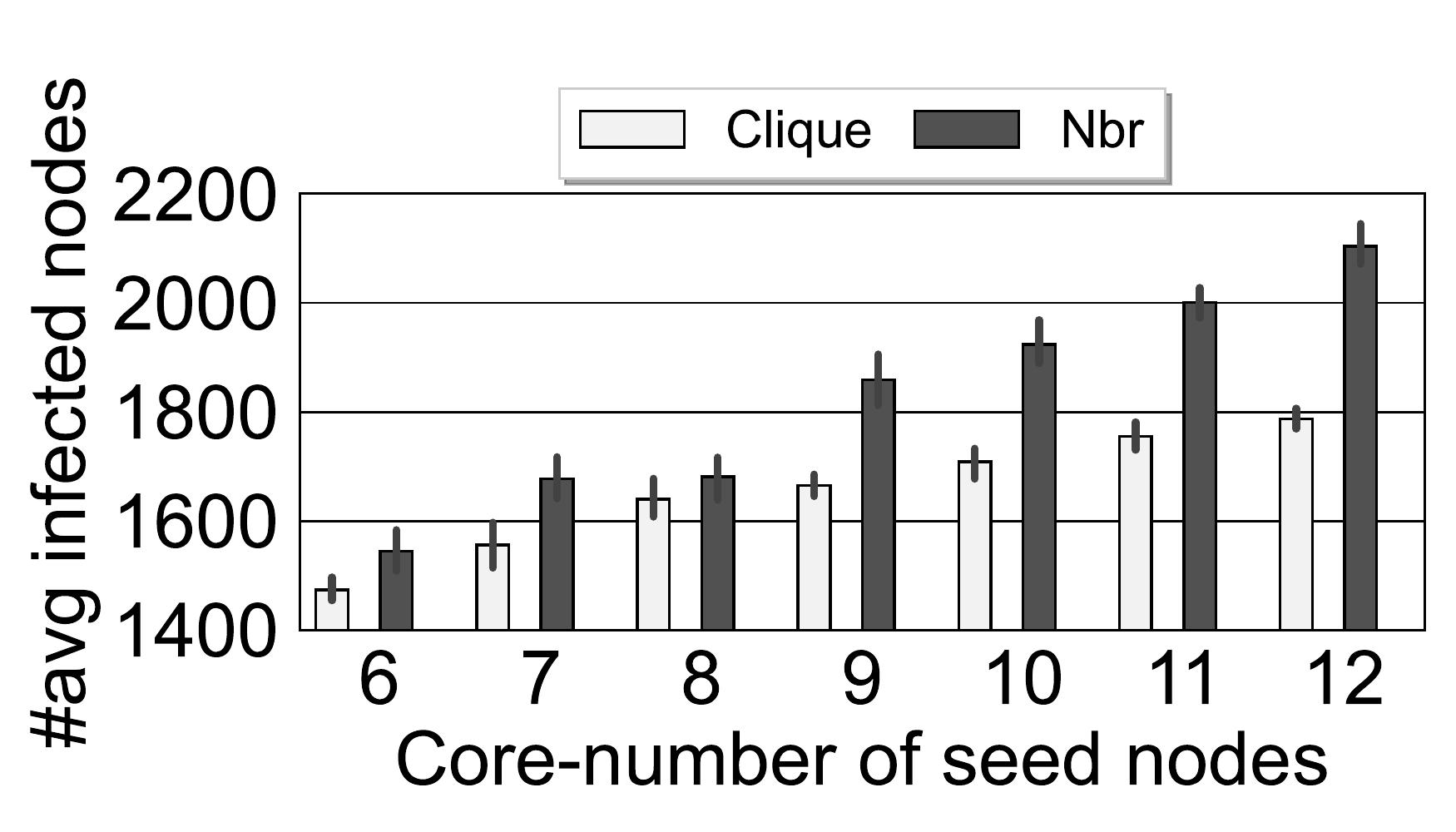}}\vspace{-2mm}
    \subfloat[{\em \scriptsize dblp}]{\includegraphics[scale=0.2]{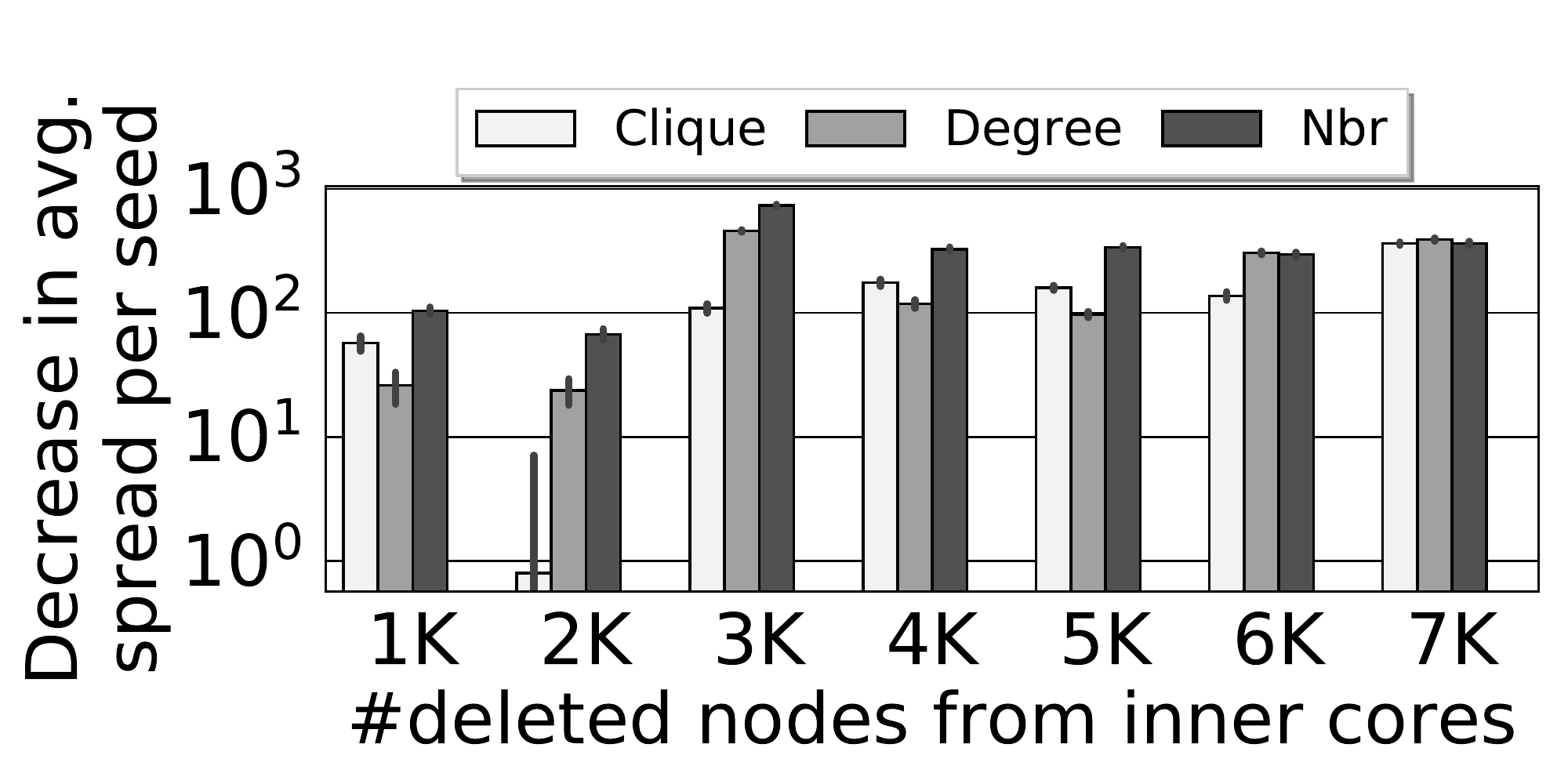}}
    \vspace{-2mm}
    % \caption{\footnotesize (a) Neighborhood decomposition based core-number vs. degree and clique graph decomposition based core-numbers of nodes.
    % \vspace{-4mm}
    %\subfloat[]{\includegraphics[scale=0.27]{application_figures/dblp_9c_100.pdf}}
    % \vspace{-1mm}
    \caption{\footnotesize %Effect of the core number on diffusion over
     \naheed{(a) The avg. \#infected nodes decreases as seeds are selected from innermost to outer cores. }  %(c) \textcolor{red}{Distribution of nodes in core-number groups according to neighborhood, degree, and clique-graph decompositions.}
     (b) Neighborhood decomposition core-number vs. degree and clique graph decomposition core-numbers of nodes. \naheed{(c) Neighborhood-based decomposition outperforms clique-based method when seeds are selected from outer cores.} (d) The decrease in avg. expected \#infected nodes per seed is the highest when the top-$k$ nodes are deleted via neighborhood-based decomposition.
     }
    \label{fig:application}
    \vspace{-6mm}
\end{figure}
%\vspace{-1mm}
\subsection{Influence Spreading and Intervention}
\label{sec:inf_applications}
%
% To this end, %we apply a diffusion modeling algorithm~\cite{kitsak} on a hypergraph and study the practical impact of nodes belonging to the innermost core. More specifically,
We consider the \emph{SIR} diffusion process~\cite{kitsak}:
Initially, all nodes except one\textemdash called a \textit{seed}\textemdash are at the \emph{susceptible} state.
The seed node is initially at the \emph{infectious} state. At each time step, each infected node infects its susceptible neighbors %in the hypergraph
with probability $\beta$ and then becomes \emph{immunized}. Once a node is immunized, it is never re-infected.

\spara{Inner-cores contain influential spreaders (Figure~\ref{fig:application}(a-c)).}
\eat{For a certain decomposition, we rank nodes in inner cores higher than those in outer cores -- we select the top-$r$ seeds (on the $x$-axis) following this order. For $(k,d)$-cores, nodes in the same neighbor $k$-core are further sorted based on their degree-based core-numbers according to $(k,d)$-core decomposition.}
\naheed{For each decomposition method, we select the top-$5$ innermost cores (on the $x$-axis), and for each core, we run the SIR model ($\beta$=0.3) from each of 100 uniformly selected seed nodes in that core. For $ (k, d) $-cores, we choose inner cores with top-$ 5 $ $ k $ values, where for each $ k $, we select seed nodes with the maximum $ d $-value.  Figure~\ref{fig:application}(a) shows the \#infected nodes averaged per seed node from a core. We find that {\bf (1)} the avg. \#infected nodes generally decreases as we move from inner to outer cores, implying that inner cores contain better quality seeds. 
{\bf (2)} Seeds selected according to $(k,d)$-core decomposition outperform seeds selected via other decomposition, indicating that our {\em $(k,d)$-core decomposition produces the best-quality seeds for maximizing diffusion}. {\bf (3)} Seeds selected via degree-based decomposition have the lowest spread. {\bf (4)} The quality of seeds derived from the neighborhood and clique graph decompositions have similar spread, the reason being that their inner cores have higher similarity in terms of constituent nodes, while outer cores from these two decompositions are quite different in {\em dblp}.}
% For each seed node in the top-$r$ set, we run the SIR model  ($\beta$=0.3) and compute the expected number of infected nodes after 100 time steps. Finally, on the $y$-axis, we report the average number of expected infected nodes per seed.
\eat{For each seed node in the top-$r$ set, we run the SIR model ($\beta$=0.3) and compute the expected number of infected nodes (spread) after 100 time-steps. Finally, on the $y$-axis, we report the spread averaged over the seeds.
We find that {\bf (1)} with larger $r$, avg. spread per seed from the top-$r$ set generally decreases, implying that inner cores contain better quality seeds. 
{\bf (2)} Seeds selected according to $(k,d)$-core decomposition outperform seeds selected via other decomposition, indicating that {\em $(k,d)$-core decomposition produces the best seeds for maximizing diffusion}. {\bf (3)} Seeds selected via degree-based decomposition have the lowest spread. The quality of seeds derived from the neighborhood and clique graph decompositions have similar spread, the reason being that their inner cores have higher similarity in terms of constituent nodes, while outer cores from these two decompositions are quite different in {\em dblp}} 
\naheed{More specifically, Figure~\ref{fig:application}(b) highlights this difference by comparing core-numbers of 1000 randomly selected nodes according to different methods. We notice a linear correlation between the neighborhood and clique graph decomposition-based core-numbers for nodes in inner cores, but such linear correlation diminishes as we consider nodes from outer cores.} We do not observe any correlation between the neighborhood and degree-based core-numbers. \naheed{ {\bf{ (5) }} Figure~\ref{fig:application}(c) shows that neighborhood-based method outperforms clique graph when seed nodes are selected from relatively outer cores.}

\spara{Deleting inner-cores for maximum intervention in spreading (Figure~\ref{fig:application}(d)).}
%As shown in Figures~\ref{fig:application}(a)-(b), innermost cores are important in diffusion of disease, information, and in general
%any propagated entity over a hypergraph.
We 
%capitalize on this observation to
devise an intervention strategy to disrupt diffusion,
which is critical in mitigating the spread of contagions in epidemiology,
limiting the spread of misinformation, or blocking competitive campaigns in marketing. The nodes in inner cores according to a specific decomposition method (neighborhood, degree, or clique) are considered more important than nodes in outer cores.
We select the top-$k$ most important nodes (on the $x$-axis) according to a specific decomposition, then we delete those nodes and incident hyperedges. For fairness across different decompositions, we consider the same set of seeds for all -- seeds are selected from outside the union of deleted node sets according to three decompositions. To determine the effectiveness of such deletion strategy, on the $y$-axis, we report the decrease in average spread (expected number of infected nodes) per seed. We observe that when deleting up to 6K nodes, deletion of important nodes via neighborhood-based decomposition always results in a maximum decrease of spread, compared to the same done via other decompositions (degree and clique). Beyond 6K nodes, each decomposition method deletes a significant number of important nodes according to that decomposition, and the quality of seeds (which are from outside deleted regions) also degrades; thus, all approaches cause a similar decrease in spread. This result shows that our {\em neighborhood-based decomposition produces the best order of important nodes for deleting a limited number of them, while causing the maximum intervention in spreading}.

\vspace{-2mm}
\subsection{Densest SubHypergraph Discovery}
\label{sec:dens_applications}
The degree-densest subgraph is a subgraph with the maximum average node-degree among all subgraphs of a given graph \cite{G84,charikar,%TsourakakisBGGT13,GT15,
FYCLX19}, which may correspond to communities \cite{DourisboureGP09} and echo chambers \cite{%asatani2021dense,
L22} in social networks, brain regions responding to stimuli \cite{legenstein_et_al:LIPIcs}. %or diseases \cite{wu2021extracting}. %and commercial value motifs in financial domains \cite{DuJDLT09}.
Following same principle, we define a new notion of densest subhypergraph, called the \emph{volume-densest subhypergraph}, based on the number of neighbors of nodes in a hypergraph.
%Volume-densest subhypergraphs are not only more effective in maintaining neighborhood cohesiveness than existing degree-densest subhypergraphs,
%but also provides valuable insights as we show in our case-studies from biology (\S \ref{subsubsec:casestudyI}) and email-communication (\S \ref{subsubsec:casestudyII}) domains.
% We first propose an approximation algorithm to compute volume-densest subhypergraph and show that the node-peeling sequence of the algorithm follows a certain total order induced by core-numbers. We derive approximation guarantee for the algorithm.
%
%\spara{Volume-density.}
The \textbf{volume-density} $\rho^N[S]$ of a subset $S \subseteq V$ of nodes in a hypergraph $H = (V,E)$
is defined as the ratio of the summation of neighborhood sizes of all nodes $u \in S$ in the induced
subhypergraph $H[S]$ to the number of nodes in $H[S]$.
\begin{small}
\begin{align}
\rho^N[S] = \frac{\sum_{u\in S} \abs{N_S(u)} }{\abs{S}}% = \rho[ CLIQUE\_CLOSURE(E_S) ]
\label{eq:vol_density}
\end{align}
\end{small}
The \textbf{volume-densest subhypergraph} is a subhypergraph which has the largest volume-density among all subhypergraphs.

\spara{Approximation algorithm.}
Inspired by Charikar~\cite{charikar},
our approach 
% to finding the (approximate) volume-densest subhypergraph
follows the peeling paradigm: %(Algorithm~\ref{alg:approx}):
In each round, we remove the node with the smallest number of neighbors in the current subhypergraph.
In particular, we sort the nodes in ascending order of their neighborhood-based core-numbers,
obtained from any neighborhood-based core-decomposition algorithm
(e.g., \textbf{Peel}, \textbf{E-Peel}, \textbf{Local-core(OPT)}, or \textbf{Local-core(P)}).
%In Algorithm~\ref{alg:approx},
We peel nodes in that order. Among nodes with the same core-number, the one with
the smallest number of neighbors in the current subhypergraph is peeled earlier.
We finally return the subhypergraph that achieves the largest volume-density.
%
%\begin{algorithm}[tb!]
%\caption{\small Approximate volume-densest subhypergraph detection}
%\label{alg:approx}
%\begin{algorithmic}[1]
%\scriptsize
%% \Require Hypergraph $H = (V,E)$, Total order $(V,\prec)$
%\Require Hypergraph $H = (V,E)$
%\Ensure Hypergraph $H^*$
%\State $S_1 \gets V$
%\For{$i \gets 1,2,\ldots,n-1$}
%% \For{$u_i \in (V,\prec)$}
%    \State $u_i \gets \argminG_{u\in S_i} \abs{N_{S_i}(u)}$
%    % \Comment{$N(u)\at{S_i}$is the set of neighbors of $u$ in $(S_i,E[S_i])$}
%    \State $S_{i+1} \gets S_i \setminus \{u_i\}$
%\EndFor
%\State $H^* \gets \argmaxG_{\{H[S_i]\mid i \in [n]\}} \rho^N[S_i]$
%\State \textbf{return} $H^*$
%\end{algorithmic}
%\end{algorithm}
%
%
%
%
\vspace{-1mm}
\begin{theor}
Our volume-densest subhypergraph discovery algorithm returns $(d_{pair}(d_{card}-2)+2)$-approximate densest subhypergraph if the maximum cardinality of a hyperedge
is $d_{card}$ and the maximum number of hyperedges between a pair of nodes is $d_{pair}$ in the input hypergraph.
\end{theor}
\begin{proof}
Let a volume-densest subhypergraph be $H^* =(S^*,E[S^*])$. For all $v\in H^*$, due to optimality of $H^*$:
\begin{small}
\begin{align}
\rho^N[H^*] = \frac{\sum_{u\in S^*} \abs{N_{S^*}(u)} }{\abs{S^*}} \geq \frac{\sum_{u\in S^*\setminus\{v\}} \abs{N_{S^*\setminus\{v\}}(u)}}{\abs{S^*}-1}
\label{eq:ineq1}
\end{align}
\end{small}
Next, we verify that:
\begin{small}
\begin{align}
& \sum_{u\in S^*\setminus\{v\}} \abs{N_{S^*\setminus\{v\}}(u)} \nonumber \\
& \geq \sum_{u\in S^*} \abs{N_{S^*}(u)} - \abs{N_{S^*}(v)} - \left(d_{pair}(d_{card}-2)+1\right)\abs{N_{S^*}(v)}
\label{eq:ineq2}
\end{align}
\end{small}
where $d_{card}$ is the maximum cardinality of any hyperedge in the input hypergraph. On the right-hand side,
we subtract $\abs{N_{S^*}(v)}$ since $v$ is removed from $S^*$. We also subtract $\left(d_{pair}(d_{card}-2)+1\right)\abs{N_{S^*}(v)}$, since
by deleting $v$, all hyperedges involving $v$ will be removed. As a result, the neighborhood size of every vertex in $\abs{N_{S^*}(v)}$
can be reduced by at most $d_{pair}(d_{card}-2)+1$. By combining the Inequalities \ref{eq:ineq1} and \ref{eq:ineq2}, we derive, for all
$v\in H^*$:
\begin{small}
\begin{align}
\abs{N_{S^*}(v)} \geq \frac{\rho^N[H^*]}{d_{pair}(d_{card}-2)+2}
\label{eq:ineq3}
\end{align}
\end{small}
We show that at the point when a node $u \in S^*$ is removed by the algorithm,
the current subhypergraph must have a volume-density at least $\frac{1}{d_{pair}(d_{card}-2)+2}$
of the optimum. Consider the iteration such that $S^* \subseteq S_i$,
but $S^* \not\subseteq S_{i+1}$.
Due to the greediness of the algorithm, the following holds for all $w \in S_i$: $\abs{N_{S_i}(w)} \geq \abs{N_{S_i}(u)} \geq \abs{N_{S^*}(u)}$
The second Inequality holds since $S^* \subseteq S_i$.
Hence,
\begin{small}
\begin{align}
& \sum_{w\in S_i} \abs{N_{S_i}(w)} \geq \abs{S_i}\abs{N_{S^*}(u)}
&\implies \frac{\sum_{w\in S_i}\abs{N_{S_i}(w)}}{\abs{S_i}} \geq \abs{N_{S^*}(u)}
\label{eq:ineq5}
\end{align}
\end{small}
Substituting in Inequality~\ref{eq:ineq3}, we get:
\begin{small}
\begin{align}
\rho^N[S_i] \geq \frac{\rho^N[H^*]}{d_{pair}(d_{card}-2)+2}
\label{eq:ineq6}
\end{align}
\end{small}
\end{proof}
\eat{
\begin{figure}
    \centering
    \vspace{-3mm}
    \includegraphics[width= 0.5\linewidth, height = 0.26\linewidth]{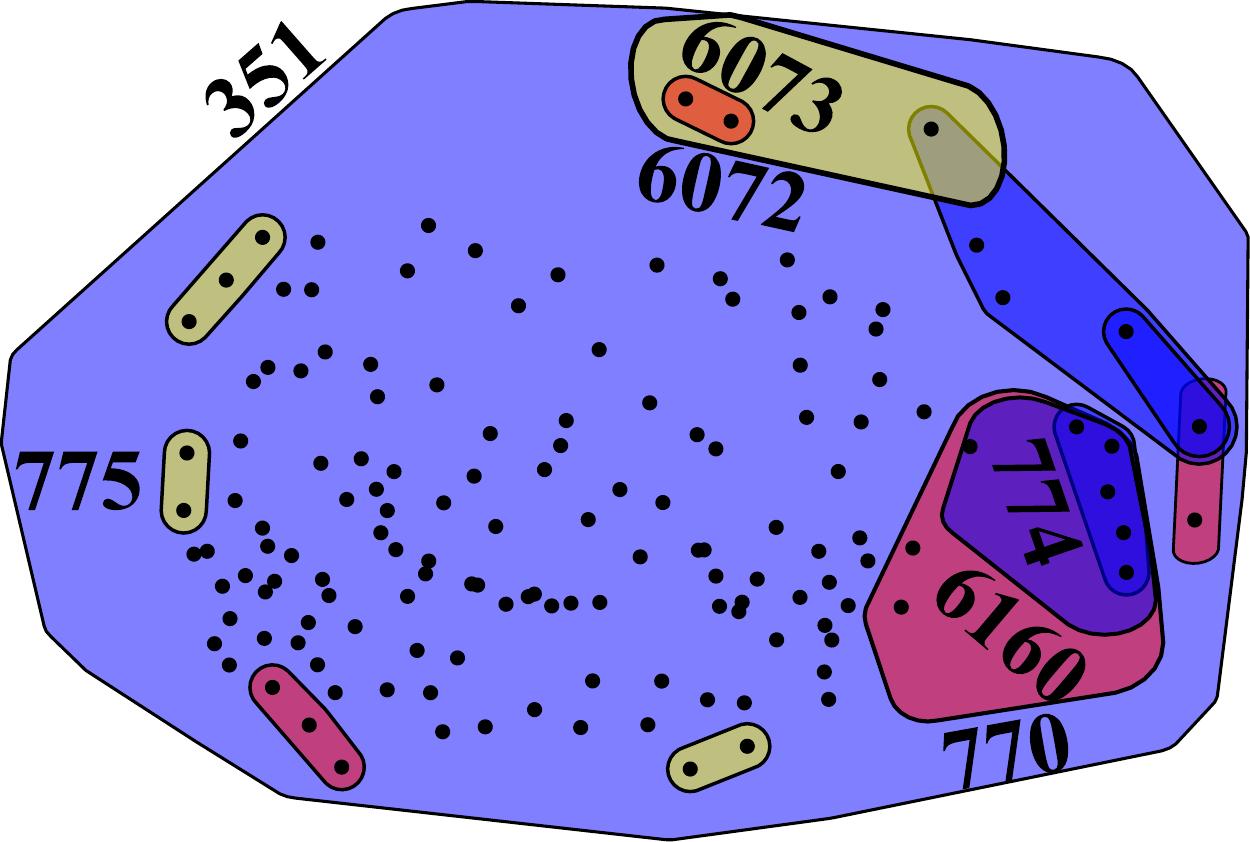}
    \vspace{-3mm}
    \caption{\footnotesize Volume-densest subhypergraph of the {\em human protein complex}}
    \label{fig:human}
    \vspace{-5mm}
\end{figure}
}
%
%The proof is given in our extended version \cite{full}. 
Notice that
$d_{pair} = 2$ if the hypergraph is a graph and our result gives $2$-approximation guarantee for the densest subgraph discovery~\cite{charikar}.

\spara{Case study: Meetup dataset.}
We extract all events with $<$ 100 participants from the Nashville meetup dataset \cite{meetup}. The extracted hypergraph contains 24\,115 nodes (participants) and 11\,027 hyperedges (events) organized by various interest groups. We compute and analyze the volume-densest subhypergraph, degree-densest subhypergraph \cite{hu2017maintaining}, and degree-densest subgraph \cite{charikar} of clique graph (\Cref{tab:meetup}). {\bf (1)} The degree-densest subhypergraph contains casual, frequent gatherings, each having less participants, from only one socializing group (\textit{Bellevue Meetup: Meet new Friends}). {\bf (2)} However, both the degree-densest subgraph of the clique graph and the volume-densest subhypergraph contain events involving multiple technical groups that arrange meetups about diverse technical themes. {\bf (3)} Despite finding events of technical themes, the volume-densest subhypergraph finds a different set of events than the degree-densest subgraph of the clique graph. The reason is that the degree-density criterion used to extract these events from the clique graph is different from volume-density. %and not optimal for finding volume-densest subhypergraph.
The degree-densest subgraph of the clique graph, despite having a high degree-density (100.7), when projected to a subhypergraph (by projecting cliques to hyperedges), has a low volume-density (0.4). 
% As a result, the degree-densest subgraph projected to a subhypergraph (by projecting cliques to hyperedges) has a low volume-density (0.4), despite having a high degree-density (100.7) as a subgraph of the clique graph. 
We also find these technical events to be less popular, having 5 participants on avg., compared to those returned by the volume-densest subhypergraph (avg. 78 participants).

\begin{table}[tb!]
\centering
\vspace{-1mm}
\captionof{table}{\footnotesize A summary of volume-densest and degree-densest subhypergraphs, degree-densest graph of clique graph.: Meetup dataset \cite{meetup}} \label{tab:meetup}
\vspace{-3mm}
\begin{scriptsize}
\resizebox{8.5cm}{!}{
\begin{tabular}{l|lll}
\multicolumn{1}{c|}{} & \multicolumn{1}{c|}{\textbf{Vol.-den. subhyp.}} & \multicolumn{1}{c|}{\textbf{Deg.-den. subg. of clique g.}} & \multicolumn{1}{c}{\textbf{Deg.-den. subhyp.}} \\ \hline
\textbf{\# Events}    & \multicolumn{1}{c|}{\textbf{27}}    & \multicolumn{1}{c|}{17} & \multicolumn{1}{c}{26} \\ \hline
\textbf{Vol.-density} & \multicolumn{1}{c|}{\textbf{116.9}} & \multicolumn{1}{c|}{0.4 (projected to subhyp.)} & \multicolumn{1}{c}{9.3} \\ \hline
\begin{tabular}[c]{@{}l@{}}\textbf{Example} \\ \textbf{events}\end{tabular}  & \multicolumn{1}{l|}{\begin{tabular}[c]{@{}l@{}}Identity and Access \\ Controls Landscape \\ in .NET; Web scraping\\ in Python; Regulatory Env. \\ Around Blockchain\end{tabular}} & \multicolumn{1}{l|}{\begin{tabular}[c]{@{}l@{}}Field Trip w/ \\ Genealogical Society;\\ Monthly Meeting: Civic \\  Tech; Nashville (Nv) \\ PHP User Group\end{tabular}} & \begin{tabular}[c]{@{}l@{}}Trivia Night \textcircled{a} \\ Plantation Pub;\\ Trivia Night \textcircled{a} \\ Three Stones Pub;\\ Dinner \textcircled{a} Dalton\end{tabular} \\ \hline
\begin{tabular}[c]{@{}l@{}}\textbf{Organizing}\\ \textbf{groups}\end{tabular} & \multicolumn{1}{l|}{\begin{tabular}[c]{@{}l@{}}Nv .NET User Group; \\ Data Science Nv.;\\ Nv. Blockchain User Group\end{tabular}}                                                                   & \multicolumn{1}{l|}{\begin{tabular}[c]{@{}l@{}}State \& Local Govt. Dev \\ Network; Dev Launchpad; \\ Nv. PHP User Group\end{tabular}}                              & \begin{tabular}[c]{@{}l@{}}Bellevue Meetup: \\ Meet new Friends\end{tabular}
\end{tabular}
}
\end{scriptsize}
\vspace{-6mm}
\end{table}
%

%\vspace{-1.5mm}
\section{Related Work}
\label{sec:related}
Recently, there has been a growing interest in hypergraph data management \cite{SybrandtSS22,FenderM13,WhangDJLDLKP20,KabiljoKPPSAP17,LeeKS20,shun2020practical}.
%, such as
%join enumeration for hypergraphs \cite{FenderM13},
%clustering \cite{WhangDJLDLKP20,AmburgVB20,KamhouaZ0CLH21}, sampling \cite{ChoeYLBKS22},
%betweenness centrality \cite{LeeK17}, community discovery \cite{BrinkmeierWR07},
%hypergraph partitioning \cite{KabiljoKPPSAP17,SybrandtSS22}, subhypergraph matching \cite{9739081},
%hyperedge prediction~\cite{kijung},
%hypergraph neural networks \cite{0004RZLY22,ZhangWCZWZ020,FengYZJG19}, motifs~\cite{LeeKS20,LeeS21},
%hypergraph null models for the purpose of hypergraph property estimation~\cite{arafat2020construction,chodrow},
%as well as parallel algorithms for computing those properties~\cite{shun2020practical}.
As we focus on core decomposition, we 
%urge readers with a much broader interest to 
refer to
recent surveys~\cite{tina21,giovanni20,Eliassi-RadLRS21,Lee22} for a general exposition.

\spara{Core decomposition in hypergraphs.}
%Although core decomposition of graphs has been studied for decades \cite{Malliaros20}, 
There are relatively few works on
hypergraph core decomposition. Ramadan et al.~\cite{ramadan2004hypergraph} propose a peeling
algorithm to find the maximal-degree based $k$-core of a hypergraph.
%They define $k$-core as the maximal subhypergraph where every vertex has maximal-degree at least $k$. Maximal degree of a vertex is the number of hyperedges incident on that vertex such that none of the hyperedges counted for degree are subsets of each other.
%Jiang et al.~\cite{jiang2017parallel} study parallel peeling process on random $k$-uniform hypergraphs and derive
%lower bound on the number of random peelings required until an empty core is found with high probability.
~\cite{jiang2017parallel,shun2020practical} discuss parallel implementations of degree-based hypergraph core computation
based on peeling approach. Sun et al.~\cite{sun2020fully} propose a fully dynamic approximation
algorithm that maintains approximate degree-based core-numbers of an unweighted hypergraph.
%under insertion and deletion of edges.
Gabert et al.~\cite{GabertPC21}
study degree-based core maintenance in dynamic hypergraphs,
and propose a parallel (shared-memory) local algorithm.
None of these works explore our neighborhood-based
hypergraph core decomposition, which is different from degree-based hypergraph core computation (\S \ref{sec:intro}). The existing approaches for degree-based core decomposition cannot be easily adapted for neighborhood-based hypergraph core computation (\S \ref{sec:algorithm}). 
%, nor do they consider algorithmic approaches other than peeling.

\spara{Core decomposition in graphs.} %Graph core decomposition has been used in
%many applications, including dense subgraph discovery \cite{charikar,AC09}, speeding up community search
%\cite{SG10}, graph clustering \cite{Giatsidis14}, and maximal cliques finding \cite{E10},
%identifying influential spreaders \cite{RePE}, %network visualization \cite{Alvarez-HamelinDBV05},
%chromatic number \cite{M83},
%engagement in social networks \cite{B15}, location-based \cite{KimGFCKC20},
%protein interaction \cite{Altaf03}, and brain networks analysis \cite{HagmannEtAl2008}.
The linear-time peeling algorithm for graph core decomposition \cite{Malliaros20} was given by Batagelj and Zaver{\v{s}}nik \cite{Batagelj11}.
Core decomposition has also been studied in disk-based \cite{ChengKCO11,K15},
distributed \cite{distributedcore,SariyuceSP18}, parallel \cite{D17}, and streaming \cite{S13} settings, and
for varieties of graphs, e.g., weighted \cite{al2017identification,Zhou0HY0021}, directed \cite{LiaoLJHXC22},
temporal \cite{GalimbertiCBBCG21}, uncertain \cite{FGKV14}, and multi-layer \cite{GalimbertiBGL20} networks.
Higher-order cores in a graph (e.g., $(k,h)$-core \cite{bonchi2019distance,LiuZHX21}, triangle $k$-core \cite{ZhangP12},
$h$-clique-core, pattern-core \cite{Tsourakakis15a,FYCLX19}) and more complex cores in an attributed network
(e.g., meta-path-based core \cite{FangYZLC20} and $(k,r)$-core \cite{ZhangZQZL17})
have been proposed. In \S \ref{sec:diff}, we reasoned that existing approaches for graph core decomposition
cannot be easily adapted for neighborhood-based hypergraph core decomposition. We also depicted
that the local approach \cite{distributedcore,eugene15}, an efficient method for graph core decomposition,
produces incorrect core-numbers for neighborhood-based hypergraph core decomposition (\S \ref{sec:local}, \S \ref{sec:experiments}). 
%\vspace{-1mm}
\section{Conclusions}
\label{sec:conclusion}
We introduced neighborhood-cohesive core decomposition of hypergraphs, 
having desirable properties such as \textbf{Uniqueness} and \textbf{Core-containment}. 
We then proposed three algorithms: \textbf{Peel}, \textbf{E-Peel}, and novel \textbf{Local-core} for hypergraph core decomposition. 
%We also adopted four hypergraph-specific optimizations to \textbf{Local-core} %proposing \textbf{Local-core(OPT)} 
%and its shared-memory parallel implementation. 
Empirical evaluation on synthetic and real-world hypergraphs 
showed that the novel \textbf{Local-core} with optimizations and parallel implementation is the most efficient among all 
proposed and baseline algorithms. 
%Our decomposition is useful in diffusion applications and densest subhypergraph extraction. 
Our proposed decomposition is more effective than the degree and clique graph-based decompositions in intervening diffusion. 
%In the densest subhypergraph finding
%application, our decomposition helps extract the volume-densest
%subhypergraph with an approximation guarantee. 
%Our novel volume-
%densest subhypergraphs are more effective in maintaining neighborhood cohesiveness than existing degree-densest subhypergraphs.
Case studies illustrated that our novel volume-densest subhypergraphs capture differently important meetup events, compared to both degree and clique graph decomposition-based densest subhypergraphs. Finally, we developed a new hypergraph-core model,
\emph{(neighborhood, degree)-core} by considering both neighborhood
and degree constraints, designed decomposition algorithm \textbf{Local-core+Peel}, and depicted its superiority in diffusion spread. 
%and finding important hyperedges from \textcolor{red}{protein complexes}.
In future, we shall design efficient algorithms for \emph{(neighborhood, degree)-core} decomposition.
% \begin{acks}
%  This work was supported by the [...] Research Fund of [...] (Number [...]). Additional funding was provided by [...] and [...]. We also thank [...] for contributing [...].
% \end{acks}

%\clearpage

\bibliographystyle{ACM-Reference-Format}
\bibliography{sample}

\end{document}